%% file: globopt.tex
\documentclass{article}
\usepackage{amssymb,amsthm,a4wide,longtable}
\usepackage{tikz}

\def\H{\mathop{\hbox{\bf H}}}
\def\HN{\bar\Gamma^*_N}
\def\I{\mathop{\hbox{$I$}}}
\def\R{\mathbb{R}}

\newtheorem{theorem}{Theorem}
\newenvironment{optproblem}{\trivlist
   \item[\hskip \labelsep{\bfseries Optimization Problem}]\itshape}{\endtrivlist}

\newtheorem{proposition}[theorem]{Proposition}

\title{\bf Using multiobjective optimization 
to map the entropy region of four random variables}
\author{L\'aszl\'o Csirmaz
\thanks{Central European University and
R\'enyi Institute, Budapest}
\thanks{e-mail:~csirmaz@renyi.hu }
\thanks{Supported by TAMOP-4.2.2.C-11/1/KONV-2012-0001 and the Lendulet
program}}
\date{}

\begin{document}
\maketitle

\begin{abstract}

Mapping the structure of the entropy region in higher dimensions is an
important open problem, as even partial knowledge about this region has far
reaching consequences in other areas in mathematics like information theory,
cryptography, probability theory and combinatorics. Presently the only known
method of exploring the entropy region is the one of Zhang and Yeung from
1998. Using some non-trivial properties of the entropy function, their
method is transformed to solving high dimensional linear multiobjective
optimization problems.

Beson's outer approximation algorithm is a fundamental tool for solving such
optimization problems. An improved version of Benson's algorithm is
presented which requires solving one scalar linear program in each iteration
rather than two or three as in previous versions. During the algorithm
design special care was taken for numerical stability. The implemented
algorithm was used to verify previous statements about the entropy region,
as well as to explore it further. Experimental results demonstrate the
viability of the improved Benson's algorithm for determining the extremal
set of medium-sized numerically ill-posed optimization problems. With larger
problem sizes, two limitations of Benson's algorithm were observed: the
inefficiency of the scalar LP solver, and the unexpectedly large number of
intermediate vertices.


\noindent
{\bf AMC numbers:} 90C60, 90C05, 94A17, 90C29

\end{abstract}

\section{Introduction}\label{sec:intro}

Exploring the 15 dimensional entropy region formed by the entropies of the
non-empty subsets of four random variables is an intriguing research
problem. 
The entropy function maps the nonempty subsets of a finite set of random
variables into the Shannon entropies of the marginal distributions. The
range of the entropy function is the {\em entropy region}; it is a subset of
a high-dimensional Euclidean space indexed by the non-empty subsets of the
random variables.
Inequalities that hold for the points of the region are
called {\em information theoretic}.  The entropy region is bounded
by hyperplanes corresponding to the well-known Shannon information
inequalities. 

Presently the only available method which goes beyond the standard
Shannon inequalities is the one of Zhang and Yeung from 1998.
The method starts with a description of ``copy
steps,'' which determine a (usually very) high dimensional linearly
constrained region. The projection of this polytope onto the 15 dimensional
space of the original entropies contains the entropy region, and, quite
frequently, its facets yield new (linear) entropy inequalities. Using some
non-trivial properties of the entropy region, this problem is transformed
into the problem of finding all extremal vertices of a 10-dimensional
projection of a high dimensional polytope -- which problem lies in the realm
of {\em linear multiobjective optimization}.

Benson's outer approximation algorithm 
is a fundamental tool
for solving multiobjective linear optimization problems. Compared to the
original version and its refinements, 
we introduce a modification which leads to a significant
improvement. The improvement is based on the observation that
the scalar LP instance, which is used to decide
whether an objective point is on the boundary of the projection or not, can
also provide a separating facet when the point is outside the facet. In all
earlier published versions of the algorithm, a separate LP instance was used 
to find such a facet.\footnote{A.~H.~Hamel, A.~L\"ohne and B.~Rudloff in 
\cite{hamel-loehne-rudloff} 
have observed the same improvement independently.} This improvement is of
independent interest as it applies to all versions of Benson's algorithm.

The implemented algorithm was used successfully to check earlier results on
the entropy region. It also generated hundreds of new entropy inequalities, and was
essential in formulating a general conjecture about the limits of the
Zhang--Yeung method. The experiments indicated the shortcomings of the
implemented variant of the algorithm, and raised an interesting theoretical
question about the structure of high dimensional polytopes.

\subsection{The entropy region}

The Shannon-inequalities bound the $2^N-1$-dimensional entropy region;
this bounding polytope is known as the {\em Shannon bound}. 
If $N=2$, then the entropy region and the Shannon bound coincide; 
if $N=3$, then the Shannon bound
is the closure of the entropy region, and there are missing points on the
boundary. (In fact, the boundary looks like a fractal, its exact structure
is unknown.) In the case of $N\ge 4$, the Shannon bound strictly exceeds 
the closure of the entropy region \cite{Yeung-course}.

To map the structure of the entropy region for $N\ge4$ is an intriguing open
problem. Even partial knowledge about this region has important consequences
in several
mathematical and engineering disciplines. N.~Pippenger argued in
\cite{pippenger} that linear information inequalities encode the fundamental
laws of Information Theory, which determine the limits of information
transmission and data compression. In communication networks the capacity
region of any multi-source network coding can be expressed in terms of the
entropy region, see the thorough review on network coding in \cite{BMRST}.
Information inequalities have a direct impact on the converse setting with
multiple receivers \cite{Yeung-course}. In cryptography, such inequalities
are used to establish bounds on the complexity of secret sharing schemes
\cite{Beimel-survey}. In probability theory, the implication problem of
conditional independence among subvectors of a random vector can be
rephrased as the investigation of the lower dimensional faces of the entropy
region \cite{studeny, matus-studeny}. Guessing number of games on directed
graphs are related to network coding, where new bounds on the entropy region
provide sharper bounds \cite{baber-etal}. Information theoretic
inequalities surface in additive combinatorics \cite{madiman-etal}, and are
intimately related to Kolmogorov complexity \cite{MMRV}, determinant
inequalities and group-theoretic inequalities \cite{Chan}.

The very first information theoretic inequality which showed that the
entropy region is strictly contained in the Shannon bound was found by Zhang
and Yeung in 1998 \cite{zhang-yeung}. Since then, many other inequalities
have been found based on their idea \cite{dougherty-etal, matus-infinite}. 
Until now
this is the only technique at our disposal: all other proposed techniques
were shown to be equivalent to the Zhang-Yeung method \cite{kaced}.

\subsection{Mapping the entropy region is an optimization problem}

Section \ref{sec:problem} outlines why the technique of Zhang and Yeung 
sketched above is 
equivalent to solving a multiobjective linear optimization problem,
with the main 
focus on describing
the general form of these problems.
Using some non-trivial properties of the entropy region, in the case of $N=4$ 
random variables the {\em objective space} of the optimization can be
reduced to be 10 dimensional.
The exact details of how to generate the optimization problem 
are given in the Appendix. 
Solving these linear optimization problems are especially challenging as
\begin{itemize}
\item[a)] the size of the problem grows exponentially, and becomes 
   prohibitively large very soon;
\item[b)] while the linear constraints form a sparse matrix, there are many
non-trivial linear combinations among them; consequently
\item[c)] the whole system is numerically ill-posed.
\end{itemize}
When $N\ge 5$, the corresponding optimization problems have 
less structure, are two order of magnitude larger,
and even in the simplest case no existing optimization technique seems to be 
able to handle them.

\subsection{Benson's algorithm revisited}

Benson's outer approximation algorithm \cite{benson} is a fundamental tool
for solving linear multiobjective optimization problems. It works in the 
low dimensional
{\em objective} space, which in the $N=4$ case has 10 dimensions, rather in
the
much larger (several hundred dimensional) problem space.  It was a natural
choice to use Benson's algorithm for solving the optimization problem
described in Section \ref{sec:problem}, and in Section \ref{sec:generating}
we describe an improved version of Benson's
algorithm. Material of Section \ref{sec:problem} is of independent
interest as the improvement applies to other variants of Benson's algorithm 
as well. The
original version \cite{benson} and other published variants
\cite{projective, ehrgott, heyde-loehne} use two scalar LP instances in each
iteration step, while this version requires solving a single LP instance of
the same size in each iteration. In typical applications of Benson's
algorithm the computation time outside the LP solver is negligible, thus the
total running time of the algorithm is reduced considerably. The same
improvement can be applied to the ``dual'' optimization problem as defined
in \cite{heyde-loehne}, thus both the primal and the dual problem require
solving the same number of scalar LP instances. As a consequence of this
optimization, the discrepancy between the running times of the primal and
dual variants, as observed in \cite{ehrgott}, vanishes. The same improvement
of Benson's algorithm has been observed independently, and put into a wider
context by Hamel, L\"ohne, and Rudloff in \cite{hamel-loehne-rudloff}.

After sketching the general idea behind the improved version, we prove its
correctness in Section \ref{subsec:correct}. The
termination of the algorithm is immediate from the facts that the extremal
polytope has finitely many facets and finitely many vertices, and that each
iteration generates either a new vertex or a new facet of the extremal
polytope.

The modified algorithm is detailed in Section \ref{subsec:details}. It uses
the {\em double description method} \cite{fukuda-prodon} for vertex
enumeration. 
Section \ref{subsec:LP} discusses the
modifications of the simplex-based LP solver we employed which
improves efficiency and numerical stability.

\subsection{The results}

Section \ref{sec:experimental} describes the experimental results. The
algorithm was run successfully for all 133 copy strings described in
\cite{dougherty-etal}. For each of those strings, all extremal solutions of
the corresponding multiobjective linear optimization problem was generated.
The 10 dimensional extremal solutions corresponded to the ``strongest''
entropy inequalities this copy string could yield.

Copy strings leading to larger optimization problems were also considered.
As the result, the total number of known computer-generated information-theoretic
entropy inequalities grew from 214 in \cite{dougherty-etal} to more than 470.

We also ran the algorithm successfully on a couple of significantly larger
problems, where the symmetry of the problem allowed to reduce the dimension
of the objective space (the dimension of the extremal polytope) from 10 to three.
Results achieved here were essential in formulating a general conjecture
about the limits of the Zhang--Yeung method \cite{book-conjecture}.

Section \ref{sec:conclusion} concludes the paper where we also discuss the
shortcomings of the described variant of Benson's algorithm. The double
description method, which is used to enumerate the intermediate vertices and
facets, seems to be the bottleneck when the extremal polytope has several
thousand vertices. It is an interesting open question how many extra
vertices the intermediate polytopes might have compared to the final number
of vertices and facets. We settle this question
(at least up to constant multipliers)
when the dimension of the
objective space is three, but in
higher dimensions, there is a huge discrepancy between the lower and upper
bounds.

\section{How mapping the entropy region leads to multiobjective
optimization}\label{sec:problem}

The entropy region is a subset of the non-negative orthant of the
$2^N-1$-dimensional Euclidean space,
and its closure (in the usual Euclidean topology) is denoted by $\HN$
\cite{Yeung-course}. It is known that $\HN$ is a convex, closed, pointed
cone, and the entropy region misses only boundary points, see \cite{matus}.

The region $\HN$ is bounded by linear facets corresponding to the so-called
Shannon entropy inequalities. If $N\le3$, $\HN$ is exactly the polytope
determined by these Shannon inequalities, while for $N\ge 4$, it is a proper
subset. Hyperplanes that cut into the Shannon polytope and contain all
entropic points on one side are called {\em non-Shannon (entropy)
inequalities}. The first such inequality was found by Zhang and Yeung in
1998 \cite{zhang-yeung}. Since then, many other inequalities have been
found. For a thorough discussion and new results, see
\cite{dougherty-etal, matus-infinite, kaced}. The Zhang--Yeung method
can be outlined as follows \cite{dougherty-etal, zhang-yeung, kaced}. 
The process starts with four
random variables $\xi_1$, $\xi_2$, $\xi_3$, and $\xi_4$. They are split into
two groups. An independent copy of the first group is created over the
second group, and these new {\em auxiliary} random variables are added to
the pool of existing (random) variables. Due to independence, several 
linear equations hold for their entropies.
This copy step is then repeated as described by the {\em copy string}.
The process
yields an extension of the original set of random variables and a list of
linear dependencies among their entropies. Then
all Shannon inequalities
for the entropies of this extended set are collected, and 
all linear dependencies are added. As
the Shannon inequalities are linear, this results in a large set of
(homogeneous) linear inequalities among the entropies of the subsets of the
original and the auxiliary random variables. Finally, it is checked whether
this set of homogeneous linear inequalities has any (new) consequences on
the fifteen entropies of the original four variables, $\xi_1$, $\xi_2$,
$\xi_3$, and $\xi_4$.

Taking the dual view, consequences of a set of homogeneous linear
inequalities are their non-negative linear combinations. From all the
combinations, we are only interested in those where all coefficients are
zero, except for the 15 non-empty subsets of the original four random
variables. The collection of these latter combinations is a convex, closed
15-dimensional pointed polyhedral cone yielding all required consequences.
The result we would like to get is the {\em extremal rays} of this cone,
that is, the best possible entropy inequalities, which are not superseded by
any linear combination of others. This view leads naturally to {\em
multiobjective optimization} as follows. In our case the {\em problem space}
consists of the (non-negative) combining coefficients (denoted by $x$) of
the Shannon inequalities, that is, the non-negative orthant of $\R^n$ where
$n$ is the number of inequalities. The {\em constraints} are the (linear)
conditions on $x\in\R^n$ expressing the fact that in the combined
inequality, all entropies containing any auxiliary random variable should
vanish. That is,
$$
         A'x=0, ~~~\mbox{where $x\ge 0$,  $x\in\R^n$}
$$
for an $m$ by $n$ matrix $A'$ compiled from the Shannon inequalities
and linear dependencies. The {\em objective space} is $\R^{15}$,
which corresponds to the 15 entropies of the original four random variables. 
The {\em objective} of the optimization is the vector of the coefficients of
these entropies in the resulting combined inequality, thus it can be written
as $P'x$, where $P'$ is a $15\times n$ matrix. The region whose extremal
rays provide the solution for our problem is the pointed polyhedral cone
$$
    \mathcal C' = \{ P'x :\, A'x = 0,~ x\ge 0 \}.
$$
Using information-theoretic considerations as discussed in
\cite{matus-studeny, balanced,
matus-personal, ww}, it is more convenient to look at
$\mathcal C'$ in a different coordinate system.  Let us consider the cone
$$
   \mathcal C = U\mathcal C' = \{ UP'x :\, A'x=0,~ x\ge 0 \},
$$
where $U$ is a $15\times15$ unimodular matrix (for details, please see the
Appendix). There are many advantages of considering $\mathcal C$ instead of
$\mathcal C'$. First, we can set one $U$-coordinate -- the so-called
Ingleton coordinate -- to 1, cutting the pointed cone $\mathcal C$ to a
polytope $\mathcal P$. Rather than searching for the extreme rays in
$\mathcal C$, now we can search for the vertices of $\mathcal P$.

Second, the other 14 $U$-coordinates are all non-negative entropy
expressions. Consequently, the coordinates of $y\in\mathcal P$ in these
directions are necessarily non-negative, that is, $\mathcal P$ lies in the
non-negative orthant of $\R^{14}$. Furthermore, if $y\in\mathcal P$ and
$y'\ge y$ coordinatewise, then $y'\in \mathcal P$. Thus the vertices of
$\mathcal P$ are exactly the {\em extremal} points of $\mathcal P$, where
$y\in\mathcal P$ is extremal if no $y'\le y$, different from $y$, is in
$\mathcal P$.

Third, the polytope $\mathcal P$ is known to be the direct product of a
10-dimensional polytope $\mathcal Q$ and the non-negative orthant of $\R^4$,
see \cite{balanced}. Therefore, the extremal vertices of $\mathcal P$ are
the extremal vertices of $\mathcal Q$ with four zero coordinates added.
One can get the points of $\mathcal Q$ directly by merging these five
additional constraints on $x$ (that the Ingleton coordinate in $UP'x$ should
be $1$, and the four additional coordinates in $UP'x$ should be zero) to the
original constraints $A'x = 0$.

The problem of finding the minimal set of entropy inequalities which
generate (via non-negative linear combinations) all other entropy
inequalities resulting from a given copy string has thus
been transformed into the following multiobjective linear optimization
problem.

\begin{optproblem}
Given the $m\times n$ matrix $A$, the $p\times n$ matrix $P$ with
$p=10$ generated from the copy string,
find the extremal vertices of the $p$-dimensional polytope
\begin{equation}\label{eq:2}
    \mathcal Q = \{ Px : Ax=b, \, x\ge 0 \},
\end{equation}
knowing that $\mathcal Q$ is in the non-negative orthant of $\R^p$, and the
column vector $b$ contains $1$ in the Ingleton row, and zero elsewhere.
\end{optproblem}

\noindent
Indeed, given the linear constraints $Ax=b$, $x\ge 0$ in the $n$-dimensional
problem space $\R^n$, one has to simultaneously minimize the $p$ linear 
objectives given by the matrix $Px$. The solution of the optimization
problem is the complete list of the extremal points of $\mathcal Q$, that
is, those $p$-tuples from $\mathcal Q$ where no coordinate can be decreased
without increasing some other coordinate and still remain in $\mathcal Q$.
This list gives the coefficients of the minimal set of entropy inequalities 
generating all consequences of the copy string.

\section{Improved variant of Benson's outer algorithm}\label{sec:generating}

Benson's algorithm \cite{benson} solves the Optimization Problem defined in
Section \ref{sec:problem} working in the $p$-dimensional objective space,
the range of $Px$, rather than in its much larger domain $\R^n$. The outline
of this algorithm is as follows. It starts from a polytope $S_0$ containing
$\mathcal Q$. In each iteration the algorithm maintains a convex bounding
polytope $S_n$ by listing all of its vertices and all of its facets. If all
vertices of $S_n$ are on the boundary of $\mathcal Q$, then we are done, as
the extremal vertices of $\mathcal Q$ are among the vertices of $S_n$.
Otherwise, we select a vertex of $S_n$ that is not in $\mathcal Q$, connect
it to an internal point of $\mathcal Q$, and find the intersection of this
line with the boundary of $\mathcal Q$. Let the intersection point be $\hat
y_n$. Then, we find a facet of $\mathcal Q$ which is adjacent to $\hat y_n$.
We add this facet to $S_n$ to get $S_{n+1}$, determine the new vertices of
$S_{n+1}$, and iterate the method. The algorithm always terminates after
finitely many iterations.

Several improvements have been suggested to the original algorithm: see,
among others, \cite{projective, ehrgott, heyde-loehne}. The first paper
discusses how $S_0$ is chosen, which may have a heavy impact on the
performance of the algorithm. Other papers suggest improvements related to
the steps where we need to find the new vertices of $S_{n+1}$ and decide
whether a new vertex is on the boundary of $\mathcal Q$.

The improvement of Benson's algorithm this paper describes is achieved by
merging the steps of finding a boundary point and finding the adjacent facet
of $\mathcal Q$. We note that the same improvement was found independently
in \cite{hamel-loehne-rudloff}.

\subsection{Extremal points}
\newcommand\QT{\mathcal Q^{+}}

The transpose of matrix $M$ is denoted by $M^T$. Vectors are
usually denoted by small letters, and are considered single column
matrices. For two vectors $x$ and $y$ of the same dimension, $xy$ denotes
their inner product, which is the same as the matrix product $x^Ty$.

\medskip

Recall that $A$ is an $m\times n$ matrix mapping $\R^n$, the {\em problem
space}, to $\R^m$, and $P$ is a $p\times n$ matrix mapping the problem space
into $\mathbb R^p$, the {\em objective space}. (In our case $n$ is the
number of Shannon inequalities generated from the copy string, and $p$ is
10, see Section \ref{sec:problem}.) Let $\mathcal A$ be the polytope
$$
    \mathcal A = \{ x\in \mathbb R^n: Ax=b,\, x\ge 0\,\}
$$
for the fixed vector $b$. Then
$$
    \mathcal Q = \{ Px : \, x \in\mathcal A \} .
$$
For better clarity, $y$ will denote points of the objective space $\R^p$,
while points 
in the problem space will be denoted by $x$. The point $y\in\mathcal Q$ is
{\em extremal} if no other point in $\mathcal Q$ ``supersedes'' it, that is,
from $y'\le y$, $y'\in\mathcal Q$ it follows that $y'=y$. The point $y$ is
{\em weakly extremal} if there is no $y'<y$ in $\mathcal Q$ (i.e. when all
coordinates of $y'$ are smaller than the corresponding coordinate in $y$).

We want to generate the set of extremal vertices of $\mathcal Q$, but instead of
$\mathcal Q$, we consider another polytope, $\QT$, which has the same set of
extremal points, and is easier to handle \cite{benson}. This polytope is
defined as
\begin{equation}\label{eq:qtdef}
  \QT = \{ y \in \mathbb R^p : y\ge y'\,  \mbox{ for some $y'\in\mathcal Q$}
       \,\}.
\end{equation}
In fact, $\QT$ is the Minkowski sum of $\mathcal Q$ and the non-negative
orthant of $\mathbb R^p$, thus it is a convex closed polytope. The following
facts can be found, e.g., in \cite[Proposition 4.3]{ehrgott}.

\begin{proposition}\label{claim:basicQT}
a) The extremal points of $\QT$	 and $\mathcal Q$ are the same.

b) The weakly extremal points of $\QT$ are exactly its boundary points. \qed
\end{proposition}

\subsection{Finding a boundary point and a supporting
hyperplane at the same time}\label{subsec:correct}

The crucial step in Benson's algorithm is to find a boundary point of $\QT$,
and then to find the supporting hyperplane at that point. In the original
version, it is achieved by solving two appropriately chosen scalar LP
problems \cite{benson, ehrgott, heyde-loehne}. In our version, we need to
solve only one scalar LP instance to do both jobs.

To describe the procedure, let $q\in\mathbb R^p$ be an internal point of
$\QT$, and let $d\in\mathbb R^p$ be some direction in the $p$-dimensional
space. Consider the ray starting at $q$ with the direction $-d$, that is,
points of the form $q-\lambda d\in\mathbb R^p$ where $\lambda\ge 0$ is a
non-negative real number. If not the whole ray is in $\QT$, then there is a
$\hat\lambda >0$ so that $q-\lambda d\in\QT$ if and only if
$\lambda\le\hat\lambda$. By the next proposition, this threshold
$\hat\lambda$ can be found by solving a scalar LP problem.

\begin{proposition}\label{claim:point+support}
Suppose $q\in\QT$ is an internal point, and not the whole ray $\{ q-\lambda d
:\,\lambda\ge 0\}$
is in $\QT$. Let $\hat\lambda$ be the solution of the LP problem
$$
\max\nolimits_{\lambda,x} \{\,\lambda:\, q-\lambda d \ge Px,\,Ax=b,\,x\ge
0\,\} .
\eqno\hbox{$\mathrm{P}(q,d)$}
$$
Then $\hat y = q-\hat\lambda d$ is a boundary point of $\QT$. Let moreover
$(\hat u,\hat v)$ be a place where the dual problem
$$
\min\nolimits_{u,v} \{\,
 bu+qv :\, A^Tu+P^Tv\ge 0,\, dv=1,\, v\ge 0\,\}
\eqno\hbox{$\mathrm{D}(q,d)$}
$$
takes the same extremal value. Then $\{y\in\mathbb R^p: y\hat v = \hat
y\hat v\}$ is a supporting hyperplane to $\QT$ at $\hat y$.
\end{proposition}
\begin{proof}
The first part of the proposition is clear: $q-\lambda d$ is in $\QT$ if it
is ${}\ge Px$ for some $x\ge 0$ with $Ax=b$. The LP problem
$\mathrm{P}(q,d)$ simply searches for the largest $\lambda$ with this
property. From this, it also follows that $\hat y$ is a boundary point of
$\QT$: there are points arbitrarily close to $\hat y$ which are not in
$\QT$.

The dual problem $\mathrm{D}(q,d)$ has the same optimal value as the primal
one, thus
\begin{equation}\label{eq:claim-dual}
\min\nolimits_{u,v} \{\,
 bu+qv :\, A^Tu+P^Tv\ge 0,\, dv=1,\, v\ge 0\,\} = \hat\lambda.
\end{equation}
Fixing the $\hat v$ part of the solution, the minimum is still $\hat\lambda$
as $u$ runs over its domain $\mathbb R^m$, and therefore
$$
    \max\nolimits_{u} \{ -bu :\, A^T(-u) \le P^T\hat v\, \} = q \hat v -
\hat\lambda.
$$
Consequently its dual also has the same optimal value:
$$
    \min\nolimits_{x}\{ x(P^T\hat v) : \, Ax=b, \, x\ge 0 \} = q \hat v -
\hat\lambda.
$$
Now $\hat v\ge 0$ by (\ref{eq:claim-dual}), and an arbitrary point $y$ of
$\QT$ can be written as $z+Px$ where $z\ge 0$, $z\in\mathbb R^p$ and $x\ge
0$, $Ax=b$, $x\in\mathbb R^n$, and so
$$
  (z+Px)^T\hat v = z\hat v + x(P^T\hat v) \ge 0 + q\hat v - \hat \lambda.
$$
On the other hand, $d\hat v =1$ by (\ref{eq:claim-dual}) again, thus
$$
   q\hat v - \hat\lambda = q\hat v -\hat\lambda(d\hat v) =
       (q - \hat\lambda d)\hat v = \hat y \hat v.
$$
This means that for any point $y$ of $\QT$, we have $y\hat v \ge \hat y\hat
v$, furthermore $\hat y$ is in $\QT$. Thus $\{y\in\mathbb R^p:
y\hat v = \hat y\hat v\}$ is a supporting hyperplane to $\QT$, as was
claimed.
\end{proof}

We remark that if $\{y\in\mathbb R^p: y w = \hat y w\}$ is a supporting
hyperplane to $\QT$ at $\hat y\in\QT$, then, necessarily, $w\ge 0$.  Indeed,
if $w_i$ were negative, then letting $e_i$ be the unit vector with 1 at the
$i$'th coordinate, $\hat y + e_i \in\QT$, and $(\hat y+e_i)w = \hat y w +
w_i < \hat y w$, which is a contradiction.

\subsection{Initial bounding polytope}\label{subsec:initial-polytope}

Following the ideas of Burton and Ozlen in \cite{projective}, we extend the
objective space by {\em positive ideal} elements. As they showed, this
extension does not restrict the applicability of the algorithm, but
simplifies the intermediate polytopes.

First of all, we know that points of $\mathcal Q$ have non-negative
coordinates. Thus $\QT$ is included in the non-negative orthant of $\mathbb
R^p$, that is, it is part of the ``ideal'' $p$-dimensional simplex with the
the origin and the $p$ ``positive endpoints'' of the coordinate axes of
$\mathbb R^p$ as its vertices.

During the algorithm we will be dealing with two kinds of ``objective''
points: ordinary (finite) points $y=\langle y_1,\dots,y_p\rangle$ lying in
the non-negative orthant (that is, $y_i\ge 0$), and {\em ideal points} at
the positive endpoints of the rays $m=\langle m_1,\dots, m_p\rangle$ in the
non-negative orthant (that is, $m_i\ge 0$ for all $i$). Using homogeneous
coordinates as suggested in \cite{projective} we can handle both types
uniformly: ordinary points have coordinates $(y,1)=\langle
y_1,\dots,y_p,1\rangle$, while ideal points can be conveniently written as
$(m,0)$, indicating that rays are invariant under multiplication by a
(positive) scalar. A {\em hyperplane} is a $p+1$-tuple $(w,d)=\langle
w_1,\dots,w_p,d\rangle$. A point $(z,j)$ is on a hyperplane if
$(z,j)^T\cdot(w,d)=\allowbreak zw + jd= 0$, and it is on its {\em
non-negative side} if $(z,j)^T\cdot(w,d)\ge 0$. All ideal points are
on the non-negative side of the hyperplane $(w,d)$ if and only if $w_i\ge 0$
for each $i$. The line connecting the points $(y,1)$ and $(m,0)$ intersects
the hyperplane $(w,d)$ in the (ordinary) point $(y+\lambda m,1)$, where the
scalar $\lambda$ is determined by the condition
$$
  (y+\lambda m, 1 )^T\cdot(w,d) = y^Tw + \lambda m^Tw + d = 0 ,
$$
where $m^Tw \not=0$, as otherwise $(m,0)$ is on the hyperplane.

\bigskip

Benson's algorithm starts with an {\em internal point $q\in\QT$}, and an
{\em initial bounding polytope $S_0 \supseteq\QT$}. In the course of the
algorithm, $q$ will be connected to the vertices of the polytope $S_n$, and
these lines will ideally intersect the boundary of $\QT$ in the relative
interior of some $p-1$-dimensional facet. If this happens, then the
supporting hyperplane to $\QT$ at the intersection point is uniquely
determined, and is, in fact, a facet of $\QT$. This preferable event occurs
if the internal point $q$ is in {\em general position}, that is, not in any
hyperplane determined by any $p$ points among the union of the vertices of
$\QT$ and $S_0$, $\dots$, $S_n$. As points of the objective space not in
general position have measure zero, choosing $q$ randomly from a set with
positive measure will ensure that the above event will happen with
probability 1. We will discuss the choice of $q$ in detail later.

The choice of the initial surrounding polytope $S_0$ is quite natural. It is
the $p$-dimensional ideal simplex with vertices $(0,\dots,0,1)$ (the
origin), and the ideal points $(e_i,0)$, where the coordinates of the ray
$e_i$ are zero except for the $i$-th coordinate. $p$ facets of this simplex
are the coordinate hyperplanes with the equation $(e_i,0)$. The $p+1$-st
facet is the {\em ideal plane} containing all ideal points. As $\QT$ is not
empty, this facet is, in fact, part of $\QT$, thus the ideal points
$(e_i,0)$ are {\em boundary points of $\QT$}. Moreover, it follows from the
remark after the proof of Proposition \ref{claim:basicQT} that ideal points
are on the non-negative side of any supporting hyperplane of $\QT$.

\subsection{Details of Benson's algorithm}\label{subsec:details}

Benson's algorithm constructs a sequence $S_n$ of bounding polytopes. We
discussed in Section \ref{subsec:initial-polytope} how to select the initial
polytope $S_0$. We also have an internal point $q\in\QT$ {\em in general
position}; we will return later to how it can be generated. With each
polytope $S_n$ we also maintain two sets: a set of hyperplanes such that the
intersections of the non-negative sides of these hyperplanes is exactly
$S_n$, and the list of vertices of $S_n$, indicating whether each vertex is
known to be a boundary point of $\QT$ or not.

Suppose we have obtained $S_n$, $n\ge 0$, and we proceed to generate
$S_{n+1}$.
\begin{enumerate}
\item\label{step1}
Let us look at the vertices of $S_n$, and select one which is not marked as a
boundary point of $\QT$. If no such vertex can be found, then the algorithm
terminates: according to Proposition \ref{claim:basicQT}, these vertices are the 
extremal vertices of $\QT$, thus the extremal 
vertices of $\mathcal Q$.

\item\label{step2}
Let $y$ be the vertex selected. Then, the boundary point of $\QT$ needs to
be found on the line segment $y$---$q$ by solving the scalar LP problem
$\mathrm{P}(q,q-y)$. The solution is denoted by $\hat\lambda$. According to
Proposition \ref{claim:point+support}, if $\hat\lambda=1$, then $q-(q-y)=y$
is on the boundary of $\QT$. If so, it is marked as such, and return
to step \ref{step1}.

\item\label{step3}
Otherwise, $0<\hat\lambda < 1$. Let us compute $\hat y = q -
\hat\lambda(y-q)$, which point is on the boundary of $\QT$. By the second
part of Proposition \ref{claim:point+support}, if $(\hat u,\hat v)$ is the
solution to the dual problem $\mathrm{D}(q,q-y)$, then the hyperplane
$h=(\hat v,-\hat y^T\hat v)$ written in homogeneous coordinates is a
supporting hyperplane to $\QT$ at $\hat y$, and $\QT$ is on its non-negative
side. Also, as $q$ is in general position, $h$ is a facet of $\QT$. We can
therefore add $h$ to the hyperplanes of $S_n$ to create the polytope
$S_{n+1}$.

\item\label{step4} 
Next, the vertices of $S_{n+1}$ need to be computed using the double
description method, see \cite{fukuda-prodon} as follows. Vertices of $S_n$ on the
non-negative side of $h$ remain vertices of $S_{n+1}$. All other vertices of
$S_{n+1}$ are the intersection points of the relative interiors of some
edges (two dimensional face) of $S_n$ and $h$.
\end{enumerate}
Please note that in Step \ref{step4}, all considered edges of $S_n$
intersect $h$ at different vertices of $S_{n+1}$, so there is no need to
check for equality between them. This leaves us to determine whether a pair
of vertices $(v_1,v_2)$ of $S_n$ is an edge of $S_n$.  To this end we use
the following observation from \cite{projective}.
\begin{proposition}\label{claim:edgetest}
Vertices $v_1$ and $v_2$ of $S_n$ are connected by an edge if and only if
every other vertex is missed by some facet of $S_n$ containing both $v_1$
and $v_2$. \qed
\end{proposition}
Proposition \ref{claim:edgetest} suggests the following fast combinatorial
test: if one considers all facets of $S_n$ adjacent to both $v_1$ and $v_2$,
and takes the intersection of the adjacency lists of these faces, then if
this intersection contains $v_1$ and $v_2$ only, they form an edge of $S_n$,
otherwise, they do not. To perform this test we also need to maintain the
adjacency lists of vertices and facets:
\begin{enumerate}\setcounter{enumi}{4}
\item\label{step5} Adjust the adjacency lists of vertices and facets of 
$S_{n+1}$.
\end{enumerate}

While the algorithm works with a mixture of ordinary and ideal points,
fortunately, the initial ideal points are on the boundary of the polytope
$\QT$, and because of this, the algorithm introduces no other ideal points.
Ideal points may occur in steps \ref{step4} and \ref{step5} only, when
calculating the intersection of the edge $v_1$ --- $v_2$ of $S_n$ with $h$;
here $v_1$ or $v_2$, but not both may be one of the ideal points. Also, some
of the ideal vertices might be adjacent to the new facet $h$, thus they may
appear on the adjacency list of $h$ (and {\em vice versa}).

\bigskip

The easiest way of finding a {\em random internal point $q\in\QT$} is to get
any feasible $x\in\mathcal A$ (that is, $x\ge 0$, $Ax=b$), and then
increasing all coordinates of $Px$ by some positive random amount. To find
such a feasible $x$ requires solving an LP problem, but we can improve the
algorithm further
if we postpone finding $q$ until it is actually needed, which is when
we determine the first new facet of $S_1$ in Step \ref{step2}. As the only
vertex of
$S_0$ that is not a boundary point of $\QT$ is the origin, we can direct a
ray from the origin in a random direction $d>0$, find the intersection point
of the ray and $\QT$ by solving the LP problem $\mathrm{P}(0,-d)$, and then
set $q$ on that ray at some random distance behind the intersection point.
This way, we not only get the internal point $q$, but also the supporting
hyperplane at the intersection of the segment $0$---$q$ and $\QT$.

If $\QT$ is not empty, then any ray $0+\lambda d$ with $d>0$ cuts into it.
Therefore, for example, we can choose $d$ so that every coordinate is a
uniform random value between 1 and 2. If the LP problem $\mathrm{P}(0,-d)$
has no solution, then $\QT$ is empty; otherwise, let $\hat\lambda>0$ be the
solution, $r$ be a uniform random value between 1 and 2, and we can set
$q=(r+\hat\lambda)d$, having $\hat y = \hat\lambda d$ as the point at the
boundary of $\QT$.

\subsection{The scalar LP problems}\label{subsec:LP}

The scalar LP problems which are to be solved repeatedly during the
algorithm have very similar structures. Figure
\ref{table:P-problem} shows their general form.
\begin{figure}[h!tb]
\begin{center}\newcommand\rl{\rule[-1.5em]{0pt}{4.0em}}%
\newcommand\rll{\rule[-0.8em]{0pt}{2.4em}}%
\newcommand\st{\rule[-0.2ex]{0pt}{2.1ex}}%
\makeatletter\def\framenotop#1#2{%
  \leavevmode
  \@begin@tempboxa\hbox{#2}%
   \setlength\@tempdima{#1}%
   \setbox\@tempboxa\hb@xt@\@tempdima
    {\kern\fboxsep\csname bm@c\endcsname\kern\fboxsep}%
   \@framenotop%
  \@end@tempboxa}%
\def\@framenotop{%
  \@tempdima\fboxrule
  \advance\@tempdima\fboxsep
  \advance\@tempdima\dp\@tempboxa
  \hbox{%
    \lower\@tempdima\hbox{%
      \vbox{\hbox{\vrule\@width\fboxrule\kern-\fboxrule\vbox{\vskip\fboxsep
       \box\@tempboxa\vskip\fboxsep}\kern-\fboxrule\vrule\@width\fboxrule}%
      \hrule\@height\fboxrule}}}%
}\makeatother%
\begin{tabular}{rc@{~}ccc}
&    $x\ge0$      &  $\lambda$ & {} & $b$ \\
\raisebox{-1.6ex}[0pt][0pt]{$u
$}&\framebox[10em]{\rl$A$} & \framebox[1.5em]{\rl$0$} &
\raisebox{-1.6ex}{$=$}
  & \framebox[1em]{\rl$0$}\\
&\framenotop{10em}{\st Ingleton} &\framenotop{1.5em}{\st $0$} &
 {} &\framenotop{1em}{\st$1$} \\\\[-8pt]
$v
$&\framebox[10em]{\rll$P$} & \framebox[1.5em]{\rll$d$} &
 $\le$ & \framebox[1em]{\rll$q$}
\\\\[-8pt]
max:& \framebox[10em]{\st$0$} & \st 1 & &
\end{tabular}\end{center}
\caption{The structure of the $\mathrm{P}(q,d)$ problem}\label{table:P-problem}
\end{figure}%
We indicated that there is only one row in the matrix $A$ where $b$ is not
zero. The internal point $q$ is fixed throughout the algorithm; only the
direction $d$ changes in each iteration. The uniformity of the LP problems
can be used to speed up the initializations required by the LP solver.

\subsection{Tweaks and modifications}\label{sec:tweaks}

Due to the nature of the original combinatorial problem, the optimization
problem is ill-posed, and special care needs to be taken to maintain
numerical stability. We dismissed using integer arithmetic as being
prohibitively expensive both in the LP solver and during computing the
vertices of the approximating polytopes. Instead, the LP solver and vertex
computation were carefully modified to achieve better numerical stability.
 
The applied scalar LP solver is a standard but finely tuned simplex method.
During the execution of the simplex method, one walks through the vertices
of a large dimensional polytope while maintaining different descriptions of
the polytope. Each vertex during the walk is determined by a set of the
original facets (equations) whose intersection the vertex is. In
simplex terminology, this set of facets is known as a {\em base}. Knowing
the base is sufficient to regenerate the internal state at any step directly
from the original facet equations. Therefore, after a predetermined number
of steps we do two things: we save the actual base, and recalculate the
description of the polytope. When we discover any problem caused by
accumulating numerical errors, we return to the last saved base, and
continue the method from that point on. We decided to save the base after
each 60 steps, which seemed to be a good choice for the size of the problems
to be solved.

As described in Section \ref{subsec:details}, Benson's algorithm advances by
computing the vertices of the new approximating polytope $S_{n+1}$ from the
vertices (and facets) of $S_n$, and from a new facet of $\QT$. The input for
this computation comes from the LP solver which gives the equation of the
facet. In our case, all facets had rational coefficients with small
denominators, and so they could be represented either exactly or with extremely
small error.  When computing the new vertices of $S_{n+1}$, we lose some of
this exactness as the new vertices are computed from the intersections of an
old edge and a new facet, and the angle between the edge and the facet can be
very small. When this happens repeatedly, the coordinates of a vertex can
be very far from their exact values. Therefore, rather than carrying these
errors forward, we decided to compute the coordinates of each new vertex
directly from the facets they are incident to.

Last, but not least, the modified Benson's algorithm requires a solution
$\hat v$ of the dual LP problem $\mathrm{D}(q,d)$ only when the optimum
$\hat\lambda$ is not one. This means that we can abort the LP solver as soon
as it finds a feasible solution with $\lambda=1$, further improving the
speed of the algorithm.

\section{Experimental results}\label{sec:experimental}

The modified Benson's algorithm as described in the previous section has
been successfully used to investigate the entropy region $\HN$ by generating
hundreds of new entropy inequalities. We show some of the results for three
different data sets. In the first two of the sets, the objective space had 
10 dimensions,
while in the third one, this dimension was reduced to 3 using the internal
symmetry of the original problem.

The algorithm was coded in C with an embedded scalar LP solver,
a school-book simplex method with the tweaks discussed in
Section \ref{sec:tweaks}. The compiled code was run on a dedicated personal
computer with a 1GHz AMD Athlon 64 X2 Dual Core processor and 4 gigabytes
of RAM. In fact, the program used only a single core (thus two instances
could be run simultaneously on the dual-core machine without affecting the
running time), and the combined memory usage was never above 2 gigabytes.

\subsection{Results for $p=10$}

The paper of Dougherty et al \cite{dougherty-etal} lists all extremal
non-Shannon entropy inequalities which are consequences of at most three
copy steps using no more than four auxiliary variables.  They used 133
different copy strings to determine 214 new entropy inequalities.
The modified Benson algorithm, as outlined in Section \ref{sec:generating},
was used to generate all extremal vertices of the polyhedrons determined by 
these copy strings. The results confirmed their report, and did not find any 
new
entropy inequalities which were not consequences of the ones on their list.

Table \ref{table:results} gives some representative data. {\em
Size} is the size of the matrix $A$ (columns and rows), followed by
\begin{table}[htb]
\newcommand\st[1]{{\small\tt#1}}
\begin{center}\begin{tabular}{lr@{$\times$}lrrr}
\multicolumn{1}{c}{\small Copy string} &\multicolumn{2}{c}{\small Size} & \small Vertices &
\small Facets & \small Time  \\
\st{r=c:ab;s=r:ac;t=r:ad}&561&80&5&20&0:01 \\ 
\st{rs=cd:ab;t=r:ad;u=s:adt}&1509&172&40&132&6:19\\ 
\st{rs=cd:ab;t=a:bcs;u=(cs):abrt}&1569&178&47&76&6:51\\
\st{rs=cd:ab;t=a:bcs;u=b:adst}&1512&178&177&261&17:40 \\ 
\st{rs=cd:ab;t=a:bcs;u=t:acr}&1532&178&85&134&18:27\\ 
\st{rs=cd:ab;t=(cr):ab;u=t:acs}&1522&172&181&245&22:58\\ 
\st{r=c:ab;st=cd:abr;u=a:bcrt}&1346&161&209&436&29:18\\
\st{rs=cd:ab;t=a:bcs;u=c:abrst}&1369&166&355&591&38:59 \\ 
\st{rs=cd:ab;t=a:bcs;u=c:abrt}&1511&178&363&599&1:04:32\\ 
\st{rs=cd:ab;t=a:bcs;u=s:abcdt}&1369&166&355&591&1:07:01\\
\st{rs=cd:ab;t=a:bcs;u=(at):bcs}&1555&177&484&676&1:39:30\\
\st{rs=cd:ab;t=a:bcs;u=a:bcst}&1509&177&880&1238&4:30:26\\
\st{rs=cd:ab;t=a:bcs;u=a:bdrt}&1513&177&2506&2708&5:11:25\\
\end{tabular}
\end{center}
\caption{Representative results for the Dougherty et al list 
   ($p=10$)}\label{table:results}
\end{table}
the number of vertices and facets. {\em Time} is the running time of the
algorithm in hours, minutes and seconds. The running time is included here
to indicate how it changes with the size of the problem, and does not
indicate any comparison with other implementations of Benson's algorithm. The
typical number of the extremal vertices is around a couple of hundreds, with
the
only exception shown in the last row of Table \ref{table:results}, where
the number of extremal vertices is 2506. Also, it might be worth mentioning
that in this case, one of the intermediate polytopes had more than 22,000
vertices, which is an about 10-fold increase.

As the size of the matrix $A$ does not vary too much from case to case, we
expected the running time of the algorithm to depend only on the number of
iterations, that is, on the sum of the number of vertices and facets.
\begin{figure}[htb]
\begin{center}%
\def\p#1,#2;{\put(#1,#2){\makebox(0,0){$\circ$}}}
\setlength\unitlength{0.008\textwidth}
\begin{picture}(70,65)(0,-5)
\put(-5,0){\vector(1,0){75}} \put(0,-5){\vector(0,1){65}}
\put(79,0){\makebox(0,0){\small Iterations}}
\put(20,-1){\line(0,1)2}\put(20,-4){\makebox(0,0){\small500}}
\put(40,-1){\line(0,1)2}\put(40,-4){\makebox(0,0){\small1000}}
\put(60,-1){\line(0,1)2}\put(60,-4){\makebox(0,0){\small1500}}
\put(-14,50){\makebox(0,0){\small Time}}
\put(-14,46.5){\makebox(0,0){\small (minutes)}}
\put(-1,6){\line(1,0)2}\put(-5,6){\makebox(0,0){\small 10}}
\put(-1,18){\line(1,0)2}\put(-5,18){\makebox(0,0){\small 30}}
\put(-1,36){\line(1,0)2}\put(-5,36){\makebox(0,0){\small 60}}
\put(-1,54){\line(1,0)2}\put(-5,54){\makebox(0,0){\small 90}}
\input plotdata
\end{picture}
\end{center}
\kern -10pt
\caption{Running time vs.~problem size}\label{fig:running-time}
\end{figure}
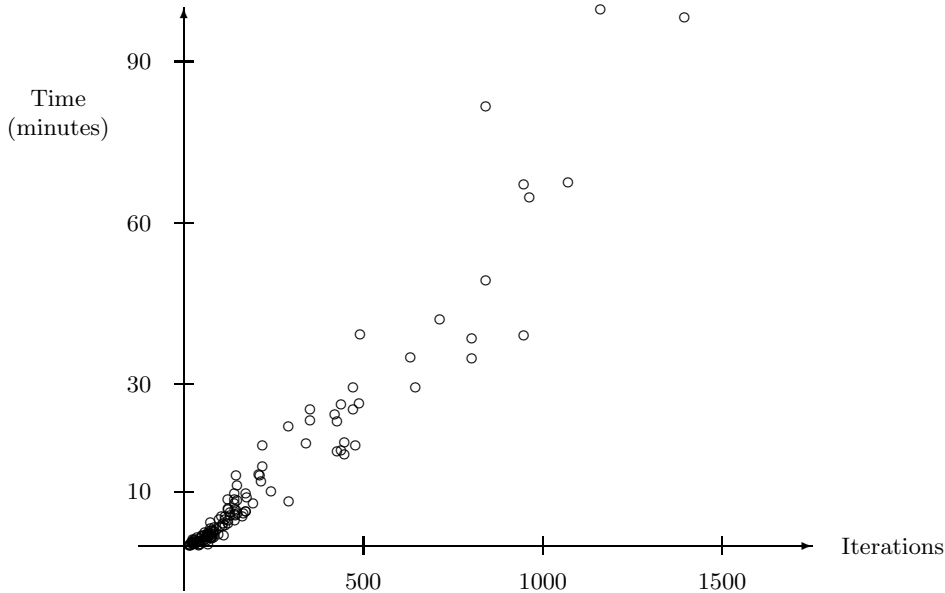
The plot on Figure \ref{fig:running-time} confirms this expectation,
indicating a linear dependence. The variance is due to the fact that
occasionally, the LP solver failed and had to resume the work from an earlier
stage, as it was discussed in Section \ref{sec:tweaks}.

There appears to be no easy way to predict the number of iterations, and
thus the expected running time. Similar copy strings require a widely
varying number of iterations, and have a varying number of extremal vertices.
Let us also mention that in each of the 133 cases, the polytope $\QT$
factors to an 8-dimensional polytope and the non-negative quadrant of
$\R^2$. The fact that the objective space is practically 8 dimensional
rather than 10 might have an impact on the observed speed of the algorithm.

\smallskip

The implemented algorithm was used to check consequences of copy strings 
beyond the ones considered in \cite{dougherty-etal}. This resulted in
increasing the
total number of known computer-generated entropy inequalities from 214
in \cite{dougherty-etal} to more than 470.
Some representatives are shown in Table \ref{table:new-copy-results}.
\begin{table}[htb]
\newcommand\st[1]{{\small\tt#1}}
\begin{center}\begin{tabular}{lr@{$\times$}lrrr}
\multicolumn{1}{c}{\small Copy string} &\multicolumn{2}{c}{\small Size} & \small Vertices &
\small Facets & \small Time  \\
\st{rs=cd:ab;tu=cr:ab;v=(cs):abtu}&4055&370&19&58&1:10:10 \\ 
\st{rs=ad:bc;tu=ar:bc;v=r:abst}&4009&370&40&103&3:24:37\\ 
\st{rs=cd:ab;t=(cr):ab;uv=cs:abt}&3891&358&30&102&3:34:31\\ 
\st{rs=cd:ab;tu=cr:ab;v=t:adr}&3963&362&167&235&9:20:19\\
\st{rs=cd:ab;tu=dr:ab;v=b:adsu}&4007&370&318&356&13:20:08 \\ 
\st{rs=cd:ab;tv=dr:ab;u=a:bcrt}&4007&370&318&356&14:34:42\\
\st{rs=cd:ab;tu=cs:ab;v=a:bcrt}&4007&370&297&648&22:02:39\\ 
\st{rs=cd:ab;t=a:bcs;uv=bt:acr}&3913&362&779&1269&37:15:33 \\ 
\st{rs=cd:ab;tu=cr:ab;v=a:bcstu}&3987&362&4510&7966&427:43:30 \\ 
\st{rs=cd:ab;tu=cs:ab;v=a:bcrtu}&3893&362&10387&13397&716:36:32\\ 
\end{tabular}
\end{center}
\caption{Some extended copy strings with $p=10$}\label{table:new-copy-results}
\end{table}
As the copy string becomes more involved, the generated problem becomes
larger, and we have also observed a significant increase in the number of vertices of
the intermediate polytopes. This unexpected increase makes
the double description method highly inefficient, and it actually becomes the
bottleneck in the algorithm. In most of the cases enlisted in Table 
\ref{table:new-copy-results},  we had to set the bound
on the number of facets and vertices
to $2^{16}$, as smaller values were exhausted very fast.
This in turn meant that the
size of the incidence matrix of vertices and facets was $2^{16}\times 2^{16}$, and
that executing the combinatorial test of Proposition
\ref{claim:edgetest} for each pair $v_1$ and $v_2$ of several thousands of 
vertices took a considerable amount of time, 
and became comparable to the time taken by the LP solver.
In about 20\% of the extended cases investigated, even this limit was 
too small, and the algorithm aborted.

\subsection{Results for $p=3$}
Table \ref{table:bigresults} shows statistics for a set of
even larger problems. In this case, the optimization
problem was determined by the
total number of random variables employed during the copy steps.
Relying on the
high symmetry in the chosen problems and with some careful preprocessing,
we were able to
\begin{table}[htb]
\begin{center}\begin{tabular}{cr@{$\times$}lrrr}
\multicolumn{1}{c}{\small Random variables} &\multicolumn{2}{c}{\small Size} & \small Vertices &
\small Facets & \small Running time  \\
10&692&99&11&13&0:01 \\ 
12&1298&150&26&24&0:55\\ 
14&2175&233&53&43&9:38\\
16&3373&338&100&78&2:18:45 \\ 
18&4942&474&171&129&6:55:40\\ 
20&5772&635&278&208&23:20:17\\ 
\end{tabular}
\end{center}
\caption{More random variables with $p=3$}\label{table:bigresults}
\end{table}
reduce the problem size significantly: 20 random variables had more than
$10^6$ entropies, and more than $4.8\cdot10^7$ Shannon inequalities. The
symmetry in the problems also made it possible to reduce the dimensions of
the objective space from $10$ to $3$.

Problems listed in this section had their running time determined almost
exclusively by the LP solver. 
Solving the numerous numerically ill-posed scalar LP problems of
such sizes appears to be at the limit of the simplex-based LP solver
implemented in our software. It would be an interesting research problem to
investigate which LP solvers can handle problems of this size efficiently
and with sufficient accuracy.

\section{Conclusion}\label{sec:conclusion}

Exploring the 15-dimensional entropy region formed by the entropies of the
non-empty subsets of four random variables is an intriguing research
problem. Presently the only available method which goes beyond the standard
Shannon inequalities is the one of Zhang and Yeung from 1998
\cite{zhang-yeung}. Using non-trivial properties of the entropy region, the
Zhang--Yeung method was transformed into the problem of finding all extremal
vertices of a 10-dimensional projection of a high dimensional polytope.
A modified variant of Benson's algorithm was used to solve this linear
multiobjective optimization problem. As opposed to the original algorithm
\cite{benson} and its refinements \cite{projective, ehrgott, heyde-loehne},
where each iteration requires solving two scalar LP problems, this variant 
uses only one scalar LP instance in each iteration. This improvement,
applicable to other variants of Benson's algorithm as well, might be of independent
interest, and was observed independently in \cite{hamel-loehne-rudloff}.

The implemented algorithm was used successfully to check the results on the
133 copy strings of \cite{dougherty-etal},
and confirm that no entropy inequality has been missed. 
Copy strings leading to larger optimization problems were also investigated. 
This resulted in more than doubling
the number of known computer-generated entropy inequalities.

The algorithm was also run on some extremely symmetrical, but significantly
larger data sets, where the objective space could be reduced to three dimensions.
The results achieved were essential in forming a general 
conjecture about the limits of the Zhang--Yeung method \cite{book-conjecture}.

A compelling theoretical problem arose as the experimental results were
evaluated. It has been observed that some 
of the intermediate 
polytopes $S_n$ (bounded by a certain subset of the facets of $\QT$) 
had 
significantly more vertices than $\QT$ itself.
It is worth noting that this phenomenon cannot occur
when the objective space has at most three dimensions. By the
well-known Euler formula, the number of vertices of any convex 3-dimensional
polytope is bounded by $2f-4$, where $f$ is the number of its faces. Thus any
intermediate polytope can have at most $2M$ vertices, where $M$ is a bound
on the total number of steps the algorithm might take. In the 3-dimensional
case, setting this limit in the software to
a couple of thousand appears to be a safe choice. However, the situation changes
when the dimension $p$ of the objective space increases. Let us 
suppose that $M$ is an 
upper bound on
the number of final vertices and facets. Then an obvious upper bound on the number
of vertices of any intermediate polytope
is $M^p$, as any vertex is determined by the $p$ facets it is
incident to. There are $p$-dimensional convex polytopes with $M$ facets and
$O(M^{[p/2]})$ vertices, but it is not known what happens if we bound the
number of facets plus the number of vertices of the final polytope, and
would like to determine how
many extra vertices an intermediate polytope might have.
The following example shows that in three dimensions,
this increase can be linear, which is, apart from the
constant, is the worst amount one can expect.
In the example, the final polyhedron consists of two regular $2n$-polygon-based pyramids 
joined
together at their bases. This polyhedron has $2n+2$ vertices and $4n$
faces. Leaving out every other facet of the top pyramid increases the number
of vertices to $3n+2$. It would be interesting to see if similar examples
can be construed in higher
dimensions.

The shortcomings of using the double description method have been
observed earlier \cite{extremalray}, and alternative algorithms have been suggested
to generate the extremal rays of pointed cones. In our case, however, these
algorithms could not be used directly, as Benson's
algorithm generates facets {\em and} extremal vertices
simultaneously, and so we do not know all the facets in advance. 
The procedure outlined in Section \ref{subsec:correct}
can be considered as an ``oracle call,'' which, given any point in the objective
space, returns whether the point is outside the polytope $\QT$, and
if yes, also provides a facet of $\QT$ that separates the polytope from this
point. The challenge then becomes to devise an algorithm which calls the
oracle, and
finds all vertices of $\QT$ in time and space {\em linear} in 
the final number of
vertices plus the number of facets.

Investigating the entropy region of {\em five random variables} instead of
four appears to be significantly harder. The very first obstacle is the lack
of our understanding even of the Shannon region of this 31-dimensional space.
There does not seem to be any suitable coordinate transformation similar to the
unimodular matrix $U$ which would simplify the structure of some 
cross-sections of this region
as $U$ did in the four-variable case. This simplification is what made
the application of Benson's algorithm possible, as we had an immediate
surrounding polytope (the non-negative orthant of $\R^{10}$), and 
as all vertices
of the projected polytope were {\em extremal vertices}, they
comprised the solution of a global multiobjective optimization 
problem.

\section*{Acknowledgments}
The author would like to acknowledge the help received during
the numerous insightful, fruitful, and
enjoyable discussions with Frantisek Mat\'u\v s on the entropy function,
matroids, and on the ultimate question of everything.


\section*{Appendix A}
\def\thesection{A}
\setcounter{subsection}{0}

\subsection{Shannon inequalities}
Let $\langle x_i:i\in I\rangle$ be a collection of random variables. For
$A\subseteq I$, we let $x_A = \langle x_i:i\in A\rangle$.
The
Shannon inequalities say that the entropy function is a monotone submodular
function on the subsets of $I$, that is,
\begin{equation}\label{eq:shannon1}
   \H(x_A)\le \H(x_B) ~~~\mbox{when $A\subseteq B \subseteq I$},
\end{equation}
and
\begin{equation}\label{eq:shannon2}
   \H(x_{A\cup B}) + \H(x_{A\cap B}) \le \H(x_A)+\H(x_B),
\end{equation}
for all subsets $A, B$ of $I$. Moreover, there is a minimal subset of these
inequalities from which all others follow: consider the inequalities from
(\ref{eq:shannon1}) where $B=I$, and $A$ is missing only one element of $I$;
and the inequalities from (\ref{eq:shannon2}) where both $A$ and $B$ has
exactly one element not in $A\cap B$.

\subsection{Independent copy of random variables}
If we split a set of random variables into two disjoint groups $\langle x_i: i\in
I\rangle$
and $\langle y_j: j\in J \rangle$, and create $\langle x'_i:i\in I\rangle$ as an {\em independent
copy} of $\langle x_i\rangle$ over $\langle y_j\rangle$,
then
the entropy of certain subsets of these variables can be
computed from other subsets as follows. Let $A, B \subseteq I$ and
$C\subseteq J$. Then,
$$
    \H(x'_A x_By_C) = \H(x'_B x_A y_C),
$$
which is due to the complete symmetry between $x'_i$ and $x_i$. The
fact that $x'_I$ and $x_I$ are independent over $y_J$ translates into the
following entropy equality:
$$
    \H(x'_A x_B y_J) = \H(x'_A y_J) + \H(x_B y_J) - \H(y_J)
$$
for all subsets $A,B\subseteq I$.

\subsection{Copy strings}
As explained in Section \ref{sec:problem}, we start from four random
variables, split them into two parts, create an independent copy of the
first part over the second, add the newly created random variables to the
group, and then repeat this process. To save on the number of variables
created, in each step certain newly generated variables are discarded,
or two or more new variables are merged into a single one. This process
is described by a {\em copy string}, which has the following form:
$$
\mbox{\tt rs=cd:ab;t=(cr):ab;u=t:acs}
$$
This string describes three iterations separated by semicolons. The initial
four variables are $a$, $b$, $c$ and $d$, and the newly created variables
are $r$, $s$, $t$ and $u$. In a copy step, variables after the colon are in the
``over'' set. We keep the copied image of those variables only which are
after the equality sign, and the copies are marked by the letters before the
equality sign. Thus in the first iteration, we make a copy of $cd$ over $ab$, and the copy of
$c$ and $d$ will be named $r$ and $s$, respectively. When variables are
enclosed in parentheses, their copies are merged into a single
variable. For example, in the second iteration we keep
the copies of $c$ and $r$ only, discard the copies of $d$ and $s$, and
merge these two copies into $t$.

\subsection{A unimodular matrix}
As explained in Section \ref{sec:problem}, it is advantageous to look at
the 15 entropies of the four random variables in another coordinate system.
The new coordinates can be computed using the unimodular matrix shown in Table
\ref{table:U}.
\begin{table}[h!tb]
\begin{center}\setlength\tabcolsep{4pt}\begin{tabular}{r|c*{14}c}
&$a$&$b$&$c$&$d$&$ab$&$ac$&$ad$&$bc$&$bd$&$cd$&$abc$&$abd$&$acd$&$bcd$&$abcd$\\
\hline
{\small Ingleton}&-1&-1& 0& 0&  1& 1& 1& 1& 1&-1&  -1& -1&  0&  0&  0 \\
& 0& 0&-1& 0&  0& 1& 0& 1& 0& 0&  -1&  0&  0&  0&  0 \\
&        0& 0& 0&-1&  0& 0& 1& 0& 1& 0&   0& -1&  0&  0&  0 \\
&        0&-1& 0& 0&  1& 0& 0& 1& 0& 0&  -1&  0&  0&  0&  0 \\
&       -1& 0& 0& 0&  1& 1& 0& 0& 0& 0&  -1&  0&  0&  0&  0 \\
&        0&-1& 0& 0&  1& 0& 0& 0& 1& 0&   0& -1&  0&  0&  0 \\
&       -1& 0& 0& 0&  1& 0& 1& 0& 0& 0&   0& -1&  0&  0&  0 \\
&       -1& 0& 0& 0&  0& 1& 1& 0& 0& 0&   0&  0& -1&  0&  0 \\
&        0&-1& 0& 0&  0& 0& 0& 1& 1& 0&   0&  0&  0& -1&  0 \\
&        0& 0& 1& 1&  0& 0& 0& 0& 0&-1&   0&  0&  0&  0&  0 \\
&        0& 0& 0& 0&  0& 0& 0& 0& 0&-1&   0&  0&  1&  1& -1 \\
z&        0& 0& 0& 0&  0& 0& 0& 0& 0& 0&   0&  0&  0& -1&  1 \\
z&        0& 0& 0& 0&  0& 0& 0& 0& 0& 0&   0&  0& -1&  0&  1 \\
z&        0& 0& 0& 0&  0& 0& 0& 0& 0& 0&   0& -1&  0&  0&  1 \\
z&        0& 0& 0& 0&  0& 0& 0& 0& 0& 0&  -1&  0&  0&  0&  1 \\
\end{tabular}\end{center}
\caption{The unimodular matrix}\label{table:U}
\end{table}
Columns represent the entropies of the subsets of the four random
variables $a$, $b$, $c$ and $d$, as indicated in the top row. The
value of the ``Ingleton row'' should be set to $1$, and rows marked by the
letter ``z'' vanish for all extremal vertices, thus they should be set to $0$.

\section*{Appendix B}
\def\I{\mathbf{I}}
This section lists entropy inequalities which were found during the
experiments described in Section \ref{sec:experimental} and have all
coefficients less than 100.
Each entry in the list contains nine integers 
representing the coefficients $c_0$, $c_1$, $\dots$, $c_8$
for the non-Shannon information inequality of the form
\begin{eqnarray*}
  &&  c_0\big(\I(c,d)-\I(a,b)+\I(a,b\,|\,c)+\I(a,b\,|\,d)\big) +{}\\
  && {}+ c_1\I(a,b\,|\,c)+c_2\I(a,b\,|\,d) + {}\\
  && {}+ c_3\I(a,c\,|\,b) + c_4\I(b,c\,|\,a)+
  c_5\I(a,d\,|\,b)+c_6\I(b,d\,|\,a) +{}\\
  && {}+ c_7\I(c,d\,|\,a) + c_8\I(c,d\,|\,b) \ge 0.
\end{eqnarray*}
Here $\I(A,B) = \H(A)+\H(B)-\H(AB)$ is the mutual information,
$\I(A,B\,|,C)=\H(AC)+\H(BC)-\H(ABC)-\H(C)$ is the conditional mutual
information. The expression after $c_0$ is the Ingleton value. Following
the list of coefficients is the copy string without the equality signs.
If the inequality, or one of its permuted variants, appears 
in the list of \cite{dougherty-etal}, then their number of the inequality
is given after the / symbol.
\bigskip
\begin{center}\makeatletter\newcommand\BL[1]{
\setbox\@tempboxa\hbox{\small #1}%
\ifdim \wd\@tempboxa < 0.48\linewidth
    \box\@tempboxa
\else
    \hbox{\vtop{\hsize 0.48\linewidth\hangindent0.2\linewidth\hangafter1\raggedright\small #1\strut}}%
\fi
}\makeatother
\def\Q#1#2{\global\def#1{\small\begin{tabular}[t]{r@{\hskip5pt}l}#2\end{tabular}}}%
\small
\input newineq %
\hbox to 0.96\linewidth{\hss\ta{ }\tb\hss}
\hbox to 0.96\linewidth{\hss\tc{ }\td\hss}
\hbox to 0.96\linewidth{\hss\te{ }\tf\hss}
\hbox to 0.96\linewidth{\hss\tg{ }\th\hss}
\hbox to 0.96\linewidth{\hss\ti{ }\tj\hss}
\end{center}

\end{document}

%% file: plotdata
\p 0.6,0;
\p 0.68,0;
\p 0.68,0.01;
\p 0.76,0;
\p 0.84,0.23;
\p 1,0.01;
\p 1,0.19;
\p 1,0.23;
\p 1,0.31;
\p 1.04,0.5;
\p 1,0.65;
\p 1.08,0.44;
\p 1.12,0.3;
\p 1.16,0.15;
\p 1.16,0.3;
\p 1.2,0.64;
\p 1.36,0.19;
\p 1.4,0.32;
\p 1.48,0.45;
\p 1.48,0.82;
\p 1.64,0;
\p 1.64,0.01;
\p 1.72,0.43;
\p 1.76,0.76;
\p 1.8,0.01;
\p 1.96,0.54;
\p 2,0.9;
\p 2,1.1;
\p 2.2,0.5;
\p 2.28,1.12;
\p 2.32,1.36;
\p 2.4,0.59;
\p 2.52,0.96;
\p 2.52,1.05;
\p 2.56,0.53;
\p 2.56,0.86;
\p 2.6,1.21;
\p 2.68,0.03;
\p 2.68,0.84;
\p 2.68,1.38;
\p 2.76,1.17;
\p 2.8,1.35;
\p 2.92,1.08;
\p 2.96,1.83;
\p 2.96,2.56;
\p 3.04,1.29;
\p 3.04,1.59;
\p 3.08,0.84;
\p 3,0.83;
\p 3.12,0.75;
\p 3.16,0.9;
\p 3.24,1.4;
\p 3.24,1.66;
\p 3.32,0.89;
\p 3.36,1.22;
\p 3.36,1.56;
\p 3.36,1.91;
\p 3.6,1.71;
\p 3.84,1.23;
\p 3.92,2.83;
\p 4.04,2.03;
\p 4.16,3.16;
\p 4.28,2.09;
\p 4.36,2.38;
\p 4.44,1.11;
\p 4.6,2.17;
\p 4.6,2.9;
\p 4.6,3.17;
\p 4.88,5.08;
\p 4.92,2.41;
\p 4.92,2.73;
\p 4.92,3.99;
\p 4.92,4.11;
\p 5.12,3.31;
\p 5.12,3.62;
\p 5.64,4.6;
\p 5.64,5.05;
\p 5.64,5.79;
\p 5.68,2.67;
\p 5.68,4.01;
\p 5.72,3.25;
\p 5.72,3.41;
\p 5.84,3.85;
\p 5.8,7.81;
\p 5.92,4.97;
\p 5.92,6.66;
\p 5.96,3.59;
\p 6.52,3.19;
\p 6.64,3.59;
\p 6.88,3.79;
\p 6.88,3.79;
\p 6.88,5.74;
\p 7,5.36;
\p 7.72,4.69;
\p 8.32,7.91;
\p 8.44,7.78;
\p 8.6,7.07;
\p 8.76,11.07;
\p 8.76,8.7;
\p 9.72,5.97;
\p 11.64,13.2;
\p 11.68,4.89;
\p 13.6,11.36;
\p 14.04,13.91;
\p 14.04,15.17;
\p 16.8,14.61;
\p 17.04,10.49;
\p 17.04,13.78;
\p 17.52,10.6;
\p 17.52,15.7;
\p 17.88,10.1;
\p 17.88,11.42;
\p 18.88,15.09;
\p 18.88,17.62;
\p 19.12,11.06;
\p 19.52,15.82;
\p 19.64,23.46;
\p 25.24,20.93;
\p 25.8,17.58;
\p 28.52,25.11;
\p 32.08,20.78;
\p 32.08,23.09;
\p 33.64,29.48;
\p 33.64,48.89;
\p 37.84,23.39;
\p 37.84,40.21;
\p 38.48,38.72;
\p 42.8,40.49;
\p 46.4,59.7;
\p 55.76,58.78;

%% file: newineq.tex
\Q\ta{%
1)&\BL{2{\hskip 3pt}1{\hskip 3pt}0{\hskip 3pt}3{\hskip 3pt}2{\hskip 3pt}0{\hskip 3pt}0{\hskip 3pt}0{\hskip 3pt}0\hskip 3pt\sf rs~cd:ab;\allowbreak t~b:acr\allowbreak/5}\\
2)&\BL{2{\hskip 3pt}1{\hskip 3pt}0{\hskip 3pt}3{\hskip 3pt}1{\hskip 3pt}3{\hskip 3pt}2{\hskip 3pt}0{\hskip 3pt}0\hskip 3pt\sf rs~cd:ab;\allowbreak t~c:abrs;\allowbreak u~a:bdst\allowbreak/30}\\
3)&\BL{2{\hskip 3pt}1{\hskip 3pt}0{\hskip 3pt}3{\hskip 3pt}1{\hskip 3pt}0{\hskip 3pt}0{\hskip 3pt}3{\hskip 3pt}0\hskip 3pt\sf rs~bd:ac;\allowbreak tu~cd:abs;\allowbreak v~b:acrst}\\
4)&\BL{2{\hskip 3pt}2{\hskip 3pt}0{\hskip 3pt}2{\hskip 3pt}1{\hskip 3pt}0{\hskip 3pt}0{\hskip 3pt}3{\hskip 3pt}0\hskip 3pt\sf rs~bd:ac;\allowbreak t~c:adr\allowbreak/4}\\
5)&\BL{2{\hskip 3pt}2{\hskip 3pt}0{\hskip 3pt}2{\hskip 3pt}1{\hskip 3pt}0{\hskip 3pt}3{\hskip 3pt}0{\hskip 3pt}0\hskip 3pt\sf rs~cd:ab;\allowbreak t~b:acs\allowbreak/9}\\
6)&\BL{2{\hskip 3pt}2{\hskip 3pt}0{\hskip 3pt}3{\hskip 3pt}1{\hskip 3pt}0{\hskip 3pt}0{\hskip 3pt}0{\hskip 3pt}0\hskip 3pt\sf rs~ac:bd;\allowbreak t~b:adr\allowbreak/8}\\
7)&\BL{2{\hskip 3pt}2{\hskip 3pt}1{\hskip 3pt}2{\hskip 3pt}1{\hskip 3pt}1{\hskip 3pt}0{\hskip 3pt}0{\hskip 3pt}0\hskip 3pt\sf r~c:ab;\allowbreak s~r:ad\allowbreak/1}\\
8)&\BL{2{\hskip 3pt}2{\hskip 3pt}1{\hskip 3pt}5{\hskip 3pt}0{\hskip 3pt}4{\hskip 3pt}1{\hskip 3pt}0{\hskip 3pt}0\hskip 3pt\sf rs~cd:ab;\allowbreak t~c:abrs;\allowbreak u~a:dst\allowbreak/27}\\
9)&\BL{2{\hskip 3pt}3{\hskip 3pt}0{\hskip 3pt}2{\hskip 3pt}1{\hskip 3pt}0{\hskip 3pt}0{\hskip 3pt}0{\hskip 3pt}0\hskip 3pt\sf rs~ac:bd;\allowbreak t~d:abr\allowbreak/7}\\
}\Q\tb{%
10)&\BL{2{\hskip 3pt}3{\hskip 3pt}1{\hskip 3pt}2{\hskip 3pt}0{\hskip 3pt}3{\hskip 3pt}1{\hskip 3pt}0{\hskip 3pt}0\hskip 3pt\sf rs~cd:ab;\allowbreak t~r:ad\allowbreak/2}\\
11)&\BL{3{\hskip 3pt}1{\hskip 3pt}0{\hskip 3pt}5{\hskip 3pt}5{\hskip 3pt}0{\hskip 3pt}0{\hskip 3pt}0{\hskip 3pt}0\hskip 3pt\sf rs~cd:ab;\allowbreak t~a:bcr;\allowbreak u~t:acr\allowbreak/94}\\
12)&\BL{3{\hskip 3pt}1{\hskip 3pt}1{\hskip 3pt}5{\hskip 3pt}4{\hskip 3pt}1{\hskip 3pt}0{\hskip 3pt}0{\hskip 3pt}0\hskip 3pt\sf rs~cd:ab;\allowbreak t~a:bcr;\allowbreak u~a:bdst\allowbreak/58}\\
13)&\BL{3{\hskip 3pt}1{\hskip 3pt}1{\hskip 3pt}5{\hskip 3pt}4{\hskip 3pt}0{\hskip 3pt}1{\hskip 3pt}0{\hskip 3pt}0\hskip 3pt\sf rs~cd:ab;\allowbreak t~b:acr;\allowbreak u~t:bcs\allowbreak/96}\\
14)&\BL{3{\hskip 3pt}2{\hskip 3pt}0{\hskip 3pt}4{\hskip 3pt}2{\hskip 3pt}0{\hskip 3pt}2{\hskip 3pt}0{\hskip 3pt}0\hskip 3pt\sf rs~cd:ab;\allowbreak t~b:adr;\allowbreak u~t:br\allowbreak/162}\\
15)&\BL{3{\hskip 3pt}2{\hskip 3pt}0{\hskip 3pt}6{\hskip 3pt}2{\hskip 3pt}0{\hskip 3pt}0{\hskip 3pt}0{\hskip 3pt}0\hskip 3pt\sf rs~ac:bd;\allowbreak t~a:dr;\allowbreak u~b:adrt\allowbreak/60}\\
16)&\BL{3{\hskip 3pt}2{\hskip 3pt}1{\hskip 3pt}4{\hskip 3pt}2{\hskip 3pt}1{\hskip 3pt}0{\hskip 3pt}0{\hskip 3pt}0\hskip 3pt\sf rs~cd:ab;\allowbreak t~c:as;\allowbreak u~b:adst\allowbreak/41}\\
17)&\BL{3{\hskip 3pt}3{\hskip 3pt}0{\hskip 3pt}5{\hskip 3pt}1{\hskip 3pt}1{\hskip 3pt}3{\hskip 3pt}0{\hskip 3pt}0\hskip 3pt\sf rs~cd:ab;\allowbreak t~c:abrs;\allowbreak u~t:adr\allowbreak/102}\\
18)&\BL{3{\hskip 3pt}3{\hskip 3pt}0{\hskip 3pt}6{\hskip 3pt}1{\hskip 3pt}1{\hskip 3pt}1{\hskip 3pt}0{\hskip 3pt}0\hskip 3pt\sf rs~ac:bd;\allowbreak t~b:ars;\allowbreak u~b:adr\allowbreak/142}\\
}\Q\tc{%
19)&\BL{3{\hskip 3pt}3{\hskip 3pt}1{\hskip 3pt}2{\hskip 3pt}2{\hskip 3pt}1{\hskip 3pt}0{\hskip 3pt}0{\hskip 3pt}5\hskip 3pt\sf rs~ac:bd;\allowbreak t~d:bcr;\allowbreak u~t:rs\allowbreak/110}\\
20)&\BL{3{\hskip 3pt}3{\hskip 3pt}1{\hskip 3pt}3{\hskip 3pt}1{\hskip 3pt}1{\hskip 3pt}0{\hskip 3pt}5{\hskip 3pt}0\hskip 3pt\sf rs~bd:ac;\allowbreak tu~cs:adr\allowbreak/12}\\
21)&\BL{3{\hskip 3pt}3{\hskip 3pt}2{\hskip 3pt}5{\hskip 3pt}1{\hskip 3pt}3{\hskip 3pt}0{\hskip 3pt}0{\hskip 3pt}0\hskip 3pt\sf r~c:ab;\allowbreak st~ac:bdr;\allowbreak u~r:ads\allowbreak/21}\\
22)&\BL{3{\hskip 3pt}3{\hskip 3pt}3{\hskip 3pt}3{\hskip 3pt}1{\hskip 3pt}3{\hskip 3pt}0{\hskip 3pt}0{\hskip 3pt}0\hskip 3pt\sf rs~cd:ab;\allowbreak t~r:ad\allowbreak/3}\\
23)&\BL{3{\hskip 3pt}3{\hskip 3pt}3{\hskip 3pt}5{\hskip 3pt}1{\hskip 3pt}2{\hskip 3pt}0{\hskip 3pt}0{\hskip 3pt}0\hskip 3pt\sf rs~cd:ab;\allowbreak t~c:ar;\allowbreak u~r:ad\allowbreak/206}\\
24)&\BL{3{\hskip 3pt}4{\hskip 3pt}0{\hskip 3pt}2{\hskip 3pt}2{\hskip 3pt}0{\hskip 3pt}0{\hskip 3pt}2{\hskip 3pt}0\hskip 3pt\sf rs~bd:ac;\allowbreak t~c:abs;\allowbreak u~t:bcr\allowbreak/120}\\
25)&\BL{3{\hskip 3pt}4{\hskip 3pt}0{\hskip 3pt}2{\hskip 3pt}2{\hskip 3pt}4{\hskip 3pt}0{\hskip 3pt}0{\hskip 3pt}0\hskip 3pt\sf rs~cd:ab;\allowbreak t~a:bcs;\allowbreak u~s:abcdt\allowbreak/178}\\
26)&\BL{3{\hskip 3pt}4{\hskip 3pt}0{\hskip 3pt}4{\hskip 3pt}1{\hskip 3pt}1{\hskip 3pt}1{\hskip 3pt}0{\hskip 3pt}4\hskip 3pt\sf rs~ad:bc;\allowbreak tu~cd:abs;\allowbreak v~a:bcrst}\\
27)&\BL{3{\hskip 3pt}4{\hskip 3pt}1{\hskip 3pt}4{\hskip 3pt}1{\hskip 3pt}1{\hskip 3pt}0{\hskip 3pt}0{\hskip 3pt}0\hskip 3pt\sf r~c:ab;\allowbreak s~r:ac;\allowbreak t~r:ad\allowbreak/10}\\
28)&\BL{3{\hskip 3pt}5{\hskip 3pt}0{\hskip 3pt}5{\hskip 3pt}1{\hskip 3pt}0{\hskip 3pt}0{\hskip 3pt}0{\hskip 3pt}0\hskip 3pt\sf rs~ac:bd;\allowbreak t~b:adr;\allowbreak u~t:abr\allowbreak/109}\\
29)&\BL{3{\hskip 3pt}5{\hskip 3pt}2{\hskip 3pt}5{\hskip 3pt}0{\hskip 3pt}4{\hskip 3pt}1{\hskip 3pt}0{\hskip 3pt}0\hskip 3pt\sf rs~cd:ab;\allowbreak t~c:as;\allowbreak u~a:bdst\allowbreak/48}\\
30)&\BL{3{\hskip 3pt}5{\hskip 3pt}3{\hskip 3pt}4{\hskip 3pt}1{\hskip 3pt}0{\hskip 3pt}7{\hskip 3pt}0{\hskip 3pt}0\hskip 3pt\sf rs~cd:ab;\allowbreak t~c:bs;\allowbreak u~t:cs\allowbreak/83}\\
31)&\BL{3{\hskip 3pt}6{\hskip 3pt}0{\hskip 3pt}2{\hskip 3pt}2{\hskip 3pt}0{\hskip 3pt}0{\hskip 3pt}0{\hskip 3pt}0\hskip 3pt\sf rs~ac:bd;\allowbreak t~a:br;\allowbreak u~d:abrt\allowbreak/56}\\
32)&\BL{4{\hskip 3pt}1{\hskip 3pt}0{\hskip 3pt}4{\hskip 3pt}4{\hskip 3pt}6{\hskip 3pt}5{\hskip 3pt}0{\hskip 3pt}0\hskip 3pt\sf rs~cd:ab;\allowbreak t~(dr):ab;\allowbreak u~b:adst\allowbreak/153}\\
33)&\BL{4{\hskip 3pt}1{\hskip 3pt}0{\hskip 3pt}5{\hskip 3pt}4{\hskip 3pt}3{\hskip 3pt}3{\hskip 3pt}0{\hskip 3pt}0\hskip 3pt\sf rs~cd:ab;\allowbreak t~b:acr;\allowbreak u~(dr):ab\allowbreak/158}\\
34)&\BL{4{\hskip 3pt}1{\hskip 3pt}0{\hskip 3pt}6{\hskip 3pt}4{\hskip 3pt}2{\hskip 3pt}4{\hskip 3pt}0{\hskip 3pt}0\hskip 3pt\sf rs~cd:ab;\allowbreak t~d:abrs;\allowbreak u~b:acrt\allowbreak/28}\\
35)&\BL{4{\hskip 3pt}1{\hskip 3pt}0{\hskip 3pt}7{\hskip 3pt}6{\hskip 3pt}1{\hskip 3pt}1{\hskip 3pt}0{\hskip 3pt}0\hskip 3pt\sf r~c:ab;\allowbreak st~cd:ab;\allowbreak u~b:acrs\allowbreak/23}\\
36)&\BL{4{\hskip 3pt}1{\hskip 3pt}0{\hskip 3pt}10{\hskip 3pt}9{\hskip 3pt}0{\hskip 3pt}0{\hskip 3pt}0{\hskip 3pt}0\hskip 3pt\sf rs~cd:ab;\allowbreak t~(cr):ab;\allowbreak u~b:acrt\allowbreak/93}\\
37)&\BL{4{\hskip 3pt}1{\hskip 3pt}1{\hskip 3pt}9{\hskip 3pt}7{\hskip 3pt}1{\hskip 3pt}0{\hskip 3pt}0{\hskip 3pt}0\hskip 3pt\sf rs~ac:bd;\allowbreak t~a:bcr;\allowbreak u~a:dr\allowbreak/169}\\
38)&\BL{4{\hskip 3pt}2{\hskip 3pt}0{\hskip 3pt}4{\hskip 3pt}3{\hskip 3pt}2{\hskip 3pt}1{\hskip 3pt}0{\hskip 3pt}0\hskip 3pt\sf rs~cd:ab;\allowbreak t~b:adr\allowbreak/6}\\
39)&\BL{4{\hskip 3pt}2{\hskip 3pt}0{\hskip 3pt}5{\hskip 3pt}4{\hskip 3pt}0{\hskip 3pt}2{\hskip 3pt}0{\hskip 3pt}0\hskip 3pt\sf rs~cd:ab;\allowbreak t~b:acs;\allowbreak uv~at:bcr}\\
40)&\BL{4{\hskip 3pt}2{\hskip 3pt}0{\hskip 3pt}5{\hskip 3pt}4{\hskip 3pt}2{\hskip 3pt}0{\hskip 3pt}0{\hskip 3pt}0\hskip 3pt\sf rs~cd:ab;\allowbreak t~a:bcs;\allowbreak uv~bt:acr}\\
41)&\BL{4{\hskip 3pt}2{\hskip 3pt}1{\hskip 3pt}3{\hskip 3pt}0{\hskip 3pt}9{\hskip 3pt}6{\hskip 3pt}0{\hskip 3pt}0\hskip 3pt\sf rs~cd:ab;\allowbreak t~a:bcr;\allowbreak u~a:bdst\allowbreak/64}\\
42)&\BL{4{\hskip 3pt}3{\hskip 3pt}0{\hskip 3pt}4{\hskip 3pt}2{\hskip 3pt}1{\hskip 3pt}1{\hskip 3pt}7{\hskip 3pt}0\hskip 3pt\sf rs~bd:ac;\allowbreak t~c:adr;\allowbreak u~a:rst\allowbreak/29}\\
43)&\BL{4{\hskip 3pt}3{\hskip 3pt}1{\hskip 3pt}2{\hskip 3pt}0{\hskip 3pt}10{\hskip 3pt}5{\hskip 3pt}0{\hskip 3pt}0\hskip 3pt\sf rs~cd:ab;\allowbreak tu~cr:ab;\allowbreak v~(rtu):ad}\\
44)&\BL{4{\hskip 3pt}3{\hskip 3pt}1{\hskip 3pt}4{\hskip 3pt}3{\hskip 3pt}0{\hskip 3pt}5{\hskip 3pt}0{\hskip 3pt}0\hskip 3pt\sf rs~cd:ab;\allowbreak t~b:adr;\allowbreak u~c:abrst\allowbreak/32}\\
45)&\BL{4{\hskip 3pt}3{\hskip 3pt}2{\hskip 3pt}4{\hskip 3pt}3{\hskip 3pt}2{\hskip 3pt}0{\hskip 3pt}0{\hskip 3pt}0\hskip 3pt\sf rs~cd:ab;\allowbreak t~c:as;\allowbreak u~a:bds\allowbreak/163}\\
46)&\BL{4{\hskip 3pt}4{\hskip 3pt}0{\hskip 3pt}5{\hskip 3pt}2{\hskip 3pt}0{\hskip 3pt}0{\hskip 3pt}0{\hskip 3pt}2\hskip 3pt\sf rs~ad:bc;\allowbreak t~c:abs;\allowbreak uv~bt:acr}\\
47)&\BL{4{\hskip 3pt}4{\hskip 3pt}0{\hskip 3pt}11{\hskip 3pt}1{\hskip 3pt}4{\hskip 3pt}4{\hskip 3pt}0{\hskip 3pt}0\hskip 3pt\sf rs~cd:ab;\allowbreak t~a:bcr;\allowbreak u~s:abcdt\allowbreak/68}\\
48)&\BL{4{\hskip 3pt}4{\hskip 3pt}3{\hskip 3pt}2{\hskip 3pt}0{\hskip 3pt}5{\hskip 3pt}2{\hskip 3pt}0{\hskip 3pt}0\hskip 3pt\sf rs~cd:ab;\allowbreak t~r:ad;\allowbreak u~r:adt\allowbreak/129}\\
49)&\BL{4{\hskip 3pt}5{\hskip 3pt}0{\hskip 3pt}4{\hskip 3pt}2{\hskip 3pt}0{\hskip 3pt}0{\hskip 3pt}0{\hskip 3pt}2\hskip 3pt\sf rs~ad:bc;\allowbreak t~c:abs;\allowbreak uv~bt:acr}\\
50)&\BL{4{\hskip 3pt}5{\hskip 3pt}0{\hskip 3pt}8{\hskip 3pt}1{\hskip 3pt}2{\hskip 3pt}3{\hskip 3pt}0{\hskip 3pt}0\hskip 3pt\sf rs~cd:ab;\allowbreak t~a:bcr;\allowbreak u~a:bdst\allowbreak/111}\\
51)&\BL{4{\hskip 3pt}5{\hskip 3pt}0{\hskip 3pt}8{\hskip 3pt}1{\hskip 3pt}3{\hskip 3pt}2{\hskip 3pt}0{\hskip 3pt}0\hskip 3pt\sf rs~cd:ab;\allowbreak t~a:bcr;\allowbreak u~b:adst\allowbreak/70}\\
52)&\BL{4{\hskip 3pt}5{\hskip 3pt}1{\hskip 3pt}3{\hskip 3pt}2{\hskip 3pt}7{\hskip 3pt}0{\hskip 3pt}0{\hskip 3pt}0\hskip 3pt\sf rs~cd:ab;\allowbreak t~a:bcs;\allowbreak u~(at):bcs\allowbreak/91}\\
53)&\BL{4{\hskip 3pt}5{\hskip 3pt}1{\hskip 3pt}3{\hskip 3pt}2{\hskip 3pt}2{\hskip 3pt}0{\hskip 3pt}6{\hskip 3pt}0\hskip 3pt\sf rs~bd:ac;\allowbreak tu~cd:abs;\allowbreak v~b:acrst}\\
54)&\BL{4{\hskip 3pt}5{\hskip 3pt}3{\hskip 3pt}5{\hskip 3pt}1{\hskip 3pt}3{\hskip 3pt}0{\hskip 3pt}0{\hskip 3pt}0\hskip 3pt\sf rs~cd:ab;\allowbreak t~c:ar;\allowbreak u~r:ad\allowbreak/176}\\
55)&\BL{4{\hskip 3pt}5{\hskip 3pt}4{\hskip 3pt}3{\hskip 3pt}2{\hskip 3pt}4{\hskip 3pt}0{\hskip 3pt}0{\hskip 3pt}0\hskip 3pt\sf rs~cd:ab;\allowbreak t~c:as;\allowbreak u~r:bst\allowbreak/38}\\
56)&\BL{4{\hskip 3pt}6{\hskip 3pt}0{\hskip 3pt}10{\hskip 3pt}1{\hskip 3pt}1{\hskip 3pt}4{\hskip 3pt}0{\hskip 3pt}0\hskip 3pt\sf rs~cd:ab;\allowbreak t~a:bcr;\allowbreak u~b:adrt\allowbreak/46}\\
57)&\BL{4{\hskip 3pt}6{\hskip 3pt}1{\hskip 3pt}3{\hskip 3pt}2{\hskip 3pt}0{\hskip 3pt}1{\hskip 3pt}0{\hskip 3pt}4\hskip 3pt\sf rs~ac:bd;\allowbreak t~a:bcr;\allowbreak u~d:abr\allowbreak/168}\\
58)&\BL{4{\hskip 3pt}6{\hskip 3pt}1{\hskip 3pt}6{\hskip 3pt}0{\hskip 3pt}16{\hskip 3pt}4{\hskip 3pt}0{\hskip 3pt}0\hskip 3pt\sf r~a:bc;\allowbreak st~ac:bdr;\allowbreak u~(as):bdr\allowbreak/20}\\
59)&\BL{4{\hskip 3pt}6{\hskip 3pt}4{\hskip 3pt}6{\hskip 3pt}0{\hskip 3pt}4{\hskip 3pt}1{\hskip 3pt}0{\hskip 3pt}0\hskip 3pt\sf r~d:ab;\allowbreak st~bd:ac;\allowbreak u~r:as\allowbreak/24}\\
60)&\BL{4{\hskip 3pt}6{\hskip 3pt}5{\hskip 3pt}3{\hskip 3pt}2{\hskip 3pt}0{\hskip 3pt}9{\hskip 3pt}0{\hskip 3pt}0\hskip 3pt\sf rs~cd:ab;\allowbreak t~c:bs;\allowbreak u~r:adst\allowbreak/31}\\
61)&\BL{4{\hskip 3pt}7{\hskip 3pt}0{\hskip 3pt}7{\hskip 3pt}1{\hskip 3pt}1{\hskip 3pt}7{\hskip 3pt}0{\hskip 3pt}0\hskip 3pt\sf rs~cd:ab;\allowbreak t~b:acs;\allowbreak u~a:bcrt\allowbreak/130}\\
62)&\BL{4{\hskip 3pt}7{\hskip 3pt}0{\hskip 3pt}7{\hskip 3pt}1{\hskip 3pt}4{\hskip 3pt}5{\hskip 3pt}0{\hskip 3pt}0\hskip 3pt\sf rs~cd:ab;\allowbreak t~b:acs;\allowbreak u~t:ar\allowbreak/114}\\
63)&\BL{4{\hskip 3pt}7{\hskip 3pt}0{\hskip 3pt}7{\hskip 3pt}1{\hskip 3pt}2{\hskip 3pt}6{\hskip 3pt}0{\hskip 3pt}0\hskip 3pt\sf rs~cd:ab;\allowbreak t~b:acs;\allowbreak u~t:ar\allowbreak/148}\\
64)&\BL{4{\hskip 3pt}9{\hskip 3pt}0{\hskip 3pt}7{\hskip 3pt}1{\hskip 3pt}2{\hskip 3pt}2{\hskip 3pt}0{\hskip 3pt}0\hskip 3pt\sf r~b:ac;\allowbreak st~cd:ab;\allowbreak u~a:bcrs\allowbreak/22}\\
65)&\BL{4{\hskip 3pt}9{\hskip 3pt}0{\hskip 3pt}9{\hskip 3pt}1{\hskip 3pt}0{\hskip 3pt}0{\hskip 3pt}7{\hskip 3pt}0\hskip 3pt\sf rs~bd:ac;\allowbreak t~c:abs;\allowbreak u~t:ar\allowbreak/107}\\
66)&\BL{4{\hskip 3pt}9{\hskip 3pt}0{\hskip 3pt}9{\hskip 3pt}1{\hskip 3pt}0{\hskip 3pt}7{\hskip 3pt}0{\hskip 3pt}0\hskip 3pt\sf rs~cd:ab;\allowbreak t~b:adr;\allowbreak u~t:ac\allowbreak/171}\\
67)&\BL{4{\hskip 3pt}9{\hskip 3pt}0{\hskip 3pt}10{\hskip 3pt}1{\hskip 3pt}0{\hskip 3pt}0{\hskip 3pt}0{\hskip 3pt}0\hskip 3pt\sf rs~ac:bd;\allowbreak t~(ar):bd;\allowbreak u~b:adrt\allowbreak/108}\\
68)&\BL{4{\hskip 3pt}9{\hskip 3pt}3{\hskip 3pt}8{\hskip 3pt}0{\hskip 3pt}5{\hskip 3pt}1{\hskip 3pt}0{\hskip 3pt}0\hskip 3pt\sf rs~cd:ab;\allowbreak t~c:as;\allowbreak u~t:ad\allowbreak/105}\\
69)&\BL{4{\hskip 3pt}10{\hskip 3pt}0{\hskip 3pt}9{\hskip 3pt}1{\hskip 3pt}0{\hskip 3pt}0{\hskip 3pt}0{\hskip 3pt}0\hskip 3pt\sf rs~ac:bd;\allowbreak t~(ar):bd;\allowbreak u~d:abrt\allowbreak/77}\\
70)&\BL{5{\hskip 3pt}1{\hskip 3pt}0{\hskip 3pt}5{\hskip 3pt}5{\hskip 3pt}10{\hskip 3pt}10{\hskip 3pt}0{\hskip 3pt}0\hskip 3pt\sf rs~cd:ab;\allowbreak t~(cs):ab;\allowbreak u~(ds):abt\allowbreak/87}\\
71)&\BL{5{\hskip 3pt}1{\hskip 3pt}0{\hskip 3pt}6{\hskip 3pt}6{\hskip 3pt}6{\hskip 3pt}5{\hskip 3pt}0{\hskip 3pt}0\hskip 3pt\sf rs~cd:ab;\allowbreak tu~dr:ab;\allowbreak v~b:adstu}\\
72)&\BL{5{\hskip 3pt}1{\hskip 3pt}0{\hskip 3pt}7{\hskip 3pt}6{\hskip 3pt}4{\hskip 3pt}4{\hskip 3pt}0{\hskip 3pt}0\hskip 3pt\sf rs~cd:ab;\allowbreak t~b:acr;\allowbreak u~(dr):ab\allowbreak/173}\\
73)&\BL{5{\hskip 3pt}1{\hskip 3pt}0{\hskip 3pt}8{\hskip 3pt}7{\hskip 3pt}3{\hskip 3pt}2{\hskip 3pt}0{\hskip 3pt}0\hskip 3pt\sf r~c:ab;\allowbreak st~cd:abr;\allowbreak u~b:acrt\allowbreak/18}\\
74)&\BL{5{\hskip 3pt}1{\hskip 3pt}0{\hskip 3pt}9{\hskip 3pt}8{\hskip 3pt}2{\hskip 3pt}2{\hskip 3pt}0{\hskip 3pt}0\hskip 3pt\sf rs~cd:ab;\allowbreak tu~dr:ab;\allowbreak v~b:adst}\\
75)&\BL{5{\hskip 3pt}1{\hskip 3pt}0{\hskip 3pt}14{\hskip 3pt}9{\hskip 3pt}1{\hskip 3pt}2{\hskip 3pt}0{\hskip 3pt}0\hskip 3pt\sf rs~cd:ab;\allowbreak t~(cr):ab;\allowbreak u~a:bdrt\allowbreak/136}\\
76)&\BL{5{\hskip 3pt}1{\hskip 3pt}0{\hskip 3pt}15{\hskip 3pt}15{\hskip 3pt}0{\hskip 3pt}0{\hskip 3pt}0{\hskip 3pt}0\hskip 3pt\sf rs~cd:ab;\allowbreak t~(cr):ab;\allowbreak u~(cr):abt\allowbreak/86}\\
77)&\BL{5{\hskip 3pt}2{\hskip 3pt}0{\hskip 3pt}7{\hskip 3pt}4{\hskip 3pt}1{\hskip 3pt}2{\hskip 3pt}0{\hskip 3pt}0\hskip 3pt\sf rs~cd:ab;\allowbreak t~b:acr;\allowbreak u~b:adr\allowbreak/185}\\
}\Q\td{%
78)&\BL{5{\hskip 3pt}2{\hskip 3pt}0{\hskip 3pt}11{\hskip 3pt}6{\hskip 3pt}2{\hskip 3pt}0{\hskip 3pt}0{\hskip 3pt}0\hskip 3pt\sf rs~cd:ab;\allowbreak t~(cr):ab;\allowbreak u~t:adr\allowbreak/112}\\
79)&\BL{5{\hskip 3pt}3{\hskip 3pt}0{\hskip 3pt}5{\hskip 3pt}3{\hskip 3pt}4{\hskip 3pt}4{\hskip 3pt}0{\hskip 3pt}0\hskip 3pt\sf rs~cd:ab;\allowbreak t~b:acs;\allowbreak u~t:bcr\allowbreak/170}\\
80)&\BL{5{\hskip 3pt}3{\hskip 3pt}0{\hskip 3pt}5{\hskip 3pt}3{\hskip 3pt}2{\hskip 3pt}5{\hskip 3pt}0{\hskip 3pt}0\hskip 3pt\sf rs~cd:ab;\allowbreak t~b:acs;\allowbreak u~t:bcr\allowbreak/151}\\
81)&\BL{5{\hskip 3pt}3{\hskip 3pt}0{\hskip 3pt}6{\hskip 3pt}5{\hskip 3pt}1{\hskip 3pt}0{\hskip 3pt}0{\hskip 3pt}0\hskip 3pt\sf rs~cd:ab;\allowbreak t~a:bcs;\allowbreak uv~bt:acr}\\
82)&\BL{5{\hskip 3pt}3{\hskip 3pt}0{\hskip 3pt}10{\hskip 3pt}3{\hskip 3pt}1{\hskip 3pt}1{\hskip 3pt}0{\hskip 3pt}0\hskip 3pt\sf rs~cd:ab;\allowbreak t~c:ar;\allowbreak u~b:acrt\allowbreak/45}\\
83)&\BL{5{\hskip 3pt}4{\hskip 3pt}0{\hskip 3pt}4{\hskip 3pt}4{\hskip 3pt}3{\hskip 3pt}1{\hskip 3pt}0{\hskip 3pt}0\hskip 3pt\sf rs~cd:ab;\allowbreak t~a:bcs;\allowbreak u~(cs):abrt\allowbreak/54}\\
84)&\BL{5{\hskip 3pt}4{\hskip 3pt}0{\hskip 3pt}7{\hskip 3pt}2{\hskip 3pt}3{\hskip 3pt}2{\hskip 3pt}0{\hskip 3pt}0\hskip 3pt\sf rs~cd:ab;\allowbreak t~a:bcr;\allowbreak u~b:acs\allowbreak/183}\\
85)&\BL{5{\hskip 3pt}4{\hskip 3pt}1{\hskip 3pt}5{\hskip 3pt}4{\hskip 3pt}2{\hskip 3pt}0{\hskip 3pt}0{\hskip 3pt}0\hskip 3pt\sf rs~cd:ab;\allowbreak t~a:bcs;\allowbreak u~c:abrt\allowbreak/146}\\
86)&\BL{5{\hskip 3pt}4{\hskip 3pt}2{\hskip 3pt}3{\hskip 3pt}0{\hskip 3pt}9{\hskip 3pt}4{\hskip 3pt}0{\hskip 3pt}0\hskip 3pt\sf rs~cd:ab;\allowbreak t~(cr):ab;\allowbreak uv~rt:ad}\\
87)&\BL{5{\hskip 3pt}5{\hskip 3pt}0{\hskip 3pt}17{\hskip 3pt}1{\hskip 3pt}7{\hskip 3pt}7{\hskip 3pt}0{\hskip 3pt}0\hskip 3pt\sf rs~cd:ab;\allowbreak t~a:bcr;\allowbreak uv~ds:abt}\\
88)&\BL{5{\hskip 3pt}6{\hskip 3pt}0{\hskip 3pt}14{\hskip 3pt}1{\hskip 3pt}6{\hskip 3pt}6{\hskip 3pt}0{\hskip 3pt}0\hskip 3pt\sf rs~cd:ab;\allowbreak t~a:bcr;\allowbreak u~(ds):abt\allowbreak/191}\\
89)&\BL{5{\hskip 3pt}6{\hskip 3pt}5{\hskip 3pt}4{\hskip 3pt}2{\hskip 3pt}8{\hskip 3pt}0{\hskip 3pt}0{\hskip 3pt}0\hskip 3pt\sf rs~cd:ab;\allowbreak t~c:as;\allowbreak u~c:bdt\allowbreak/79}\\
90)&\BL{5{\hskip 3pt}7{\hskip 3pt}5{\hskip 3pt}3{\hskip 3pt}3{\hskip 3pt}9{\hskip 3pt}0{\hskip 3pt}0{\hskip 3pt}0\hskip 3pt\sf rs~cd:ab;\allowbreak t~c:as;\allowbreak u~r:bdst\allowbreak/39}\\
91)&\BL{5{\hskip 3pt}8{\hskip 3pt}3{\hskip 3pt}8{\hskip 3pt}1{\hskip 3pt}3{\hskip 3pt}0{\hskip 3pt}0{\hskip 3pt}0\hskip 3pt\sf rs~cd:ab;\allowbreak t~c:ar;\allowbreak u~r:ad\allowbreak/159}\\
92)&\BL{5{\hskip 3pt}9{\hskip 3pt}0{\hskip 3pt}11{\hskip 3pt}1{\hskip 3pt}2{\hskip 3pt}2{\hskip 3pt}0{\hskip 3pt}0\hskip 3pt\sf rs~cd:ab;\allowbreak t~a:bcr;\allowbreak u~a:bcrt\allowbreak/34}\\
93)&\BL{5{\hskip 3pt}9{\hskip 3pt}2{\hskip 3pt}9{\hskip 3pt}1{\hskip 3pt}2{\hskip 3pt}0{\hskip 3pt}0{\hskip 3pt}0\hskip 3pt\sf rs~cd:ab;\allowbreak tu~cd:ar\allowbreak/14}\\
94)&\BL{5{\hskip 3pt}10{\hskip 3pt}5{\hskip 3pt}10{\hskip 3pt}0{\hskip 3pt}5{\hskip 3pt}1{\hskip 3pt}0{\hskip 3pt}0\hskip 3pt\sf rs~cd:ab;\allowbreak t~c:as;\allowbreak u~t:ad\allowbreak/113}\\
95)&\BL{5{\hskip 3pt}11{\hskip 3pt}0{\hskip 3pt}13{\hskip 3pt}1{\hskip 3pt}1{\hskip 3pt}1{\hskip 3pt}0{\hskip 3pt}0\hskip 3pt\sf r~a:bc;\allowbreak st~cd:abr;\allowbreak u~s:acr\allowbreak/25}\\
96)&\BL{5{\hskip 3pt}15{\hskip 3pt}0{\hskip 3pt}15{\hskip 3pt}1{\hskip 3pt}0{\hskip 3pt}0{\hskip 3pt}0{\hskip 3pt}0\hskip 3pt\sf rs~ac:bd;\allowbreak t~(ar):bd;\allowbreak u~(ar):bdt\allowbreak/67}\\
97)&\BL{6{\hskip 3pt}1{\hskip 3pt}0{\hskip 3pt}7{\hskip 3pt}7{\hskip 3pt}10{\hskip 3pt}10{\hskip 3pt}0{\hskip 3pt}0\hskip 3pt\sf rs~cd:ab;\allowbreak tu~cs:ab;\allowbreak v~(ds):abtu}\\
98)&\BL{6{\hskip 3pt}1{\hskip 3pt}0{\hskip 3pt}8{\hskip 3pt}8{\hskip 3pt}7{\hskip 3pt}7{\hskip 3pt}0{\hskip 3pt}0\hskip 3pt\sf rs~cd:ab;\allowbreak tu~cs:ab;\allowbreak v~(ds):abtu}\\
99)&\BL{6{\hskip 3pt}1{\hskip 3pt}0{\hskip 3pt}12{\hskip 3pt}10{\hskip 3pt}3{\hskip 3pt}5{\hskip 3pt}0{\hskip 3pt}0\hskip 3pt\sf rs~cd:ab;\allowbreak tu~cs:ab;\allowbreak v~b:acrtu}\\
100)&\BL{6{\hskip 3pt}1{\hskip 3pt}0{\hskip 3pt}12{\hskip 3pt}9{\hskip 3pt}4{\hskip 3pt}6{\hskip 3pt}0{\hskip 3pt}0\hskip 3pt\sf rs~cd:ab;\allowbreak tu~cs:ab;\allowbreak v~b:acrtu}\\
101)&\BL{6{\hskip 3pt}1{\hskip 3pt}0{\hskip 3pt}13{\hskip 3pt}12{\hskip 3pt}2{\hskip 3pt}3{\hskip 3pt}0{\hskip 3pt}0\hskip 3pt\sf rs~cd:ab;\allowbreak tu~cr:ab;\allowbreak v~a:bcstu}\\
102)&\BL{6{\hskip 3pt}1{\hskip 3pt}0{\hskip 3pt}14{\hskip 3pt}12{\hskip 3pt}3{\hskip 3pt}2{\hskip 3pt}0{\hskip 3pt}0\hskip 3pt\sf rs~cd:ab;\allowbreak tu~cr:ab;\allowbreak v~b:acstu}\\
103)&\BL{6{\hskip 3pt}1{\hskip 3pt}0{\hskip 3pt}14{\hskip 3pt}11{\hskip 3pt}2{\hskip 3pt}3{\hskip 3pt}0{\hskip 3pt}0\hskip 3pt\sf rs~cd:ab;\allowbreak tu~cr:ab;\allowbreak v~a:bcstu}\\
104)&\BL{6{\hskip 3pt}1{\hskip 3pt}0{\hskip 3pt}16{\hskip 3pt}16{\hskip 3pt}1{\hskip 3pt}1{\hskip 3pt}0{\hskip 3pt}0\hskip 3pt\sf rs~cd:ab;\allowbreak tu~cr:ab;\allowbreak v~(cr):abtu}\\
105)&\BL{6{\hskip 3pt}2{\hskip 3pt}0{\hskip 3pt}6{\hskip 3pt}5{\hskip 3pt}4{\hskip 3pt}5{\hskip 3pt}0{\hskip 3pt}0\hskip 3pt\sf rs~cd:ab;\allowbreak t~b:adr;\allowbreak u~t:bdr\allowbreak/132}\\
106)&\BL{6{\hskip 3pt}2{\hskip 3pt}0{\hskip 3pt}7{\hskip 3pt}6{\hskip 3pt}1{\hskip 3pt}2{\hskip 3pt}0{\hskip 3pt}0\hskip 3pt\sf rs~cd:ab;\allowbreak t~a:bcs;\allowbreak u~b:ads\allowbreak/167}\\
107)&\BL{6{\hskip 3pt}2{\hskip 3pt}0{\hskip 3pt}8{\hskip 3pt}5{\hskip 3pt}2{\hskip 3pt}1{\hskip 3pt}0{\hskip 3pt}0\hskip 3pt\sf rs~cd:ab;\allowbreak t~b:acr;\allowbreak u~b:acs\allowbreak/180}\\
108)&\BL{6{\hskip 3pt}2{\hskip 3pt}2{\hskip 3pt}18{\hskip 3pt}7{\hskip 3pt}3{\hskip 3pt}0{\hskip 3pt}0{\hskip 3pt}0\hskip 3pt\sf rs~cd:ab;\allowbreak tu~cr:ab;\allowbreak v~a:bcstu}\\
109)&\BL{6{\hskip 3pt}2{\hskip 3pt}2{\hskip 3pt}19{\hskip 3pt}7{\hskip 3pt}2{\hskip 3pt}0{\hskip 3pt}0{\hskip 3pt}0\hskip 3pt\sf rs~cd:ab;\allowbreak tu~cr:ab;\allowbreak v~a:bcstu}\\
110)&\BL{6{\hskip 3pt}3{\hskip 3pt}0{\hskip 3pt}6{\hskip 3pt}4{\hskip 3pt}3{\hskip 3pt}4{\hskip 3pt}0{\hskip 3pt}0\hskip 3pt\sf rs~cd:ab;\allowbreak t~b:acs;\allowbreak u~t:bcr\allowbreak/122}\\
111)&\BL{6{\hskip 3pt}3{\hskip 3pt}0{\hskip 3pt}8{\hskip 3pt}4{\hskip 3pt}2{\hskip 3pt}5{\hskip 3pt}0{\hskip 3pt}0\hskip 3pt\sf rs~cd:ab;\allowbreak t~b:adr;\allowbreak u~t:bcr\allowbreak/202}\\
112)&\BL{6{\hskip 3pt}3{\hskip 3pt}0{\hskip 3pt}8{\hskip 3pt}4{\hskip 3pt}2{\hskip 3pt}4{\hskip 3pt}0{\hskip 3pt}1\hskip 3pt\sf rs~cd:ab;\allowbreak t~b:acs;\allowbreak uv~at:bcr}\\
113)&\BL{6{\hskip 3pt}3{\hskip 3pt}2{\hskip 3pt}3{\hskip 3pt}0{\hskip 3pt}12{\hskip 3pt}7{\hskip 3pt}0{\hskip 3pt}0\hskip 3pt\sf rs~cd:ab;\allowbreak tu~cr:ab;\allowbreak v~(rtu):ad}\\
114)&\BL{6{\hskip 3pt}3{\hskip 3pt}2{\hskip 3pt}9{\hskip 3pt}5{\hskip 3pt}2{\hskip 3pt}0{\hskip 3pt}0{\hskip 3pt}0\hskip 3pt\sf rs~cd:ab;\allowbreak t~b:acr;\allowbreak u~r:ad\allowbreak/174}\\
115)&\BL{6{\hskip 3pt}4{\hskip 3pt}0{\hskip 3pt}8{\hskip 3pt}3{\hskip 3pt}3{\hskip 3pt}4{\hskip 3pt}0{\hskip 3pt}0\hskip 3pt\sf rs~cd:ab;\allowbreak t~b:acs;\allowbreak uv~at:bcr}\\
116)&\BL{6{\hskip 3pt}5{\hskip 3pt}0{\hskip 3pt}10{\hskip 3pt}2{\hskip 3pt}8{\hskip 3pt}5{\hskip 3pt}0{\hskip 3pt}0\hskip 3pt\sf rs~cd:ab;\allowbreak t~c:abrs;\allowbreak u~a:bdst\allowbreak/78}\\
117)&\BL{6{\hskip 3pt}5{\hskip 3pt}0{\hskip 3pt}10{\hskip 3pt}2{\hskip 3pt}9{\hskip 3pt}4{\hskip 3pt}0{\hskip 3pt}0\hskip 3pt\sf rs~cd:ab;\allowbreak t~c:abrs;\allowbreak u~a:bdst\allowbreak/61}\\
118)&\BL{6{\hskip 3pt}5{\hskip 3pt}0{\hskip 3pt}10{\hskip 3pt}2{\hskip 3pt}6{\hskip 3pt}9{\hskip 3pt}0{\hskip 3pt}0\hskip 3pt\sf rs~cd:ab;\allowbreak t~c:abs;\allowbreak u~a:bdst\allowbreak/49}\\
119)&\BL{6{\hskip 3pt}5{\hskip 3pt}2{\hskip 3pt}26{\hskip 3pt}3{\hskip 3pt}2{\hskip 3pt}0{\hskip 3pt}0{\hskip 3pt}0\hskip 3pt\sf rs~ac:bd;\allowbreak tu~ar:bd;\allowbreak v~a:bcrtu}\\
120)&\BL{6{\hskip 3pt}6{\hskip 3pt}0{\hskip 3pt}5{\hskip 3pt}5{\hskip 3pt}3{\hskip 3pt}0{\hskip 3pt}0{\hskip 3pt}0\hskip 3pt\sf rs~cd:ab;\allowbreak t~a:bcs;\allowbreak u~c:abrt\allowbreak/71}\\
121)&\BL{6{\hskip 3pt}6{\hskip 3pt}0{\hskip 3pt}6{\hskip 3pt}3{\hskip 3pt}4{\hskip 3pt}3{\hskip 3pt}0{\hskip 3pt}0\hskip 3pt\sf r~b:ac;\allowbreak st~cd:ab;\allowbreak u~b:act\allowbreak/19}\\
122)&\BL{6{\hskip 3pt}6{\hskip 3pt}5{\hskip 3pt}6{\hskip 3pt}0{\hskip 3pt}6{\hskip 3pt}3{\hskip 3pt}0{\hskip 3pt}0\hskip 3pt\sf rs~cd:ab;\allowbreak t~a:bcr;\allowbreak u~r:ad\allowbreak/188}\\
123)&\BL{6{\hskip 3pt}7{\hskip 3pt}0{\hskip 3pt}4{\hskip 3pt}4{\hskip 3pt}1{\hskip 3pt}1{\hskip 3pt}7{\hskip 3pt}0\hskip 3pt\sf rs~bd:ac;\allowbreak tu~cd:abs;\allowbreak v~b:acrst}\\
124)&\BL{6{\hskip 3pt}7{\hskip 3pt}0{\hskip 3pt}10{\hskip 3pt}2{\hskip 3pt}2{\hskip 3pt}2{\hskip 3pt}0{\hskip 3pt}0\hskip 3pt\sf rs~cd:ab;\allowbreak t~a:bcr;\allowbreak u~b:acrt\allowbreak/51}\\
125)&\BL{6{\hskip 3pt}7{\hskip 3pt}2{\hskip 3pt}5{\hskip 3pt}3{\hskip 3pt}0{\hskip 3pt}2{\hskip 3pt}10{\hskip 3pt}0\hskip 3pt\sf rs~bd:ac;\allowbreak t~c:adr;\allowbreak u~t:rs\allowbreak/95}\\
126)&\BL{6{\hskip 3pt}7{\hskip 3pt}2{\hskip 3pt}5{\hskip 3pt}0{\hskip 3pt}12{\hskip 3pt}5{\hskip 3pt}0{\hskip 3pt}0\hskip 3pt\sf r~c:ab;\allowbreak st~cd:abr;\allowbreak uv~cr:at}\\
127)&\BL{6{\hskip 3pt}8{\hskip 3pt}0{\hskip 3pt}9{\hskip 3pt}2{\hskip 3pt}2{\hskip 3pt}3{\hskip 3pt}0{\hskip 3pt}0\hskip 3pt\sf rs~cd:ab;\allowbreak t~a:bcr;\allowbreak u~(ds):abt\allowbreak/144}\\
128)&\BL{6{\hskip 3pt}8{\hskip 3pt}6{\hskip 3pt}4{\hskip 3pt}3{\hskip 3pt}11{\hskip 3pt}0{\hskip 3pt}0{\hskip 3pt}0\hskip 3pt\sf rs~cd:ab;\allowbreak t~c:as;\allowbreak u~c:bdt\allowbreak/74}\\
129)&\BL{6{\hskip 3pt}9{\hskip 3pt}0{\hskip 3pt}10{\hskip 3pt}2{\hskip 3pt}0{\hskip 3pt}0{\hskip 3pt}11{\hskip 3pt}0\hskip 3pt\sf rs~bd:ac;\allowbreak t~c:adr;\allowbreak u~a:bcrt\allowbreak/35}\\
130)&\BL{6{\hskip 3pt}10{\hskip 3pt}0{\hskip 3pt}9{\hskip 3pt}2{\hskip 3pt}0{\hskip 3pt}11{\hskip 3pt}0{\hskip 3pt}0\hskip 3pt\sf rs~cd:ab;\allowbreak t~b:adr;\allowbreak u~a:bcrt\allowbreak/43}\\
131)&\BL{6{\hskip 3pt}11{\hskip 3pt}3{\hskip 3pt}9{\hskip 3pt}2{\hskip 3pt}0{\hskip 3pt}3{\hskip 3pt}0{\hskip 3pt}0\hskip 3pt\sf rs~cd:ab;\allowbreak t~c:ar;\allowbreak u~t:bd\allowbreak/147}\\
132)&\BL{7{\hskip 3pt}1{\hskip 3pt}0{\hskip 3pt}9{\hskip 3pt}9{\hskip 3pt}13{\hskip 3pt}13{\hskip 3pt}0{\hskip 3pt}0\hskip 3pt\sf rs~cd:ab;\allowbreak t~(cs):ab;\allowbreak uv~ds:abt}\\
133)&\BL{7{\hskip 3pt}1{\hskip 3pt}0{\hskip 3pt}10{\hskip 3pt}10{\hskip 3pt}9{\hskip 3pt}9{\hskip 3pt}0{\hskip 3pt}0\hskip 3pt\sf rs~cd:ab;\allowbreak tu~cs:ab;\allowbreak v~(ds):abtu}\\
134)&\BL{7{\hskip 3pt}1{\hskip 3pt}0{\hskip 3pt}12{\hskip 3pt}12{\hskip 3pt}5{\hskip 3pt}5{\hskip 3pt}0{\hskip 3pt}0\hskip 3pt\sf rs~cd:ab;\allowbreak tu~(cr)(cs):ab\allowbreak/13}\\
135)&\BL{7{\hskip 3pt}1{\hskip 3pt}0{\hskip 3pt}15{\hskip 3pt}15{\hskip 3pt}3{\hskip 3pt}3{\hskip 3pt}0{\hskip 3pt}0\hskip 3pt\sf rs~cd:ab;\allowbreak t~(cr):ab;\allowbreak uv~cs:abt}\\
136)&\BL{7{\hskip 3pt}1{\hskip 3pt}0{\hskip 3pt}20{\hskip 3pt}20{\hskip 3pt}2{\hskip 3pt}2{\hskip 3pt}0{\hskip 3pt}0\hskip 3pt\sf rs~cd:ab;\allowbreak tu~cr:ab;\allowbreak v~(cr):abtu}\\
}\Q\te{%
137)&\BL{7{\hskip 3pt}4{\hskip 3pt}0{\hskip 3pt}6{\hskip 3pt}6{\hskip 3pt}3{\hskip 3pt}2{\hskip 3pt}0{\hskip 3pt}0\hskip 3pt\sf rs~cd:ab;\allowbreak t~a:bcs;\allowbreak u~b:acs\allowbreak/175}\\
138)&\BL{7{\hskip 3pt}6{\hskip 3pt}0{\hskip 3pt}6{\hskip 3pt}5{\hskip 3pt}1{\hskip 3pt}5{\hskip 3pt}0{\hskip 3pt}0\hskip 3pt\sf rs~cd:ab;\allowbreak t~b:adr;\allowbreak u~r:abcdt\allowbreak/133}\\
139)&\BL{7{\hskip 3pt}7{\hskip 3pt}2{\hskip 3pt}28{\hskip 3pt}3{\hskip 3pt}2{\hskip 3pt}0{\hskip 3pt}0{\hskip 3pt}0\hskip 3pt\sf rs~ac:bd;\allowbreak tu~ar:bd;\allowbreak v~a:bcrtu}\\
140)&\BL{7{\hskip 3pt}7{\hskip 3pt}3{\hskip 3pt}6{\hskip 3pt}4{\hskip 3pt}7{\hskip 3pt}0{\hskip 3pt}0{\hskip 3pt}0\hskip 3pt\sf rs~cd:ab;\allowbreak t~a:bcs;\allowbreak u~t:as\allowbreak/179}\\
141)&\BL{7{\hskip 3pt}7{\hskip 3pt}5{\hskip 3pt}6{\hskip 3pt}0{\hskip 3pt}10{\hskip 3pt}3{\hskip 3pt}0{\hskip 3pt}0\hskip 3pt\sf rs~cd:ab;\allowbreak t~c:as;\allowbreak u~r:act\allowbreak/72}\\
142)&\BL{7{\hskip 3pt}8{\hskip 3pt}1{\hskip 3pt}6{\hskip 3pt}4{\hskip 3pt}6{\hskip 3pt}0{\hskip 3pt}0{\hskip 3pt}0\hskip 3pt\sf rs~cd:ab;\allowbreak t~a:bcs;\allowbreak u~s:abcdt\allowbreak/155}\\
143)&\BL{7{\hskip 3pt}8{\hskip 3pt}5{\hskip 3pt}6{\hskip 3pt}3{\hskip 3pt}7{\hskip 3pt}0{\hskip 3pt}0{\hskip 3pt}0\hskip 3pt\sf rs~cd:ab;\allowbreak t~c:as;\allowbreak u~t:cs\allowbreak/101}\\
144)&\BL{7{\hskip 3pt}9{\hskip 3pt}0{\hskip 3pt}12{\hskip 3pt}2{\hskip 3pt}3{\hskip 3pt}4{\hskip 3pt}0{\hskip 3pt}0\hskip 3pt\sf rs~cd:ab;\allowbreak t~a:bcr;\allowbreak u~a:bdst\allowbreak/69}\\
145)&\BL{7{\hskip 3pt}9{\hskip 3pt}0{\hskip 3pt}12{\hskip 3pt}2{\hskip 3pt}4{\hskip 3pt}3{\hskip 3pt}0{\hskip 3pt}0\hskip 3pt\sf rs~cd:ab;\allowbreak t~a:bcr;\allowbreak u~b:adst\allowbreak/59}\\
146)&\BL{7{\hskip 3pt}11{\hskip 3pt}0{\hskip 3pt}21{\hskip 3pt}2{\hskip 3pt}7{\hskip 3pt}1{\hskip 3pt}0{\hskip 3pt}3\hskip 3pt\sf rs~ac:bd;\allowbreak t~a:bcr;\allowbreak u~(ar):bdt\allowbreak/192}\\
147)&\BL{7{\hskip 3pt}11{\hskip 3pt}0{\hskip 3pt}26{\hskip 3pt}2{\hskip 3pt}2{\hskip 3pt}1{\hskip 3pt}0{\hskip 3pt}0\hskip 3pt\sf rs~ac:bd;\allowbreak t~(ar):bd;\allowbreak u~a:bcrt\allowbreak/156}\\
148)&\BL{7{\hskip 3pt}12{\hskip 3pt}0{\hskip 3pt}11{\hskip 3pt}2{\hskip 3pt}1{\hskip 3pt}12{\hskip 3pt}0{\hskip 3pt}0\hskip 3pt\sf rs~cd:ab;\allowbreak t~b:acs;\allowbreak u~a:bcrt\allowbreak/106}\\
149)&\BL{7{\hskip 3pt}12{\hskip 3pt}0{\hskip 3pt}11{\hskip 3pt}2{\hskip 3pt}5{\hskip 3pt}10{\hskip 3pt}0{\hskip 3pt}0\hskip 3pt\sf rs~cd:ab;\allowbreak t~b:adr;\allowbreak u~a:bcrt\allowbreak/50}\\
150)&\BL{7{\hskip 3pt}12{\hskip 3pt}2{\hskip 3pt}8{\hskip 3pt}1{\hskip 3pt}8{\hskip 3pt}2{\hskip 3pt}0{\hskip 3pt}0\hskip 3pt\sf rs~cd:ab;\allowbreak t~c:as;\allowbreak u~d:ast\allowbreak/42}\\
151)&\BL{7{\hskip 3pt}18{\hskip 3pt}6{\hskip 3pt}18{\hskip 3pt}1{\hskip 3pt}6{\hskip 3pt}0{\hskip 3pt}0{\hskip 3pt}0\hskip 3pt\sf rs~cd:ab;\allowbreak t~c:ar;\allowbreak u~t:ad\allowbreak/97}\\
152)&\BL{7{\hskip 3pt}28{\hskip 3pt}4{\hskip 3pt}28{\hskip 3pt}1{\hskip 3pt}4{\hskip 3pt}0{\hskip 3pt}0{\hskip 3pt}0\hskip 3pt\sf rs~cd:ab;\allowbreak t~r:ac;\allowbreak u~s:at\allowbreak/104}\\
153)&\BL{8{\hskip 3pt}1{\hskip 3pt}0{\hskip 3pt}14{\hskip 3pt}14{\hskip 3pt}7{\hskip 3pt}7{\hskip 3pt}0{\hskip 3pt}0\hskip 3pt\sf rs~cd:ab;\allowbreak t~(cr):ab;\allowbreak uv~cs:abt}\\
154)&\BL{8{\hskip 3pt}1{\hskip 3pt}0{\hskip 3pt}16{\hskip 3pt}16{\hskip 3pt}5{\hskip 3pt}5{\hskip 3pt}0{\hskip 3pt}0\hskip 3pt\sf rs~cd:ab;\allowbreak t~(cr):ab;\allowbreak uv~cs:abt}\\
155)&\BL{8{\hskip 3pt}1{\hskip 3pt}0{\hskip 3pt}18{\hskip 3pt}18{\hskip 3pt}4{\hskip 3pt}4{\hskip 3pt}0{\hskip 3pt}0\hskip 3pt\sf rs~cd:ab;\allowbreak t~(cr):ab;\allowbreak uv~cs:abt}\\
156)&\BL{8{\hskip 3pt}2{\hskip 3pt}0{\hskip 3pt}14{\hskip 3pt}13{\hskip 3pt}2{\hskip 3pt}1{\hskip 3pt}0{\hskip 3pt}0\hskip 3pt\sf rs~cd:ab;\allowbreak t~b:acs;\allowbreak u~(cr):ab\allowbreak/193}\\
157)&\BL{8{\hskip 3pt}2{\hskip 3pt}0{\hskip 3pt}23{\hskip 3pt}12{\hskip 3pt}1{\hskip 3pt}4{\hskip 3pt}0{\hskip 3pt}0\hskip 3pt\sf rs~cd:ab;\allowbreak t~(cr):ab;\allowbreak u~a:bcst\allowbreak/125}\\
158)&\BL{8{\hskip 3pt}4{\hskip 3pt}0{\hskip 3pt}9{\hskip 3pt}5{\hskip 3pt}4{\hskip 3pt}7{\hskip 3pt}0{\hskip 3pt}0\hskip 3pt\sf rs~cd:ab;\allowbreak t~b:acs;\allowbreak uv~at:bcr}\\
159)&\BL{8{\hskip 3pt}4{\hskip 3pt}0{\hskip 3pt}10{\hskip 3pt}5{\hskip 3pt}4{\hskip 3pt}6{\hskip 3pt}0{\hskip 3pt}0\hskip 3pt\sf rs~cd:ab;\allowbreak t~b:acs;\allowbreak uv~at:bcr}\\
160)&\BL{8{\hskip 3pt}6{\hskip 3pt}0{\hskip 3pt}7{\hskip 3pt}6{\hskip 3pt}2{\hskip 3pt}4{\hskip 3pt}0{\hskip 3pt}0\hskip 3pt\sf rs~cd:ab;\allowbreak t~b:acs;\allowbreak u~c:abrst\allowbreak/131}\\
161)&\BL{8{\hskip 3pt}6{\hskip 3pt}0{\hskip 3pt}10{\hskip 3pt}4{\hskip 3pt}3{\hskip 3pt}4{\hskip 3pt}0{\hskip 3pt}0\hskip 3pt\sf rs~cd:ab;\allowbreak t~b:acs;\allowbreak uv~at:bcr}\\
162)&\BL{8{\hskip 3pt}6{\hskip 3pt}4{\hskip 3pt}13{\hskip 3pt}4{\hskip 3pt}8{\hskip 3pt}0{\hskip 3pt}0{\hskip 3pt}0\hskip 3pt\sf rs~cd:ab;\allowbreak t~a:bcs;\allowbreak u~a:bdrt\allowbreak/80}\\
163)&\BL{8{\hskip 3pt}9{\hskip 3pt}0{\hskip 3pt}9{\hskip 3pt}5{\hskip 3pt}0{\hskip 3pt}0{\hskip 3pt}0{\hskip 3pt}0\hskip 3pt\sf rs~ac:bd;\allowbreak t~d:abr;\allowbreak u~d:abrt\allowbreak/82}\\
164)&\BL{8{\hskip 3pt}9{\hskip 3pt}2{\hskip 3pt}7{\hskip 3pt}4{\hskip 3pt}0{\hskip 3pt}2{\hskip 3pt}12{\hskip 3pt}0\hskip 3pt\sf rs~bd:ac;\allowbreak t~c:adr;\allowbreak u~t:rs\allowbreak/115}\\
165)&\BL{8{\hskip 3pt}9{\hskip 3pt}3{\hskip 3pt}7{\hskip 3pt}0{\hskip 3pt}15{\hskip 3pt}6{\hskip 3pt}0{\hskip 3pt}0\hskip 3pt\sf r~c:ab;\allowbreak st~cd:abr;\allowbreak uv~cr:at}\\
166)&\BL{8{\hskip 3pt}9{\hskip 3pt}8{\hskip 3pt}7{\hskip 3pt}3{\hskip 3pt}8{\hskip 3pt}0{\hskip 3pt}0{\hskip 3pt}0\hskip 3pt\sf rs~cd:ab;\allowbreak t~c:as;\allowbreak u~c:bst\allowbreak/55}\\
167)&\BL{8{\hskip 3pt}11{\hskip 3pt}1{\hskip 3pt}6{\hskip 3pt}4{\hskip 3pt}14{\hskip 3pt}0{\hskip 3pt}0{\hskip 3pt}0\hskip 3pt\sf rs~cd:ab;\allowbreak t~a:bcs;\allowbreak u~a:bcst\allowbreak/40}\\
168)&\BL{8{\hskip 3pt}13{\hskip 3pt}0{\hskip 3pt}25{\hskip 3pt}2{\hskip 3pt}9{\hskip 3pt}2{\hskip 3pt}0{\hskip 3pt}1\hskip 3pt\sf rs~ac:bd;\allowbreak t~a:bcr;\allowbreak u~(ar):bdt\allowbreak/211}\\
169)&\BL{8{\hskip 3pt}13{\hskip 3pt}0{\hskip 3pt}34{\hskip 3pt}2{\hskip 3pt}3{\hskip 3pt}2{\hskip 3pt}0{\hskip 3pt}0\hskip 3pt\sf rs~ac:bd;\allowbreak t~(ar):bd;\allowbreak u~r:abst\allowbreak/166}\\
170)&\BL{9{\hskip 3pt}1{\hskip 3pt}0{\hskip 3pt}17{\hskip 3pt}17{\hskip 3pt}8{\hskip 3pt}8{\hskip 3pt}0{\hskip 3pt}0\hskip 3pt\sf rs~cd:ab;\allowbreak t~(cr):ab;\allowbreak uv~cs:abt}\\
171)&\BL{9{\hskip 3pt}1{\hskip 3pt}0{\hskip 3pt}19{\hskip 3pt}19{\hskip 3pt}6{\hskip 3pt}6{\hskip 3pt}0{\hskip 3pt}0\hskip 3pt\sf rs~cd:ab;\allowbreak t~(cr):ab;\allowbreak uv~cs:abt}\\
172)&\BL{9{\hskip 3pt}1{\hskip 3pt}0{\hskip 3pt}22{\hskip 3pt}22{\hskip 3pt}5{\hskip 3pt}5{\hskip 3pt}0{\hskip 3pt}0\hskip 3pt\sf rs~cd:ab;\allowbreak t~(cr):ab;\allowbreak uv~cs:abt}\\
173)&\BL{9{\hskip 3pt}3{\hskip 3pt}0{\hskip 3pt}13{\hskip 3pt}8{\hskip 3pt}2{\hskip 3pt}9{\hskip 3pt}0{\hskip 3pt}0\hskip 3pt\sf rs~cd:ab;\allowbreak t~b:acs;\allowbreak u~b:acrt\allowbreak/128}\\
174)&\BL{9{\hskip 3pt}3{\hskip 3pt}0{\hskip 3pt}16{\hskip 3pt}7{\hskip 3pt}11{\hskip 3pt}12{\hskip 3pt}0{\hskip 3pt}0\hskip 3pt\sf rs~cd:ab;\allowbreak tu~cs:ab;\allowbreak v~a:bcrtu}\\
175)&\BL{9{\hskip 3pt}3{\hskip 3pt}0{\hskip 3pt}16{\hskip 3pt}7{\hskip 3pt}13{\hskip 3pt}10{\hskip 3pt}0{\hskip 3pt}0\hskip 3pt\sf rs~cd:ab;\allowbreak tu~cs:ab;\allowbreak v~a:bcrtu}\\
176)&\BL{9{\hskip 3pt}4{\hskip 3pt}3{\hskip 3pt}10{\hskip 3pt}0{\hskip 3pt}22{\hskip 3pt}10{\hskip 3pt}0{\hskip 3pt}0\hskip 3pt\sf rs~cd:ab;\allowbreak t~(cr):ab;\allowbreak u~a:bdst\allowbreak/119}\\
177)&\BL{9{\hskip 3pt}6{\hskip 3pt}0{\hskip 3pt}11{\hskip 3pt}5{\hskip 3pt}3{\hskip 3pt}7{\hskip 3pt}0{\hskip 3pt}0\hskip 3pt\sf rs~cd:ab;\allowbreak t~b:acs;\allowbreak uv~at:bcr}\\
178)&\BL{9{\hskip 3pt}6{\hskip 3pt}0{\hskip 3pt}12{\hskip 3pt}5{\hskip 3pt}3{\hskip 3pt}5{\hskip 3pt}0{\hskip 3pt}0\hskip 3pt\sf rs~cd:ab;\allowbreak t~b:acs;\allowbreak uv~at:bcr}\\
179)&\BL{9{\hskip 3pt}6{\hskip 3pt}6{\hskip 3pt}12{\hskip 3pt}5{\hskip 3pt}6{\hskip 3pt}0{\hskip 3pt}0{\hskip 3pt}0\hskip 3pt\sf rs~cd:ab;\allowbreak t~(cr):ab;\allowbreak uv~rt:ad}\\
180)&\BL{9{\hskip 3pt}7{\hskip 3pt}0{\hskip 3pt}13{\hskip 3pt}4{\hskip 3pt}4{\hskip 3pt}5{\hskip 3pt}0{\hskip 3pt}0\hskip 3pt\sf rs~cd:ab;\allowbreak t~b:acs;\allowbreak uv~at:bcr}\\
181)&\BL{9{\hskip 3pt}7{\hskip 3pt}6{\hskip 3pt}14{\hskip 3pt}4{\hskip 3pt}9{\hskip 3pt}0{\hskip 3pt}0{\hskip 3pt}0\hskip 3pt\sf rs~cd:ab;\allowbreak t~a:bcs;\allowbreak u~(rt):ad\allowbreak/99}\\
182)&\BL{9{\hskip 3pt}8{\hskip 3pt}0{\hskip 3pt}15{\hskip 3pt}3{\hskip 3pt}15{\hskip 3pt}5{\hskip 3pt}0{\hskip 3pt}0\hskip 3pt\sf rs~cd:ab;\allowbreak t~c:abrs;\allowbreak u~a:bdst\allowbreak/57}\\
183)&\BL{9{\hskip 3pt}8{\hskip 3pt}2{\hskip 3pt}9{\hskip 3pt}6{\hskip 3pt}4{\hskip 3pt}0{\hskip 3pt}0{\hskip 3pt}0\hskip 3pt\sf rs~cd:ab;\allowbreak t~a:bcs;\allowbreak u~c:abrst\allowbreak/139}\\
184)&\BL{9{\hskip 3pt}9{\hskip 3pt}0{\hskip 3pt}11{\hskip 3pt}5{\hskip 3pt}0{\hskip 3pt}0{\hskip 3pt}0{\hskip 3pt}3\hskip 3pt\sf rs~ad:bc;\allowbreak t~c:abs;\allowbreak uv~bt:acr}\\
185)&\BL{9{\hskip 3pt}9{\hskip 3pt}3{\hskip 3pt}7{\hskip 3pt}6{\hskip 3pt}9{\hskip 3pt}0{\hskip 3pt}0{\hskip 3pt}0\hskip 3pt\sf rs~cd:ab;\allowbreak t~a:bcs;\allowbreak u~t:ac\allowbreak/177}\\
186)&\BL{9{\hskip 3pt}9{\hskip 3pt}4{\hskip 3pt}8{\hskip 3pt}0{\hskip 3pt}15{\hskip 3pt}6{\hskip 3pt}0{\hskip 3pt}0\hskip 3pt\sf r~c:ab;\allowbreak st~cd:abr;\allowbreak uv~cr:at}\\
187)&\BL{9{\hskip 3pt}10{\hskip 3pt}3{\hskip 3pt}7{\hskip 3pt}5{\hskip 3pt}8{\hskip 3pt}0{\hskip 3pt}0{\hskip 3pt}0\hskip 3pt\sf rs~cd:ab;\allowbreak t~a:bcs;\allowbreak u~t:as\allowbreak/152}\\
188)&\BL{9{\hskip 3pt}10{\hskip 3pt}6{\hskip 3pt}8{\hskip 3pt}4{\hskip 3pt}7{\hskip 3pt}0{\hskip 3pt}0{\hskip 3pt}0\hskip 3pt\sf rs~cd:ab;\allowbreak t~c:as;\allowbreak u~s:act\allowbreak/65}\\
189)&\BL{9{\hskip 3pt}15{\hskip 3pt}0{\hskip 3pt}31{\hskip 3pt}2{\hskip 3pt}12{\hskip 3pt}3{\hskip 3pt}0{\hskip 3pt}0\hskip 3pt\sf rs~ac:bd;\allowbreak t~a:bcr;\allowbreak u~(ar):bdt\allowbreak/181}\\
190)&\BL{9{\hskip 3pt}15{\hskip 3pt}0{\hskip 3pt}43{\hskip 3pt}2{\hskip 3pt}4{\hskip 3pt}3{\hskip 3pt}0{\hskip 3pt}0\hskip 3pt\sf rs~ac:bd;\allowbreak t~(ar):bd;\allowbreak u~r:abst\allowbreak/182}\\
191)&\BL{9{\hskip 3pt}16{\hskip 3pt}2{\hskip 3pt}12{\hskip 3pt}3{\hskip 3pt}4{\hskip 3pt}0{\hskip 3pt}0{\hskip 3pt}0\hskip 3pt\sf rs~cd:ab;\allowbreak t~c:as;\allowbreak u~d:ast\allowbreak/172}\\
192)&\BL{10{\hskip 3pt}1{\hskip 3pt}0{\hskip 3pt}20{\hskip 3pt}20{\hskip 3pt}10{\hskip 3pt}10{\hskip 3pt}0{\hskip 3pt}0\hskip 3pt\sf rs~cd:ab;\allowbreak t~(cr):ab;\allowbreak uv~cs:abt}\\
193)&\BL{10{\hskip 3pt}1{\hskip 3pt}0{\hskip 3pt}23{\hskip 3pt}23{\hskip 3pt}7{\hskip 3pt}7{\hskip 3pt}0{\hskip 3pt}0\hskip 3pt\sf rs~cd:ab;\allowbreak t~(cr):ab;\allowbreak uv~cs:abt}\\
194)&\BL{10{\hskip 3pt}2{\hskip 3pt}0{\hskip 3pt}18{\hskip 3pt}17{\hskip 3pt}4{\hskip 3pt}3{\hskip 3pt}0{\hskip 3pt}0\hskip 3pt\sf rs~cd:ab;\allowbreak t~b:adr;\allowbreak u~(cr):ab\allowbreak/203}\\
195)&\BL{10{\hskip 3pt}2{\hskip 3pt}0{\hskip 3pt}23{\hskip 3pt}16{\hskip 3pt}3{\hskip 3pt}6{\hskip 3pt}0{\hskip 3pt}0\hskip 3pt\sf rs~cd:ab;\allowbreak tu~cr:ab;\allowbreak v~a:bcstu}\\
}\Q\tf{%
196)&\BL{10{\hskip 3pt}5{\hskip 3pt}0{\hskip 3pt}11{\hskip 3pt}7{\hskip 3pt}4{\hskip 3pt}5{\hskip 3pt}0{\hskip 3pt}1\hskip 3pt\sf rs~cd:ab;\allowbreak t~b:acs;\allowbreak uv~at:bcr}\\
197)&\BL{10{\hskip 3pt}6{\hskip 3pt}0{\hskip 3pt}11{\hskip 3pt}10{\hskip 3pt}5{\hskip 3pt}0{\hskip 3pt}0{\hskip 3pt}0\hskip 3pt\sf rs~cd:ab;\allowbreak t~a:bcs;\allowbreak uv~bt:acr}\\
198)&\BL{10{\hskip 3pt}6{\hskip 3pt}0{\hskip 3pt}12{\hskip 3pt}9{\hskip 3pt}0{\hskip 3pt}5{\hskip 3pt}0{\hskip 3pt}0\hskip 3pt\sf rs~cd:ab;\allowbreak t~b:acs;\allowbreak uv~at:bcr}\\
199)&\BL{10{\hskip 3pt}6{\hskip 3pt}0{\hskip 3pt}13{\hskip 3pt}6{\hskip 3pt}3{\hskip 3pt}9{\hskip 3pt}0{\hskip 3pt}0\hskip 3pt\sf rs~cd:ab;\allowbreak t~b:acs;\allowbreak uv~at:bcr}\\
200)&\BL{10{\hskip 3pt}7{\hskip 3pt}0{\hskip 3pt}9{\hskip 3pt}8{\hskip 3pt}2{\hskip 3pt}4{\hskip 3pt}0{\hskip 3pt}0\hskip 3pt\sf rs~cd:ab;\allowbreak t~b:acs;\allowbreak u~b:acst\allowbreak/98}\\
201)&\BL{10{\hskip 3pt}7{\hskip 3pt}0{\hskip 3pt}11{\hskip 3pt}11{\hskip 3pt}0{\hskip 3pt}0{\hskip 3pt}0{\hskip 3pt}0\hskip 3pt\sf rs~cd:ab;\allowbreak t~a:bcr;\allowbreak u~a:bcrt\allowbreak/90}\\
202)&\BL{10{\hskip 3pt}7{\hskip 3pt}0{\hskip 3pt}14{\hskip 3pt}5{\hskip 3pt}4{\hskip 3pt}8{\hskip 3pt}0{\hskip 3pt}0\hskip 3pt\sf rs~cd:ab;\allowbreak t~b:acs;\allowbreak uv~at:bcr}\\
203)&\BL{10{\hskip 3pt}8{\hskip 3pt}3{\hskip 3pt}11{\hskip 3pt}7{\hskip 3pt}3{\hskip 3pt}0{\hskip 3pt}0{\hskip 3pt}0\hskip 3pt\sf rs~cd:ab;\allowbreak t~a:bcs;\allowbreak u~s:abcdt\allowbreak/75}\\
204)&\BL{10{\hskip 3pt}10{\hskip 3pt}0{\hskip 3pt}11{\hskip 3pt}6{\hskip 3pt}0{\hskip 3pt}0{\hskip 3pt}0{\hskip 3pt}5\hskip 3pt\sf rs~ad:bc;\allowbreak t~c:abs;\allowbreak uv~bt:acr}\\
205)&\BL{10{\hskip 3pt}10{\hskip 3pt}1{\hskip 3pt}9{\hskip 3pt}7{\hskip 3pt}5{\hskip 3pt}0{\hskip 3pt}0{\hskip 3pt}0\hskip 3pt\sf rs~cd:ab;\allowbreak t~a:bcs;\allowbreak u~c:abrt\allowbreak/165}\\
206)&\BL{10{\hskip 3pt}10{\hskip 3pt}2{\hskip 3pt}9{\hskip 3pt}6{\hskip 3pt}7{\hskip 3pt}0{\hskip 3pt}0{\hskip 3pt}0\hskip 3pt\sf rs~cd:ab;\allowbreak t~a:bcs;\allowbreak u~c:abrt\allowbreak/81}\\
207)&\BL{10{\hskip 3pt}11{\hskip 3pt}8{\hskip 3pt}7{\hskip 3pt}0{\hskip 3pt}12{\hskip 3pt}4{\hskip 3pt}0{\hskip 3pt}0\hskip 3pt\sf rs~cd:ab;\allowbreak t~r:ad;\allowbreak u~r:adt\allowbreak/138}\\
208)&\BL{10{\hskip 3pt}12{\hskip 3pt}0{\hskip 3pt}9{\hskip 3pt}6{\hskip 3pt}0{\hskip 3pt}0{\hskip 3pt}0{\hskip 3pt}5\hskip 3pt\sf rs~ad:bc;\allowbreak t~c:abs;\allowbreak uv~bt:acr}\\
209)&\BL{10{\hskip 3pt}13{\hskip 3pt}0{\hskip 3pt}16{\hskip 3pt}3{\hskip 3pt}5{\hskip 3pt}5{\hskip 3pt}0{\hskip 3pt}0\hskip 3pt\sf rs~cd:ab;\allowbreak t~a:bcr;\allowbreak u~d:abst\allowbreak/63}\\
210)&\BL{10{\hskip 3pt}16{\hskip 3pt}0{\hskip 3pt}16{\hskip 3pt}3{\hskip 3pt}1{\hskip 3pt}18{\hskip 3pt}0{\hskip 3pt}0\hskip 3pt\sf rs~cd:ab;\allowbreak t~b:adr;\allowbreak u~a:bcrt\allowbreak/118}\\
211)&\BL{11{\hskip 3pt}6{\hskip 3pt}1{\hskip 3pt}12{\hskip 3pt}12{\hskip 3pt}4{\hskip 3pt}0{\hskip 3pt}0{\hskip 3pt}0\hskip 3pt\sf rs~cd:ab;\allowbreak t~a:bcs;\allowbreak uv~bt:acr}\\
212)&\BL{11{\hskip 3pt}6{\hskip 3pt}1{\hskip 3pt}14{\hskip 3pt}10{\hskip 3pt}4{\hskip 3pt}0{\hskip 3pt}0{\hskip 3pt}0\hskip 3pt\sf rs~cd:ab;\allowbreak t~a:bcs;\allowbreak uv~bt:acr}\\
213)&\BL{11{\hskip 3pt}7{\hskip 3pt}0{\hskip 3pt}12{\hskip 3pt}11{\hskip 3pt}0{\hskip 3pt}4{\hskip 3pt}0{\hskip 3pt}0\hskip 3pt\sf rs~cd:ab;\allowbreak t~b:acs;\allowbreak uv~at:bcr}\\
214)&\BL{11{\hskip 3pt}7{\hskip 3pt}0{\hskip 3pt}12{\hskip 3pt}11{\hskip 3pt}4{\hskip 3pt}0{\hskip 3pt}0{\hskip 3pt}0\hskip 3pt\sf rs~cd:ab;\allowbreak t~a:bcs;\allowbreak uv~bt:acr}\\
215)&\BL{11{\hskip 3pt}7{\hskip 3pt}0{\hskip 3pt}13{\hskip 3pt}6{\hskip 3pt}5{\hskip 3pt}7{\hskip 3pt}0{\hskip 3pt}0\hskip 3pt\sf rs~cd:ab;\allowbreak t~b:acs;\allowbreak uv~at:bcr}\\
216)&\BL{11{\hskip 3pt}7{\hskip 3pt}0{\hskip 3pt}16{\hskip 3pt}6{\hskip 3pt}4{\hskip 3pt}10{\hskip 3pt}0{\hskip 3pt}0\hskip 3pt\sf rs~cd:ab;\allowbreak t~b:acs;\allowbreak uv~at:bcr}\\
217)&\BL{11{\hskip 3pt}7{\hskip 3pt}0{\hskip 3pt}27{\hskip 3pt}5{\hskip 3pt}10{\hskip 3pt}10{\hskip 3pt}0{\hskip 3pt}0\hskip 3pt\sf rs~cd:ab;\allowbreak tu~cs:ab;\allowbreak v~a:bcrtu}\\
218)&\BL{11{\hskip 3pt}8{\hskip 3pt}0{\hskip 3pt}14{\hskip 3pt}5{\hskip 3pt}6{\hskip 3pt}7{\hskip 3pt}0{\hskip 3pt}0\hskip 3pt\sf rs~cd:ab;\allowbreak t~b:acs;\allowbreak uv~at:bcr}\\
219)&\BL{11{\hskip 3pt}9{\hskip 3pt}0{\hskip 3pt}12{\hskip 3pt}7{\hskip 3pt}0{\hskip 3pt}0{\hskip 3pt}13{\hskip 3pt}0\hskip 3pt\sf rs~bd:ac;\allowbreak t~c:abs;\allowbreak u~b:acrst;\allowbreak v~b:acrstu}\\
220)&\BL{11{\hskip 3pt}12{\hskip 3pt}0{\hskip 3pt}28{\hskip 3pt}3{\hskip 3pt}7{\hskip 3pt}5{\hskip 3pt}0{\hskip 3pt}0\hskip 3pt\sf r~a:bc;\allowbreak st~cd:abr;\allowbreak u~b:acrt\allowbreak/17}\\
221)&\BL{11{\hskip 3pt}18{\hskip 3pt}0{\hskip 3pt}19{\hskip 3pt}3{\hskip 3pt}2{\hskip 3pt}21{\hskip 3pt}0{\hskip 3pt}0\hskip 3pt\sf rs~cd:ab;\allowbreak t~b:acs;\allowbreak u~a:bcrt\allowbreak/127}\\
222)&\BL{11{\hskip 3pt}18{\hskip 3pt}4{\hskip 3pt}16{\hskip 3pt}3{\hskip 3pt}6{\hskip 3pt}0{\hskip 3pt}0{\hskip 3pt}0\hskip 3pt\sf rs~cd:ab;\allowbreak t~r:ad;\allowbreak u~s:adt\allowbreak/140}\\
223)&\BL{11{\hskip 3pt}20{\hskip 3pt}2{\hskip 3pt}18{\hskip 3pt}3{\hskip 3pt}4{\hskip 3pt}0{\hskip 3pt}0{\hskip 3pt}0\hskip 3pt\sf rs~cd:ab;\allowbreak t~c:ar;\allowbreak u~d:art\allowbreak/36}\\
224)&\BL{12{\hskip 3pt}4{\hskip 3pt}0{\hskip 3pt}25{\hskip 3pt}18{\hskip 3pt}3{\hskip 3pt}0{\hskip 3pt}0{\hskip 3pt}0\hskip 3pt\sf rs~cd:ab;\allowbreak t~(cr):ab;\allowbreak u~t:acs\allowbreak/134}\\
225)&\BL{12{\hskip 3pt}7{\hskip 3pt}0{\hskip 3pt}14{\hskip 3pt}7{\hskip 3pt}5{\hskip 3pt}11{\hskip 3pt}0{\hskip 3pt}0\hskip 3pt\sf rs~cd:ab;\allowbreak t~b:acs;\allowbreak uv~at:bcr}\\
226)&\BL{12{\hskip 3pt}8{\hskip 3pt}0{\hskip 3pt}31{\hskip 3pt}5{\hskip 3pt}13{\hskip 3pt}10{\hskip 3pt}0{\hskip 3pt}0\hskip 3pt\sf rs~cd:ab;\allowbreak tu~cs:ab;\allowbreak v~a:bcrtu}\\
227)&\BL{12{\hskip 3pt}8{\hskip 3pt}0{\hskip 3pt}31{\hskip 3pt}5{\hskip 3pt}12{\hskip 3pt}12{\hskip 3pt}0{\hskip 3pt}0\hskip 3pt\sf rs~cd:ab;\allowbreak tu~cs:ab;\allowbreak v~a:bcrtu}\\
228)&\BL{12{\hskip 3pt}9{\hskip 3pt}0{\hskip 3pt}29{\hskip 3pt}4{\hskip 3pt}18{\hskip 3pt}16{\hskip 3pt}0{\hskip 3pt}0\hskip 3pt\sf rs~cd:ab;\allowbreak tu~cs:ab;\allowbreak v~a:bcrtu}\\
229)&\BL{12{\hskip 3pt}9{\hskip 3pt}0{\hskip 3pt}29{\hskip 3pt}4{\hskip 3pt}14{\hskip 3pt}23{\hskip 3pt}0{\hskip 3pt}0\hskip 3pt\sf rs~cd:ab;\allowbreak tu~cs:ab;\allowbreak v~a:bcrtu}\\
230)&\BL{12{\hskip 3pt}9{\hskip 3pt}0{\hskip 3pt}29{\hskip 3pt}4{\hskip 3pt}17{\hskip 3pt}17{\hskip 3pt}0{\hskip 3pt}0\hskip 3pt\sf rs~cd:ab;\allowbreak tu~cs:ab;\allowbreak v~a:bcrtu}\\
231)&\BL{12{\hskip 3pt}9{\hskip 3pt}0{\hskip 3pt}34{\hskip 3pt}4{\hskip 3pt}18{\hskip 3pt}11{\hskip 3pt}0{\hskip 3pt}0\hskip 3pt\sf rs~cd:ab;\allowbreak tu~cs:ab;\allowbreak v~a:bcrtu}\\
232)&\BL{12{\hskip 3pt}9{\hskip 3pt}0{\hskip 3pt}34{\hskip 3pt}4{\hskip 3pt}17{\hskip 3pt}12{\hskip 3pt}0{\hskip 3pt}0\hskip 3pt\sf rs~cd:ab;\allowbreak tu~cs:ab;\allowbreak v~a:bcrtu}\\
233)&\BL{12{\hskip 3pt}9{\hskip 3pt}0{\hskip 3pt}34{\hskip 3pt}4{\hskip 3pt}14{\hskip 3pt}18{\hskip 3pt}0{\hskip 3pt}0\hskip 3pt\sf rs~cd:ab;\allowbreak tu~cs:ab;\allowbreak v~a:bcrtu}\\
234)&\BL{12{\hskip 3pt}9{\hskip 3pt}0{\hskip 3pt}36{\hskip 3pt}4{\hskip 3pt}18{\hskip 3pt}10{\hskip 3pt}0{\hskip 3pt}0\hskip 3pt\sf rs~cd:ab;\allowbreak tu~cs:ab;\allowbreak v~a:bcrtu}\\
235)&\BL{12{\hskip 3pt}9{\hskip 3pt}0{\hskip 3pt}36{\hskip 3pt}4{\hskip 3pt}14{\hskip 3pt}17{\hskip 3pt}0{\hskip 3pt}0\hskip 3pt\sf rs~cd:ab;\allowbreak tu~cs:ab;\allowbreak v~a:bcrtu}\\
236)&\BL{12{\hskip 3pt}10{\hskip 3pt}5{\hskip 3pt}8{\hskip 3pt}0{\hskip 3pt}21{\hskip 3pt}9{\hskip 3pt}0{\hskip 3pt}0\hskip 3pt\sf rs~cd:ab;\allowbreak t~(cr):ab;\allowbreak uv~rt:ad}\\
237)&\BL{12{\hskip 3pt}11{\hskip 3pt}10{\hskip 3pt}9{\hskip 3pt}0{\hskip 3pt}14{\hskip 3pt}5{\hskip 3pt}0{\hskip 3pt}0\hskip 3pt\sf rs~cd:ab;\allowbreak t~r:ad;\allowbreak u~a:bdrt\allowbreak/76}\\
238)&\BL{12{\hskip 3pt}12{\hskip 3pt}9{\hskip 3pt}12{\hskip 3pt}0{\hskip 3pt}15{\hskip 3pt}5{\hskip 3pt}0{\hskip 3pt}0\hskip 3pt\sf rs~cd:ab;\allowbreak t~(cr):ab;\allowbreak uv~rt:ad}\\
239)&\BL{12{\hskip 3pt}13{\hskip 3pt}0{\hskip 3pt}9{\hskip 3pt}9{\hskip 3pt}0{\hskip 3pt}0{\hskip 3pt}17{\hskip 3pt}0\hskip 3pt\sf rs~bd:ac;\allowbreak t~c:abs;\allowbreak u~c:abrt;\allowbreak v~c:abrtu}\\
240)&\BL{13{\hskip 3pt}6{\hskip 3pt}4{\hskip 3pt}50{\hskip 3pt}11{\hskip 3pt}9{\hskip 3pt}0{\hskip 3pt}0{\hskip 3pt}0\hskip 3pt\sf rs~cd:ab;\allowbreak tu~cr:ab;\allowbreak v~a:bcstu}\\
241)&\BL{13{\hskip 3pt}6{\hskip 3pt}4{\hskip 3pt}53{\hskip 3pt}11{\hskip 3pt}6{\hskip 3pt}0{\hskip 3pt}0{\hskip 3pt}0\hskip 3pt\sf rs~cd:ab;\allowbreak tu~cr:ab;\allowbreak v~a:bcstu}\\
242)&\BL{13{\hskip 3pt}8{\hskip 3pt}0{\hskip 3pt}15{\hskip 3pt}7{\hskip 3pt}7{\hskip 3pt}9{\hskip 3pt}0{\hskip 3pt}0\hskip 3pt\sf rs~cd:ab;\allowbreak t~b:acs;\allowbreak uv~at:bcr}\\
243)&\BL{13{\hskip 3pt}8{\hskip 3pt}0{\hskip 3pt}16{\hskip 3pt}7{\hskip 3pt}7{\hskip 3pt}8{\hskip 3pt}0{\hskip 3pt}0\hskip 3pt\sf rs~cd:ab;\allowbreak t~b:acs;\allowbreak uv~at:bcr}\\
244)&\BL{13{\hskip 3pt}8{\hskip 3pt}0{\hskip 3pt}18{\hskip 3pt}7{\hskip 3pt}6{\hskip 3pt}10{\hskip 3pt}0{\hskip 3pt}0\hskip 3pt\sf rs~cd:ab;\allowbreak t~b:acs;\allowbreak uv~at:bcr}\\
245)&\BL{13{\hskip 3pt}8{\hskip 3pt}0{\hskip 3pt}32{\hskip 3pt}6{\hskip 3pt}12{\hskip 3pt}14{\hskip 3pt}0{\hskip 3pt}0\hskip 3pt\sf rs~cd:ab;\allowbreak tu~cs:ab;\allowbreak v~a:bcrtu}\\
246)&\BL{13{\hskip 3pt}9{\hskip 3pt}0{\hskip 3pt}37{\hskip 3pt}5{\hskip 3pt}16{\hskip 3pt}10{\hskip 3pt}0{\hskip 3pt}0\hskip 3pt\sf rs~cd:ab;\allowbreak tu~cs:ab;\allowbreak v~a:bcrtu}\\
247)&\BL{13{\hskip 3pt}9{\hskip 3pt}0{\hskip 3pt}37{\hskip 3pt}5{\hskip 3pt}14{\hskip 3pt}14{\hskip 3pt}0{\hskip 3pt}0\hskip 3pt\sf rs~cd:ab;\allowbreak tu~cs:ab;\allowbreak v~a:bcrtu}\\
248)&\BL{13{\hskip 3pt}13{\hskip 3pt}11{\hskip 3pt}10{\hskip 3pt}0{\hskip 3pt}16{\hskip 3pt}5{\hskip 3pt}0{\hskip 3pt}0\hskip 3pt\sf rs~cd:ab;\allowbreak t~c:as;\allowbreak u~s:act\allowbreak/73}\\
249)&\BL{13{\hskip 3pt}24{\hskip 3pt}4{\hskip 3pt}22{\hskip 3pt}3{\hskip 3pt}6{\hskip 3pt}0{\hskip 3pt}0{\hskip 3pt}0\hskip 3pt\sf rs~cd:ab;\allowbreak t~c:ar;\allowbreak u~d:art\allowbreak/124}\\
250)&\BL{14{\hskip 3pt}8{\hskip 3pt}0{\hskip 3pt}16{\hskip 3pt}8{\hskip 3pt}7{\hskip 3pt}11{\hskip 3pt}0{\hskip 3pt}0\hskip 3pt\sf rs~cd:ab;\allowbreak t~b:acs;\allowbreak uv~at:bcr}\\
251)&\BL{14{\hskip 3pt}8{\hskip 3pt}1{\hskip 3pt}15{\hskip 3pt}15{\hskip 3pt}5{\hskip 3pt}0{\hskip 3pt}0{\hskip 3pt}0\hskip 3pt\sf rs~cd:ab;\allowbreak t~a:bcs;\allowbreak uv~bt:acr}\\
252)&\BL{14{\hskip 3pt}8{\hskip 3pt}1{\hskip 3pt}17{\hskip 3pt}13{\hskip 3pt}5{\hskip 3pt}0{\hskip 3pt}0{\hskip 3pt}0\hskip 3pt\sf rs~cd:ab;\allowbreak t~a:bcs;\allowbreak uv~bt:acr}\\
}\Q\tg{%
253)&\BL{14{\hskip 3pt}8{\hskip 3pt}2{\hskip 3pt}10{\hskip 3pt}5{\hskip 3pt}20{\hskip 3pt}7{\hskip 3pt}0{\hskip 3pt}0\hskip 3pt\sf rs~cd:ab;\allowbreak t~r:ad;\allowbreak u~b:adrt\allowbreak/205}\\
254)&\BL{14{\hskip 3pt}9{\hskip 3pt}0{\hskip 3pt}17{\hskip 3pt}8{\hskip 3pt}5{\hskip 3pt}8{\hskip 3pt}0{\hskip 3pt}1\hskip 3pt\sf rs~cd:ab;\allowbreak t~b:acs;\allowbreak uv~at:bcr}\\
255)&\BL{14{\hskip 3pt}9{\hskip 3pt}0{\hskip 3pt}34{\hskip 3pt}6{\hskip 3pt}14{\hskip 3pt}16{\hskip 3pt}0{\hskip 3pt}0\hskip 3pt\sf rs~cd:ab;\allowbreak tu~cs:ab;\allowbreak v~a:bcrtu}\\
256)&\BL{14{\hskip 3pt}10{\hskip 3pt}0{\hskip 3pt}17{\hskip 3pt}7{\hskip 3pt}6{\hskip 3pt}9{\hskip 3pt}0{\hskip 3pt}0\hskip 3pt\sf rs~cd:ab;\allowbreak t~b:acs;\allowbreak uv~at:bcr}\\
257)&\BL{14{\hskip 3pt}10{\hskip 3pt}0{\hskip 3pt}18{\hskip 3pt}7{\hskip 3pt}6{\hskip 3pt}7{\hskip 3pt}0{\hskip 3pt}0\hskip 3pt\sf rs~cd:ab;\allowbreak t~b:acs;\allowbreak uv~at:bcr}\\
258)&\BL{14{\hskip 3pt}10{\hskip 3pt}0{\hskip 3pt}43{\hskip 3pt}5{\hskip 3pt}19{\hskip 3pt}11{\hskip 3pt}0{\hskip 3pt}0\hskip 3pt\sf rs~cd:ab;\allowbreak tu~cs:ab;\allowbreak v~a:bcrtu}\\
259)&\BL{14{\hskip 3pt}10{\hskip 3pt}0{\hskip 3pt}43{\hskip 3pt}5{\hskip 3pt}16{\hskip 3pt}17{\hskip 3pt}0{\hskip 3pt}0\hskip 3pt\sf rs~cd:ab;\allowbreak tu~cs:ab;\allowbreak v~a:bcrtu}\\
260)&\BL{14{\hskip 3pt}16{\hskip 3pt}7{\hskip 3pt}14{\hskip 3pt}0{\hskip 3pt}22{\hskip 3pt}8{\hskip 3pt}0{\hskip 3pt}0\hskip 3pt\sf r~c:ab;\allowbreak st~cd:abr;\allowbreak uv~cr:at}\\
261)&\BL{14{\hskip 3pt}22{\hskip 3pt}0{\hskip 3pt}45{\hskip 3pt}4{\hskip 3pt}17{\hskip 3pt}2{\hskip 3pt}0{\hskip 3pt}0\hskip 3pt\sf rs~ac:bd;\allowbreak t~r:abs;\allowbreak u~(ar):bdt\allowbreak/204}\\
262)&\BL{15{\hskip 3pt}9{\hskip 3pt}0{\hskip 3pt}18{\hskip 3pt}14{\hskip 3pt}0{\hskip 3pt}6{\hskip 3pt}0{\hskip 3pt}0\hskip 3pt\sf rs~cd:ab;\allowbreak t~b:acs;\allowbreak uv~at:bcr}\\
263)&\BL{15{\hskip 3pt}10{\hskip 3pt}0{\hskip 3pt}15{\hskip 3pt}12{\hskip 3pt}3{\hskip 3pt}4{\hskip 3pt}0{\hskip 3pt}0\hskip 3pt\sf rs~cd:ab;\allowbreak t~b:acs;\allowbreak u~b:acst\allowbreak/210}\\
264)&\BL{15{\hskip 3pt}10{\hskip 3pt}0{\hskip 3pt}41{\hskip 3pt}6{\hskip 3pt}21{\hskip 3pt}11{\hskip 3pt}0{\hskip 3pt}0\hskip 3pt\sf rs~cd:ab;\allowbreak tu~cs:ab;\allowbreak v~a:bcrtu}\\
265)&\BL{15{\hskip 3pt}10{\hskip 3pt}0{\hskip 3pt}42{\hskip 3pt}6{\hskip 3pt}24{\hskip 3pt}10{\hskip 3pt}0{\hskip 3pt}0\hskip 3pt\sf rs~cd:ab;\allowbreak tu~cs:ab;\allowbreak v~a:bcrtu}\\
266)&\BL{15{\hskip 3pt}11{\hskip 3pt}0{\hskip 3pt}17{\hskip 3pt}9{\hskip 3pt}3{\hskip 3pt}10{\hskip 3pt}0{\hskip 3pt}0\hskip 3pt\sf rs~cd:ab;\allowbreak t~b:acs;\allowbreak uv~at:bcr}\\
267)&\BL{15{\hskip 3pt}11{\hskip 3pt}0{\hskip 3pt}20{\hskip 3pt}7{\hskip 3pt}7{\hskip 3pt}10{\hskip 3pt}0{\hskip 3pt}0\hskip 3pt\sf rs~cd:ab;\allowbreak t~b:acs;\allowbreak uv~at:bcr}\\
268)&\BL{16{\hskip 3pt}10{\hskip 3pt}0{\hskip 3pt}18{\hskip 3pt}9{\hskip 3pt}7{\hskip 3pt}11{\hskip 3pt}0{\hskip 3pt}0\hskip 3pt\sf rs~cd:ab;\allowbreak t~b:acs;\allowbreak uv~at:bcr}\\
269)&\BL{16{\hskip 3pt}10{\hskip 3pt}0{\hskip 3pt}20{\hskip 3pt}9{\hskip 3pt}6{\hskip 3pt}14{\hskip 3pt}0{\hskip 3pt}0\hskip 3pt\sf rs~cd:ab;\allowbreak t~b:acs;\allowbreak uv~at:bcr}\\
270)&\BL{16{\hskip 3pt}11{\hskip 3pt}0{\hskip 3pt}18{\hskip 3pt}10{\hskip 3pt}4{\hskip 3pt}9{\hskip 3pt}0{\hskip 3pt}0\hskip 3pt\sf rs~cd:ab;\allowbreak t~b:acs;\allowbreak uv~at:bcr}\\
271)&\BL{16{\hskip 3pt}11{\hskip 3pt}0{\hskip 3pt}38{\hskip 3pt}6{\hskip 3pt}18{\hskip 3pt}27{\hskip 3pt}0{\hskip 3pt}0\hskip 3pt\sf rs~cd:ab;\allowbreak tu~cs:ab;\allowbreak v~a:bcrtu}\\
272)&\BL{16{\hskip 3pt}11{\hskip 3pt}0{\hskip 3pt}44{\hskip 3pt}6{\hskip 3pt}18{\hskip 3pt}21{\hskip 3pt}0{\hskip 3pt}0\hskip 3pt\sf rs~cd:ab;\allowbreak tu~cs:ab;\allowbreak v~a:bcrtu}\\
273)&\BL{16{\hskip 3pt}11{\hskip 3pt}0{\hskip 3pt}44{\hskip 3pt}6{\hskip 3pt}21{\hskip 3pt}15{\hskip 3pt}0{\hskip 3pt}0\hskip 3pt\sf rs~cd:ab;\allowbreak tu~cs:ab;\allowbreak v~a:bcrtu}\\
274)&\BL{16{\hskip 3pt}11{\hskip 3pt}0{\hskip 3pt}46{\hskip 3pt}6{\hskip 3pt}18{\hskip 3pt}20{\hskip 3pt}0{\hskip 3pt}0\hskip 3pt\sf rs~cd:ab;\allowbreak tu~cs:ab;\allowbreak v~a:bcrtu}\\
275)&\BL{16{\hskip 3pt}11{\hskip 3pt}0{\hskip 3pt}46{\hskip 3pt}6{\hskip 3pt}21{\hskip 3pt}14{\hskip 3pt}0{\hskip 3pt}0\hskip 3pt\sf rs~cd:ab;\allowbreak tu~cs:ab;\allowbreak v~a:bcrtu}\\
276)&\BL{16{\hskip 3pt}11{\hskip 3pt}0{\hskip 3pt}47{\hskip 3pt}6{\hskip 3pt}24{\hskip 3pt}12{\hskip 3pt}0{\hskip 3pt}0\hskip 3pt\sf rs~cd:ab;\allowbreak tu~cs:ab;\allowbreak v~a:bcrtu}\\
277)&\BL{16{\hskip 3pt}12{\hskip 3pt}0{\hskip 3pt}24{\hskip 3pt}7{\hskip 3pt}8{\hskip 3pt}11{\hskip 3pt}0{\hskip 3pt}0\hskip 3pt\sf rs~cd:ab;\allowbreak t~b:acs;\allowbreak uv~at:bcr}\\
278)&\BL{16{\hskip 3pt}19{\hskip 3pt}9{\hskip 3pt}16{\hskip 3pt}0{\hskip 3pt}24{\hskip 3pt}8{\hskip 3pt}0{\hskip 3pt}0\hskip 3pt\sf r~c:ab;\allowbreak st~cd:abr;\allowbreak uv~cr:at}\\
279)&\BL{16{\hskip 3pt}26{\hskip 3pt}0{\hskip 3pt}51{\hskip 3pt}4{\hskip 3pt}19{\hskip 3pt}4{\hskip 3pt}0{\hskip 3pt}0\hskip 3pt\sf rs~ac:bd;\allowbreak t~a:bcr;\allowbreak u~(ar):bdt\allowbreak/209}\\
280)&\BL{17{\hskip 3pt}10{\hskip 3pt}0{\hskip 3pt}23{\hskip 3pt}10{\hskip 3pt}6{\hskip 3pt}14{\hskip 3pt}0{\hskip 3pt}1\hskip 3pt\sf rs~cd:ab;\allowbreak t~b:acs;\allowbreak uv~at:bcr}\\
281)&\BL{17{\hskip 3pt}10{\hskip 3pt}0{\hskip 3pt}23{\hskip 3pt}10{\hskip 3pt}6{\hskip 3pt}15{\hskip 3pt}0{\hskip 3pt}0\hskip 3pt\sf rs~cd:ab;\allowbreak t~b:acs;\allowbreak uv~at:bcr}\\
282)&\BL{17{\hskip 3pt}11{\hskip 3pt}0{\hskip 3pt}20{\hskip 3pt}9{\hskip 3pt}8{\hskip 3pt}12{\hskip 3pt}0{\hskip 3pt}0\hskip 3pt\sf rs~cd:ab;\allowbreak t~b:acs;\allowbreak uv~at:bcr}\\
283)&\BL{17{\hskip 3pt}12{\hskip 3pt}0{\hskip 3pt}26{\hskip 3pt}8{\hskip 3pt}8{\hskip 3pt}13{\hskip 3pt}0{\hskip 3pt}0\hskip 3pt\sf rs~cd:ab;\allowbreak t~b:acs;\allowbreak uv~at:bcr}\\
284)&\BL{17{\hskip 3pt}12{\hskip 3pt}0{\hskip 3pt}42{\hskip 3pt}6{\hskip 3pt}20{\hskip 3pt}31{\hskip 3pt}0{\hskip 3pt}0\hskip 3pt\sf rs~cd:ab;\allowbreak tu~cs:ab;\allowbreak v~a:bcrtu}\\
285)&\BL{17{\hskip 3pt}12{\hskip 3pt}0{\hskip 3pt}42{\hskip 3pt}6{\hskip 3pt}24{\hskip 3pt}23{\hskip 3pt}0{\hskip 3pt}0\hskip 3pt\sf rs~cd:ab;\allowbreak tu~cs:ab;\allowbreak v~a:bcrtu}\\
286)&\BL{17{\hskip 3pt}12{\hskip 3pt}0{\hskip 3pt}42{\hskip 3pt}6{\hskip 3pt}25{\hskip 3pt}22{\hskip 3pt}0{\hskip 3pt}0\hskip 3pt\sf rs~cd:ab;\allowbreak tu~cs:ab;\allowbreak v~a:bcrtu}\\
287)&\BL{17{\hskip 3pt}12{\hskip 3pt}0{\hskip 3pt}48{\hskip 3pt}6{\hskip 3pt}20{\hskip 3pt}25{\hskip 3pt}0{\hskip 3pt}0\hskip 3pt\sf rs~cd:ab;\allowbreak tu~cs:ab;\allowbreak v~a:bcrtu}\\
288)&\BL{17{\hskip 3pt}12{\hskip 3pt}0{\hskip 3pt}48{\hskip 3pt}6{\hskip 3pt}24{\hskip 3pt}17{\hskip 3pt}0{\hskip 3pt}0\hskip 3pt\sf rs~cd:ab;\allowbreak tu~cs:ab;\allowbreak v~a:bcrtu}\\
289)&\BL{17{\hskip 3pt}12{\hskip 3pt}0{\hskip 3pt}48{\hskip 3pt}6{\hskip 3pt}25{\hskip 3pt}16{\hskip 3pt}0{\hskip 3pt}0\hskip 3pt\sf rs~cd:ab;\allowbreak tu~cs:ab;\allowbreak v~a:bcrtu}\\
290)&\BL{17{\hskip 3pt}12{\hskip 3pt}0{\hskip 3pt}52{\hskip 3pt}6{\hskip 3pt}25{\hskip 3pt}14{\hskip 3pt}0{\hskip 3pt}0\hskip 3pt\sf rs~cd:ab;\allowbreak tu~cs:ab;\allowbreak v~a:bcrtu}\\
291)&\BL{17{\hskip 3pt}12{\hskip 3pt}0{\hskip 3pt}52{\hskip 3pt}6{\hskip 3pt}24{\hskip 3pt}15{\hskip 3pt}0{\hskip 3pt}0\hskip 3pt\sf rs~cd:ab;\allowbreak tu~cs:ab;\allowbreak v~a:bcrtu}\\
292)&\BL{17{\hskip 3pt}12{\hskip 3pt}0{\hskip 3pt}52{\hskip 3pt}6{\hskip 3pt}20{\hskip 3pt}23{\hskip 3pt}0{\hskip 3pt}0\hskip 3pt\sf rs~cd:ab;\allowbreak tu~cs:ab;\allowbreak v~a:bcrtu}\\
293)&\BL{18{\hskip 3pt}4{\hskip 3pt}0{\hskip 3pt}51{\hskip 3pt}30{\hskip 3pt}3{\hskip 3pt}6{\hskip 3pt}0{\hskip 3pt}0\hskip 3pt\sf rs~cd:ab;\allowbreak t~(cr):ab;\allowbreak u~a:bcst\allowbreak/141}\\
294)&\BL{18{\hskip 3pt}6{\hskip 3pt}1{\hskip 3pt}47{\hskip 3pt}27{\hskip 3pt}3{\hskip 3pt}0{\hskip 3pt}0{\hskip 3pt}0\hskip 3pt\sf rs~cd:ab;\allowbreak tu~cr:ab;\allowbreak v~a:bcstu}\\
295)&\BL{18{\hskip 3pt}8{\hskip 3pt}5{\hskip 3pt}66{\hskip 3pt}16{\hskip 3pt}12{\hskip 3pt}0{\hskip 3pt}0{\hskip 3pt}0\hskip 3pt\sf rs~cd:ab;\allowbreak tu~cr:ab;\allowbreak v~a:bcstu}\\
296)&\BL{18{\hskip 3pt}8{\hskip 3pt}5{\hskip 3pt}70{\hskip 3pt}16{\hskip 3pt}8{\hskip 3pt}0{\hskip 3pt}0{\hskip 3pt}0\hskip 3pt\sf rs~cd:ab;\allowbreak tu~cr:ab;\allowbreak v~a:bcstu}\\
297)&\BL{18{\hskip 3pt}9{\hskip 3pt}3{\hskip 3pt}22{\hskip 3pt}19{\hskip 3pt}5{\hskip 3pt}0{\hskip 3pt}0{\hskip 3pt}0\hskip 3pt\sf rs~cd:ab;\allowbreak t~a:bcs;\allowbreak uv~bt:acr}\\
298)&\BL{18{\hskip 3pt}9{\hskip 3pt}3{\hskip 3pt}23{\hskip 3pt}18{\hskip 3pt}5{\hskip 3pt}0{\hskip 3pt}0{\hskip 3pt}0\hskip 3pt\sf rs~cd:ab;\allowbreak t~a:bcs;\allowbreak uv~bt:acr}\\
299)&\BL{18{\hskip 3pt}9{\hskip 3pt}4{\hskip 3pt}24{\hskip 3pt}17{\hskip 3pt}5{\hskip 3pt}0{\hskip 3pt}0{\hskip 3pt}0\hskip 3pt\sf rs~cd:ab;\allowbreak t~a:bcs;\allowbreak uv~bt:acr}\\
300)&\BL{18{\hskip 3pt}10{\hskip 3pt}0{\hskip 3pt}21{\hskip 3pt}18{\hskip 3pt}0{\hskip 3pt}8{\hskip 3pt}0{\hskip 3pt}0\hskip 3pt\sf rs~cd:ab;\allowbreak t~b:acs;\allowbreak uv~at:bcr}\\
301)&\BL{18{\hskip 3pt}10{\hskip 3pt}0{\hskip 3pt}21{\hskip 3pt}18{\hskip 3pt}8{\hskip 3pt}0{\hskip 3pt}0{\hskip 3pt}0\hskip 3pt\sf rs~cd:ab;\allowbreak t~a:bcs;\allowbreak uv~bt:acr}\\
302)&\BL{18{\hskip 3pt}11{\hskip 3pt}0{\hskip 3pt}22{\hskip 3pt}10{\hskip 3pt}8{\hskip 3pt}14{\hskip 3pt}0{\hskip 3pt}0\hskip 3pt\sf rs~cd:ab;\allowbreak t~b:acs;\allowbreak uv~at:bcr}\\
303)&\BL{18{\hskip 3pt}12{\hskip 3pt}0{\hskip 3pt}26{\hskip 3pt}9{\hskip 3pt}8{\hskip 3pt}17{\hskip 3pt}0{\hskip 3pt}0\hskip 3pt\sf rs~cd:ab;\allowbreak t~b:acs;\allowbreak uv~at:bcr}\\
304)&\BL{18{\hskip 3pt}12{\hskip 3pt}0{\hskip 3pt}30{\hskip 3pt}9{\hskip 3pt}8{\hskip 3pt}15{\hskip 3pt}0{\hskip 3pt}0\hskip 3pt\sf rs~cd:ab;\allowbreak t~b:acs;\allowbreak uv~at:bcr}\\
305)&\BL{18{\hskip 3pt}14{\hskip 3pt}0{\hskip 3pt}16{\hskip 3pt}13{\hskip 3pt}2{\hskip 3pt}13{\hskip 3pt}0{\hskip 3pt}0\hskip 3pt\sf rs~cd:ab;\allowbreak t~b:acs;\allowbreak u~(at):bcs\allowbreak/143}\\
306)&\BL{18{\hskip 3pt}18{\hskip 3pt}0{\hskip 3pt}21{\hskip 3pt}10{\hskip 3pt}0{\hskip 3pt}0{\hskip 3pt}0{\hskip 3pt}8\hskip 3pt\sf rs~ad:bc;\allowbreak t~c:abs;\allowbreak uv~bt:acr}\\
307)&\BL{18{\hskip 3pt}19{\hskip 3pt}12{\hskip 3pt}15{\hskip 3pt}0{\hskip 3pt}20{\hskip 3pt}12{\hskip 3pt}0{\hskip 3pt}0\hskip 3pt\sf rs~cd:ab;\allowbreak t~c:as;\allowbreak u~b:acrt\allowbreak/149}\\
308)&\BL{18{\hskip 3pt}20{\hskip 3pt}0{\hskip 3pt}31{\hskip 3pt}6{\hskip 3pt}6{\hskip 3pt}6{\hskip 3pt}0{\hskip 3pt}0\hskip 3pt\sf rs~cd:ab;\allowbreak t~a:bcr;\allowbreak u~b:acrt\allowbreak/37}\\
309)&\BL{18{\hskip 3pt}21{\hskip 3pt}0{\hskip 3pt}18{\hskip 3pt}10{\hskip 3pt}0{\hskip 3pt}0{\hskip 3pt}0{\hskip 3pt}8\hskip 3pt\sf rs~ad:bc;\allowbreak t~c:abs;\allowbreak uv~bt:acr}\\
}\Q\th{%
310)&\BL{19{\hskip 3pt}11{\hskip 3pt}0{\hskip 3pt}21{\hskip 3pt}20{\hskip 3pt}0{\hskip 3pt}7{\hskip 3pt}0{\hskip 3pt}0\hskip 3pt\sf rs~cd:ab;\allowbreak t~b:acs;\allowbreak uv~at:bcr}\\
311)&\BL{19{\hskip 3pt}11{\hskip 3pt}0{\hskip 3pt}22{\hskip 3pt}19{\hskip 3pt}7{\hskip 3pt}0{\hskip 3pt}0{\hskip 3pt}0\hskip 3pt\sf rs~cd:ab;\allowbreak t~a:bcs;\allowbreak uv~bt:acr}\\
312)&\BL{19{\hskip 3pt}11{\hskip 3pt}0{\hskip 3pt}23{\hskip 3pt}11{\hskip 3pt}9{\hskip 3pt}12{\hskip 3pt}0{\hskip 3pt}1\hskip 3pt\sf rs~cd:ab;\allowbreak t~b:acs;\allowbreak uv~at:bcr}\\
313)&\BL{19{\hskip 3pt}12{\hskip 3pt}0{\hskip 3pt}22{\hskip 3pt}10{\hskip 3pt}10{\hskip 3pt}14{\hskip 3pt}0{\hskip 3pt}0\hskip 3pt\sf rs~cd:ab;\allowbreak t~b:acs;\allowbreak uv~at:bcr}\\
314)&\BL{19{\hskip 3pt}12{\hskip 3pt}0{\hskip 3pt}25{\hskip 3pt}10{\hskip 3pt}9{\hskip 3pt}15{\hskip 3pt}0{\hskip 3pt}0\hskip 3pt\sf rs~cd:ab;\allowbreak t~b:acs;\allowbreak uv~at:bcr}\\
315)&\BL{19{\hskip 3pt}13{\hskip 3pt}0{\hskip 3pt}45{\hskip 3pt}7{\hskip 3pt}29{\hskip 3pt}23{\hskip 3pt}0{\hskip 3pt}0\hskip 3pt\sf rs~cd:ab;\allowbreak tu~cs:ab;\allowbreak v~a:bcrtu}\\
316)&\BL{19{\hskip 3pt}13{\hskip 3pt}0{\hskip 3pt}49{\hskip 3pt}7{\hskip 3pt}29{\hskip 3pt}19{\hskip 3pt}0{\hskip 3pt}0\hskip 3pt\sf rs~cd:ab;\allowbreak tu~cs:ab;\allowbreak v~a:bcrtu}\\
317)&\BL{19{\hskip 3pt}13{\hskip 3pt}0{\hskip 3pt}52{\hskip 3pt}7{\hskip 3pt}30{\hskip 3pt}17{\hskip 3pt}0{\hskip 3pt}0\hskip 3pt\sf rs~cd:ab;\allowbreak tu~cs:ab;\allowbreak v~a:bcrtu}\\
318)&\BL{19{\hskip 3pt}13{\hskip 3pt}0{\hskip 3pt}56{\hskip 3pt}7{\hskip 3pt}30{\hskip 3pt}15{\hskip 3pt}0{\hskip 3pt}0\hskip 3pt\sf rs~cd:ab;\allowbreak tu~cs:ab;\allowbreak v~a:bcrtu}\\
319)&\BL{19{\hskip 3pt}21{\hskip 3pt}13{\hskip 3pt}17{\hskip 3pt}0{\hskip 3pt}21{\hskip 3pt}12{\hskip 3pt}0{\hskip 3pt}0\hskip 3pt\sf rs~cd:ab;\allowbreak t~c:as;\allowbreak u~b:acrt\allowbreak/164}\\
320)&\BL{20{\hskip 3pt}4{\hskip 3pt}0{\hskip 3pt}51{\hskip 3pt}34{\hskip 3pt}5{\hskip 3pt}8{\hskip 3pt}0{\hskip 3pt}0\hskip 3pt\sf rs~cd:ab;\allowbreak tu~cr:ab;\allowbreak v~a:bcstu}\\
321)&\BL{20{\hskip 3pt}10{\hskip 3pt}4{\hskip 3pt}15{\hskip 3pt}14{\hskip 3pt}7{\hskip 3pt}8{\hskip 3pt}0{\hskip 3pt}0\hskip 3pt\sf rs~cd:ab;\allowbreak t~a:bcs;\allowbreak u~a:bcst\allowbreak/214}\\
322)&\BL{20{\hskip 3pt}12{\hskip 3pt}1{\hskip 3pt}22{\hskip 3pt}20{\hskip 3pt}0{\hskip 3pt}8{\hskip 3pt}0{\hskip 3pt}0\hskip 3pt\sf rs~cd:ab;\allowbreak t~b:acs;\allowbreak uv~at:bcr}\\
323)&\BL{20{\hskip 3pt}12{\hskip 3pt}1{\hskip 3pt}23{\hskip 3pt}19{\hskip 3pt}8{\hskip 3pt}0{\hskip 3pt}0{\hskip 3pt}0\hskip 3pt\sf rs~cd:ab;\allowbreak t~a:bcs;\allowbreak uv~bt:acr}\\
324)&\BL{20{\hskip 3pt}13{\hskip 3pt}0{\hskip 3pt}24{\hskip 3pt}11{\hskip 3pt}8{\hskip 3pt}13{\hskip 3pt}0{\hskip 3pt}1\hskip 3pt\sf rs~cd:ab;\allowbreak t~b:acs;\allowbreak uv~at:bcr}\\
325)&\BL{20{\hskip 3pt}13{\hskip 3pt}0{\hskip 3pt}25{\hskip 3pt}11{\hskip 3pt}8{\hskip 3pt}11{\hskip 3pt}0{\hskip 3pt}1\hskip 3pt\sf rs~cd:ab;\allowbreak t~b:acs;\allowbreak uv~at:bcr}\\
326)&\BL{20{\hskip 3pt}14{\hskip 3pt}0{\hskip 3pt}18{\hskip 3pt}17{\hskip 3pt}4{\hskip 3pt}5{\hskip 3pt}0{\hskip 3pt}0\hskip 3pt\sf rs~cd:ab;\allowbreak t~b:acs;\allowbreak u~b:acst\allowbreak/189}\\
327)&\BL{20{\hskip 3pt}14{\hskip 3pt}0{\hskip 3pt}20{\hskip 3pt}16{\hskip 3pt}3{\hskip 3pt}6{\hskip 3pt}0{\hskip 3pt}0\hskip 3pt\sf rs~cd:ab;\allowbreak t~b:acs;\allowbreak u~b:acst\allowbreak/208}\\
328)&\BL{20{\hskip 3pt}16{\hskip 3pt}0{\hskip 3pt}18{\hskip 3pt}17{\hskip 3pt}2{\hskip 3pt}7{\hskip 3pt}0{\hskip 3pt}0\hskip 3pt\sf rs~cd:ab;\allowbreak t~b:acs;\allowbreak u~b:acst;\allowbreak v~b:acstu}\\
329)&\BL{21{\hskip 3pt}11{\hskip 3pt}2{\hskip 3pt}24{\hskip 3pt}23{\hskip 3pt}0{\hskip 3pt}7{\hskip 3pt}0{\hskip 3pt}0\hskip 3pt\sf rs~cd:ab;\allowbreak t~b:acs;\allowbreak uv~at:bcr}\\
330)&\BL{21{\hskip 3pt}12{\hskip 3pt}0{\hskip 3pt}24{\hskip 3pt}12{\hskip 3pt}11{\hskip 3pt}15{\hskip 3pt}0{\hskip 3pt}0\hskip 3pt\sf rs~cd:ab;\allowbreak t~b:acs;\allowbreak uv~at:bcr}\\
331)&\BL{21{\hskip 3pt}12{\hskip 3pt}0{\hskip 3pt}25{\hskip 3pt}12{\hskip 3pt}10{\hskip 3pt}19{\hskip 3pt}0{\hskip 3pt}0\hskip 3pt\sf rs~cd:ab;\allowbreak t~b:acs;\allowbreak uv~at:bcr}\\
332)&\BL{21{\hskip 3pt}12{\hskip 3pt}0{\hskip 3pt}26{\hskip 3pt}12{\hskip 3pt}10{\hskip 3pt}16{\hskip 3pt}0{\hskip 3pt}1\hskip 3pt\sf rs~cd:ab;\allowbreak t~b:acs;\allowbreak uv~at:bcr}\\
333)&\BL{21{\hskip 3pt}12{\hskip 3pt}0{\hskip 3pt}26{\hskip 3pt}12{\hskip 3pt}10{\hskip 3pt}17{\hskip 3pt}0{\hskip 3pt}0\hskip 3pt\sf rs~cd:ab;\allowbreak t~b:acs;\allowbreak uv~at:bcr}\\
334)&\BL{21{\hskip 3pt}13{\hskip 3pt}0{\hskip 3pt}29{\hskip 3pt}12{\hskip 3pt}7{\hskip 3pt}18{\hskip 3pt}0{\hskip 3pt}0\hskip 3pt\sf rs~cd:ab;\allowbreak t~b:acs;\allowbreak uv~at:bcr}\\
335)&\BL{21{\hskip 3pt}14{\hskip 3pt}0{\hskip 3pt}26{\hskip 3pt}11{\hskip 3pt}9{\hskip 3pt}15{\hskip 3pt}0{\hskip 3pt}0\hskip 3pt\sf rs~cd:ab;\allowbreak t~b:acs;\allowbreak uv~at:bcr}\\
336)&\BL{21{\hskip 3pt}14{\hskip 3pt}0{\hskip 3pt}48{\hskip 3pt}8{\hskip 3pt}34{\hskip 3pt}24{\hskip 3pt}0{\hskip 3pt}0\hskip 3pt\sf rs~cd:ab;\allowbreak tu~cs:ab;\allowbreak v~a:bcrtu}\\
337)&\BL{21{\hskip 3pt}14{\hskip 3pt}0{\hskip 3pt}50{\hskip 3pt}8{\hskip 3pt}34{\hskip 3pt}22{\hskip 3pt}0{\hskip 3pt}0\hskip 3pt\sf rs~cd:ab;\allowbreak tu~cs:ab;\allowbreak v~a:bcrtu}\\
338)&\BL{21{\hskip 3pt}14{\hskip 3pt}0{\hskip 3pt}56{\hskip 3pt}8{\hskip 3pt}36{\hskip 3pt}18{\hskip 3pt}0{\hskip 3pt}0\hskip 3pt\sf rs~cd:ab;\allowbreak tu~cs:ab;\allowbreak v~a:bcrtu}\\
339)&\BL{21{\hskip 3pt}14{\hskip 3pt}0{\hskip 3pt}60{\hskip 3pt}8{\hskip 3pt}36{\hskip 3pt}16{\hskip 3pt}0{\hskip 3pt}0\hskip 3pt\sf rs~cd:ab;\allowbreak tu~cs:ab;\allowbreak v~a:bcrtu}\\
340)&\BL{22{\hskip 3pt}4{\hskip 3pt}0{\hskip 3pt}53{\hskip 3pt}38{\hskip 3pt}7{\hskip 3pt}10{\hskip 3pt}0{\hskip 3pt}0\hskip 3pt\sf rs~cd:ab;\allowbreak tu~cr:ab;\allowbreak v~a:bcstu}\\
341)&\BL{22{\hskip 3pt}11{\hskip 3pt}0{\hskip 3pt}24{\hskip 3pt}14{\hskip 3pt}11{\hskip 3pt}17{\hskip 3pt}0{\hskip 3pt}0\hskip 3pt\sf rs~cd:ab;\allowbreak t~b:acs;\allowbreak uv~at:bcr}\\
342)&\BL{22{\hskip 3pt}13{\hskip 3pt}0{\hskip 3pt}24{\hskip 3pt}13{\hskip 3pt}9{\hskip 3pt}19{\hskip 3pt}0{\hskip 3pt}0\hskip 3pt\sf rs~cd:ab;\allowbreak t~b:acs;\allowbreak uv~at:bcr}\\
343)&\BL{22{\hskip 3pt}13{\hskip 3pt}0{\hskip 3pt}27{\hskip 3pt}13{\hskip 3pt}8{\hskip 3pt}19{\hskip 3pt}0{\hskip 3pt}1\hskip 3pt\sf rs~cd:ab;\allowbreak t~b:acs;\allowbreak uv~at:bcr}\\
344)&\BL{22{\hskip 3pt}14{\hskip 3pt}0{\hskip 3pt}25{\hskip 3pt}12{\hskip 3pt}10{\hskip 3pt}16{\hskip 3pt}0{\hskip 3pt}0\hskip 3pt\sf rs~cd:ab;\allowbreak t~b:acs;\allowbreak uv~at:bcr}\\
345)&\BL{22{\hskip 3pt}15{\hskip 3pt}0{\hskip 3pt}28{\hskip 3pt}11{\hskip 3pt}11{\hskip 3pt}12{\hskip 3pt}0{\hskip 3pt}0\hskip 3pt\sf rs~cd:ab;\allowbreak t~b:acs;\allowbreak uv~at:bcr}\\
346)&\BL{22{\hskip 3pt}23{\hskip 3pt}0{\hskip 3pt}19{\hskip 3pt}14{\hskip 3pt}0{\hskip 3pt}0{\hskip 3pt}27{\hskip 3pt}0\hskip 3pt\sf rs~bd:ac;\allowbreak t~c:abs;\allowbreak u~c:abrt;\allowbreak v~c:abrtu}\\
347)&\BL{22{\hskip 3pt}26{\hskip 3pt}9{\hskip 3pt}22{\hskip 3pt}0{\hskip 3pt}39{\hskip 3pt}15{\hskip 3pt}0{\hskip 3pt}0\hskip 3pt\sf r~c:ab;\allowbreak st~cd:abr;\allowbreak uv~cr:at}\\
348)&\BL{23{\hskip 3pt}13{\hskip 3pt}0{\hskip 3pt}25{\hskip 3pt}14{\hskip 3pt}10{\hskip 3pt}18{\hskip 3pt}0{\hskip 3pt}0\hskip 3pt\sf rs~cd:ab;\allowbreak t~b:acs;\allowbreak uv~at:bcr}\\
349)&\BL{23{\hskip 3pt}13{\hskip 3pt}0{\hskip 3pt}27{\hskip 3pt}23{\hskip 3pt}0{\hskip 3pt}9{\hskip 3pt}0{\hskip 3pt}0\hskip 3pt\sf rs~cd:ab;\allowbreak t~b:acs;\allowbreak uv~at:bcr}\\
350)&\BL{23{\hskip 3pt}13{\hskip 3pt}0{\hskip 3pt}30{\hskip 3pt}14{\hskip 3pt}8{\hskip 3pt}22{\hskip 3pt}0{\hskip 3pt}0\hskip 3pt\sf rs~cd:ab;\allowbreak t~b:acs;\allowbreak uv~at:bcr}\\
351)&\BL{23{\hskip 3pt}13{\hskip 3pt}0{\hskip 3pt}30{\hskip 3pt}14{\hskip 3pt}8{\hskip 3pt}20{\hskip 3pt}0{\hskip 3pt}2\hskip 3pt\sf rs~cd:ab;\allowbreak t~b:acs;\allowbreak uv~at:bcr}\\
352)&\BL{23{\hskip 3pt}14{\hskip 3pt}0{\hskip 3pt}27{\hskip 3pt}13{\hskip 3pt}11{\hskip 3pt}14{\hskip 3pt}0{\hskip 3pt}0\hskip 3pt\sf rs~cd:ab;\allowbreak t~b:acs;\allowbreak uv~at:bcr}\\
353)&\BL{23{\hskip 3pt}15{\hskip 3pt}0{\hskip 3pt}57{\hskip 3pt}10{\hskip 3pt}23{\hskip 3pt}20{\hskip 3pt}0{\hskip 3pt}0\hskip 3pt\sf rs~cd:ab;\allowbreak tu~cs:ab;\allowbreak v~a:bcrtu}\\
354)&\BL{23{\hskip 3pt}15{\hskip 3pt}0{\hskip 3pt}57{\hskip 3pt}10{\hskip 3pt}22{\hskip 3pt}22{\hskip 3pt}0{\hskip 3pt}0\hskip 3pt\sf rs~cd:ab;\allowbreak tu~cs:ab;\allowbreak v~a:bcrtu}\\
355)&\BL{23{\hskip 3pt}17{\hskip 3pt}0{\hskip 3pt}27{\hskip 3pt}14{\hskip 3pt}4{\hskip 3pt}13{\hskip 3pt}0{\hskip 3pt}0\hskip 3pt\sf rs~cd:ab;\allowbreak t~b:acs;\allowbreak uv~at:bcr}\\
356)&\BL{23{\hskip 3pt}17{\hskip 3pt}0{\hskip 3pt}71{\hskip 3pt}8{\hskip 3pt}34{\hskip 3pt}17{\hskip 3pt}0{\hskip 3pt}0\hskip 3pt\sf rs~cd:ab;\allowbreak tu~cs:ab;\allowbreak v~a:bcrtu}\\
357)&\BL{23{\hskip 3pt}24{\hskip 3pt}0{\hskip 3pt}20{\hskip 3pt}15{\hskip 3pt}0{\hskip 3pt}0{\hskip 3pt}25{\hskip 3pt}0\hskip 3pt\sf rs~bd:ac;\allowbreak t~c:abs;\allowbreak u~c:abrt;\allowbreak v~c:abrtu}\\
358)&\BL{24{\hskip 3pt}13{\hskip 3pt}3{\hskip 3pt}28{\hskip 3pt}25{\hskip 3pt}7{\hskip 3pt}0{\hskip 3pt}0{\hskip 3pt}0\hskip 3pt\sf rs~cd:ab;\allowbreak t~a:bcs;\allowbreak uv~bt:acr}\\
359)&\BL{24{\hskip 3pt}13{\hskip 3pt}3{\hskip 3pt}29{\hskip 3pt}24{\hskip 3pt}7{\hskip 3pt}0{\hskip 3pt}0{\hskip 3pt}0\hskip 3pt\sf rs~cd:ab;\allowbreak t~a:bcs;\allowbreak uv~bt:acr}\\
360)&\BL{24{\hskip 3pt}14{\hskip 3pt}1{\hskip 3pt}28{\hskip 3pt}23{\hskip 3pt}0{\hskip 3pt}10{\hskip 3pt}0{\hskip 3pt}0\hskip 3pt\sf rs~cd:ab;\allowbreak t~b:acs;\allowbreak uv~at:bcr}\\
361)&\BL{24{\hskip 3pt}16{\hskip 3pt}0{\hskip 3pt}31{\hskip 3pt}12{\hskip 3pt}12{\hskip 3pt}17{\hskip 3pt}0{\hskip 3pt}0\hskip 3pt\sf rs~cd:ab;\allowbreak t~b:acs;\allowbreak uv~at:bcr}\\
362)&\BL{24{\hskip 3pt}19{\hskip 3pt}0{\hskip 3pt}77{\hskip 3pt}8{\hskip 3pt}36{\hskip 3pt}17{\hskip 3pt}0{\hskip 3pt}0\hskip 3pt\sf rs~cd:ab;\allowbreak tu~cs:ab;\allowbreak v~a:bcrtu}\\
363)&\BL{24{\hskip 3pt}27{\hskip 3pt}15{\hskip 3pt}24{\hskip 3pt}0{\hskip 3pt}33{\hskip 3pt}11{\hskip 3pt}0{\hskip 3pt}0\hskip 3pt\sf r~c:ab;\allowbreak st~cd:abr;\allowbreak uv~cr:at}\\
364)&\BL{25{\hskip 3pt}14{\hskip 3pt}1{\hskip 3pt}29{\hskip 3pt}25{\hskip 3pt}0{\hskip 3pt}10{\hskip 3pt}0{\hskip 3pt}0\hskip 3pt\sf rs~cd:ab;\allowbreak t~b:acs;\allowbreak uv~at:bcr}\\
365)&\BL{25{\hskip 3pt}16{\hskip 3pt}0{\hskip 3pt}36{\hskip 3pt}14{\hskip 3pt}8{\hskip 3pt}21{\hskip 3pt}0{\hskip 3pt}0\hskip 3pt\sf rs~cd:ab;\allowbreak t~b:acs;\allowbreak uv~at:bcr}\\
}\Q\ti{%
366)&\BL{25{\hskip 3pt}41{\hskip 3pt}0{\hskip 3pt}81{\hskip 3pt}6{\hskip 3pt}30{\hskip 3pt}7{\hskip 3pt}0{\hskip 3pt}0\hskip 3pt\sf rs~ac:bd;\allowbreak t~r:abs;\allowbreak u~(ar):bdt\allowbreak/213}\\
367)&\BL{26{\hskip 3pt}15{\hskip 3pt}0{\hskip 3pt}31{\hskip 3pt}15{\hskip 3pt}12{\hskip 3pt}19{\hskip 3pt}0{\hskip 3pt}1\hskip 3pt\sf rs~cd:ab;\allowbreak t~b:acs;\allowbreak uv~at:bcr}\\
368)&\BL{26{\hskip 3pt}16{\hskip 3pt}0{\hskip 3pt}31{\hskip 3pt}14{\hskip 3pt}13{\hskip 3pt}21{\hskip 3pt}0{\hskip 3pt}0\hskip 3pt\sf rs~cd:ab;\allowbreak t~b:acs;\allowbreak uv~at:bcr}\\
369)&\BL{26{\hskip 3pt}17{\hskip 3pt}0{\hskip 3pt}25{\hskip 3pt}21{\hskip 3pt}6{\hskip 3pt}7{\hskip 3pt}0{\hskip 3pt}0\hskip 3pt\sf rs~cd:ab;\allowbreak t~b:acs;\allowbreak u~b:acst\allowbreak/197}\\
370)&\BL{26{\hskip 3pt}18{\hskip 3pt}0{\hskip 3pt}32{\hskip 3pt}13{\hskip 3pt}12{\hskip 3pt}17{\hskip 3pt}0{\hskip 3pt}0\hskip 3pt\sf rs~cd:ab;\allowbreak t~b:acs;\allowbreak uv~at:bcr}\\
371)&\BL{26{\hskip 3pt}19{\hskip 3pt}0{\hskip 3pt}61{\hskip 3pt}9{\hskip 3pt}40{\hskip 3pt}33{\hskip 3pt}0{\hskip 3pt}0\hskip 3pt\sf rs~cd:ab;\allowbreak tu~cs:ab;\allowbreak v~a:bcrtu}\\
372)&\BL{26{\hskip 3pt}19{\hskip 3pt}0{\hskip 3pt}80{\hskip 3pt}9{\hskip 3pt}39{\hskip 3pt}20{\hskip 3pt}0{\hskip 3pt}0\hskip 3pt\sf rs~cd:ab;\allowbreak tu~cs:ab;\allowbreak v~a:bcrtu}\\
373)&\BL{27{\hskip 3pt}15{\hskip 3pt}0{\hskip 3pt}32{\hskip 3pt}17{\hskip 3pt}10{\hskip 3pt}18{\hskip 3pt}0{\hskip 3pt}2\hskip 3pt\sf rs~cd:ab;\allowbreak t~b:acs;\allowbreak uv~at:bcr}\\
374)&\BL{27{\hskip 3pt}15{\hskip 3pt}0{\hskip 3pt}32{\hskip 3pt}16{\hskip 3pt}12{\hskip 3pt}26{\hskip 3pt}0{\hskip 3pt}0\hskip 3pt\sf rs~cd:ab;\allowbreak t~b:acs;\allowbreak uv~at:bcr}\\
375)&\BL{27{\hskip 3pt}16{\hskip 3pt}0{\hskip 3pt}33{\hskip 3pt}15{\hskip 3pt}13{\hskip 3pt}22{\hskip 3pt}0{\hskip 3pt}0\hskip 3pt\sf rs~cd:ab;\allowbreak t~b:acs;\allowbreak uv~at:bcr}\\
376)&\BL{27{\hskip 3pt}24{\hskip 3pt}15{\hskip 3pt}24{\hskip 3pt}0{\hskip 3pt}39{\hskip 3pt}16{\hskip 3pt}0{\hskip 3pt}0\hskip 3pt\sf r~c:ab;\allowbreak st~cd:abr;\allowbreak uv~cr:at}\\
377)&\BL{27{\hskip 3pt}29{\hskip 3pt}0{\hskip 3pt}23{\hskip 3pt}17{\hskip 3pt}0{\hskip 3pt}0{\hskip 3pt}31{\hskip 3pt}0\hskip 3pt\sf rs~bd:ac;\allowbreak t~c:abs;\allowbreak u~c:abrt;\allowbreak v~c:abrtu}\\
378)&\BL{28{\hskip 3pt}16{\hskip 3pt}1{\hskip 3pt}31{\hskip 3pt}29{\hskip 3pt}0{\hskip 3pt}11{\hskip 3pt}0{\hskip 3pt}0\hskip 3pt\sf rs~cd:ab;\allowbreak t~b:acs;\allowbreak uv~at:bcr}\\
379)&\BL{28{\hskip 3pt}16{\hskip 3pt}1{\hskip 3pt}33{\hskip 3pt}27{\hskip 3pt}11{\hskip 3pt}0{\hskip 3pt}0{\hskip 3pt}0\hskip 3pt\sf rs~cd:ab;\allowbreak t~a:bcs;\allowbreak uv~bt:acr}\\
380)&\BL{28{\hskip 3pt}19{\hskip 3pt}0{\hskip 3pt}35{\hskip 3pt}14{\hskip 3pt}14{\hskip 3pt}17{\hskip 3pt}0{\hskip 3pt}0\hskip 3pt\sf rs~cd:ab;\allowbreak t~b:acs;\allowbreak uv~at:bcr}\\
381)&\BL{28{\hskip 3pt}20{\hskip 3pt}0{\hskip 3pt}38{\hskip 3pt}13{\hskip 3pt}14{\hskip 3pt}21{\hskip 3pt}0{\hskip 3pt}0\hskip 3pt\sf rs~cd:ab;\allowbreak t~b:acs;\allowbreak uv~at:bcr}\\
382)&\BL{28{\hskip 3pt}32{\hskip 3pt}13{\hskip 3pt}28{\hskip 3pt}0{\hskip 3pt}45{\hskip 3pt}17{\hskip 3pt}0{\hskip 3pt}0\hskip 3pt\sf r~c:ab;\allowbreak st~cd:abr;\allowbreak uv~cr:at}\\
383)&\BL{29{\hskip 3pt}19{\hskip 3pt}0{\hskip 3pt}70{\hskip 3pt}12{\hskip 3pt}30{\hskip 3pt}36{\hskip 3pt}0{\hskip 3pt}0\hskip 3pt\sf rs~cd:ab;\allowbreak tu~cs:ab;\allowbreak v~a:bcrtu}\\
384)&\BL{30{\hskip 3pt}8{\hskip 3pt}0{\hskip 3pt}81{\hskip 3pt}48{\hskip 3pt}3{\hskip 3pt}6{\hskip 3pt}0{\hskip 3pt}0\hskip 3pt\sf rs~cd:ab;\allowbreak tu~cr:ab;\allowbreak v~a:bcstu}\\
385)&\BL{30{\hskip 3pt}23{\hskip 3pt}0{\hskip 3pt}91{\hskip 3pt}10{\hskip 3pt}44{\hskip 3pt}23{\hskip 3pt}0{\hskip 3pt}0\hskip 3pt\sf rs~cd:ab;\allowbreak tu~cs:ab;\allowbreak v~a:bcrtu}\\
386)&\BL{31{\hskip 3pt}16{\hskip 3pt}0{\hskip 3pt}36{\hskip 3pt}19{\hskip 3pt}16{\hskip 3pt}24{\hskip 3pt}0{\hskip 3pt}1\hskip 3pt\sf rs~cd:ab;\allowbreak t~b:acs;\allowbreak uv~at:bcr}\\
387)&\BL{31{\hskip 3pt}17{\hskip 3pt}0{\hskip 3pt}34{\hskip 3pt}19{\hskip 3pt}14{\hskip 3pt}24{\hskip 3pt}0{\hskip 3pt}0\hskip 3pt\sf rs~cd:ab;\allowbreak t~b:acs;\allowbreak uv~at:bcr}\\
388)&\BL{31{\hskip 3pt}17{\hskip 3pt}0{\hskip 3pt}37{\hskip 3pt}20{\hskip 3pt}11{\hskip 3pt}19{\hskip 3pt}0{\hskip 3pt}3\hskip 3pt\sf rs~cd:ab;\allowbreak t~b:acs;\allowbreak uv~at:bcr}\\
389)&\BL{31{\hskip 3pt}18{\hskip 3pt}0{\hskip 3pt}38{\hskip 3pt}18{\hskip 3pt}14{\hskip 3pt}20{\hskip 3pt}0{\hskip 3pt}2\hskip 3pt\sf rs~cd:ab;\allowbreak t~b:acs;\allowbreak uv~at:bcr}\\
390)&\BL{31{\hskip 3pt}24{\hskip 3pt}0{\hskip 3pt}97{\hskip 3pt}11{\hskip 3pt}53{\hskip 3pt}19{\hskip 3pt}0{\hskip 3pt}0\hskip 3pt\sf rs~cd:ab;\allowbreak tu~cs:ab;\allowbreak v~a:bcrtu}\\
391)&\BL{32{\hskip 3pt}18{\hskip 3pt}1{\hskip 3pt}37{\hskip 3pt}32{\hskip 3pt}0{\hskip 3pt}13{\hskip 3pt}0{\hskip 3pt}0\hskip 3pt\sf rs~cd:ab;\allowbreak t~b:acs;\allowbreak uv~at:bcr}\\
392)&\BL{32{\hskip 3pt}21{\hskip 3pt}0{\hskip 3pt}39{\hskip 3pt}17{\hskip 3pt}14{\hskip 3pt}21{\hskip 3pt}0{\hskip 3pt}1\hskip 3pt\sf rs~cd:ab;\allowbreak t~b:acs;\allowbreak uv~at:bcr}\\
393)&\BL{32{\hskip 3pt}22{\hskip 3pt}0{\hskip 3pt}70{\hskip 3pt}12{\hskip 3pt}55{\hskip 3pt}36{\hskip 3pt}0{\hskip 3pt}0\hskip 3pt\sf rs~cd:ab;\allowbreak tu~cs:ab;\allowbreak v~a:bcrtu}\\
394)&\BL{33{\hskip 3pt}21{\hskip 3pt}0{\hskip 3pt}43{\hskip 3pt}18{\hskip 3pt}13{\hskip 3pt}28{\hskip 3pt}0{\hskip 3pt}0\hskip 3pt\sf rs~cd:ab;\allowbreak t~b:acs;\allowbreak uv~at:bcr}\\
395)&\BL{33{\hskip 3pt}25{\hskip 3pt}0{\hskip 3pt}97{\hskip 3pt}11{\hskip 3pt}46{\hskip 3pt}29{\hskip 3pt}0{\hskip 3pt}0\hskip 3pt\sf rs~cd:ab;\allowbreak tu~cs:ab;\allowbreak v~a:bcrtu}\\
396)&\BL{33{\hskip 3pt}25{\hskip 3pt}0{\hskip 3pt}97{\hskip 3pt}11{\hskip 3pt}47{\hskip 3pt}28{\hskip 3pt}0{\hskip 3pt}0\hskip 3pt\sf rs~cd:ab;\allowbreak tu~cs:ab;\allowbreak v~a:bcrtu}\\
397)&\BL{33{\hskip 3pt}25{\hskip 3pt}0{\hskip 3pt}99{\hskip 3pt}11{\hskip 3pt}47{\hskip 3pt}27{\hskip 3pt}0{\hskip 3pt}0\hskip 3pt\sf rs~cd:ab;\allowbreak tu~cs:ab;\allowbreak v~a:bcrtu}\\
398)&\BL{33{\hskip 3pt}25{\hskip 3pt}0{\hskip 3pt}99{\hskip 3pt}11{\hskip 3pt}46{\hskip 3pt}28{\hskip 3pt}0{\hskip 3pt}0\hskip 3pt\sf rs~cd:ab;\allowbreak tu~cs:ab;\allowbreak v~a:bcrtu}\\
399)&\BL{34{\hskip 3pt}20{\hskip 3pt}0{\hskip 3pt}44{\hskip 3pt}19{\hskip 3pt}16{\hskip 3pt}25{\hskip 3pt}0{\hskip 3pt}2\hskip 3pt\sf rs~cd:ab;\allowbreak t~b:acs;\allowbreak uv~at:bcr}\\
400)&\BL{34{\hskip 3pt}20{\hskip 3pt}0{\hskip 3pt}50{\hskip 3pt}19{\hskip 3pt}16{\hskip 3pt}25{\hskip 3pt}0{\hskip 3pt}0\hskip 3pt\sf rs~cd:ab;\allowbreak t~b:acs;\allowbreak uv~at:bcr}\\
401)&\BL{34{\hskip 3pt}21{\hskip 3pt}0{\hskip 3pt}42{\hskip 3pt}19{\hskip 3pt}15{\hskip 3pt}20{\hskip 3pt}0{\hskip 3pt}2\hskip 3pt\sf rs~cd:ab;\allowbreak t~b:acs;\allowbreak uv~at:bcr}\\
402)&\BL{35{\hskip 3pt}22{\hskip 3pt}0{\hskip 3pt}45{\hskip 3pt}19{\hskip 3pt}15{\hskip 3pt}28{\hskip 3pt}0{\hskip 3pt}0\hskip 3pt\sf rs~cd:ab;\allowbreak t~b:acs;\allowbreak uv~at:bcr}\\
403)&\BL{35{\hskip 3pt}37{\hskip 3pt}0{\hskip 3pt}30{\hskip 3pt}23{\hskip 3pt}0{\hskip 3pt}0{\hskip 3pt}37{\hskip 3pt}0\hskip 3pt\sf rs~bd:ac;\allowbreak t~c:abs;\allowbreak u~c:abrt;\allowbreak v~c:abrtu}\\
}\Q\tj{%
404)&\BL{36{\hskip 3pt}21{\hskip 3pt}0{\hskip 3pt}47{\hskip 3pt}21{\hskip 3pt}14{\hskip 3pt}29{\hskip 3pt}0{\hskip 3pt}3\hskip 3pt\sf rs~cd:ab;\allowbreak t~b:acs;\allowbreak uv~at:bcr}\\
405)&\BL{36{\hskip 3pt}25{\hskip 3pt}0{\hskip 3pt}35{\hskip 3pt}29{\hskip 3pt}6{\hskip 3pt}11{\hskip 3pt}0{\hskip 3pt}0\hskip 3pt\sf rs~cd:ab;\allowbreak t~b:acs;\allowbreak u~b:acst\allowbreak/201}\\
406)&\BL{36{\hskip 3pt}26{\hskip 3pt}0{\hskip 3pt}42{\hskip 3pt}21{\hskip 3pt}9{\hskip 3pt}21{\hskip 3pt}0{\hskip 3pt}0\hskip 3pt\sf rs~cd:ab;\allowbreak t~b:acs;\allowbreak uv~at:bcr}\\
407)&\BL{38{\hskip 3pt}23{\hskip 3pt}0{\hskip 3pt}44{\hskip 3pt}21{\hskip 3pt}19{\hskip 3pt}26{\hskip 3pt}0{\hskip 3pt}0\hskip 3pt\sf rs~cd:ab;\allowbreak t~b:acs;\allowbreak uv~at:bcr}\\
408)&\BL{38{\hskip 3pt}23{\hskip 3pt}0{\hskip 3pt}45{\hskip 3pt}23{\hskip 3pt}12{\hskip 3pt}29{\hskip 3pt}0{\hskip 3pt}3\hskip 3pt\sf rs~cd:ab;\allowbreak t~b:acs;\allowbreak uv~at:bcr}\\
409)&\BL{38{\hskip 3pt}24{\hskip 3pt}0{\hskip 3pt}46{\hskip 3pt}21{\hskip 3pt}18{\hskip 3pt}21{\hskip 3pt}0{\hskip 3pt}0\hskip 3pt\sf rs~cd:ab;\allowbreak t~b:acs;\allowbreak uv~at:bcr}\\
410)&\BL{40{\hskip 3pt}21{\hskip 3pt}0{\hskip 3pt}44{\hskip 3pt}25{\hskip 3pt}19{\hskip 3pt}30{\hskip 3pt}0{\hskip 3pt}0\hskip 3pt\sf rs~cd:ab;\allowbreak t~b:acs;\allowbreak uv~at:bcr}\\
411)&\BL{40{\hskip 3pt}23{\hskip 3pt}0{\hskip 3pt}47{\hskip 3pt}25{\hskip 3pt}14{\hskip 3pt}27{\hskip 3pt}0{\hskip 3pt}3\hskip 3pt\sf rs~cd:ab;\allowbreak t~b:acs;\allowbreak uv~at:bcr}\\
412)&\BL{40{\hskip 3pt}23{\hskip 3pt}0{\hskip 3pt}51{\hskip 3pt}23{\hskip 3pt}18{\hskip 3pt}29{\hskip 3pt}0{\hskip 3pt}3\hskip 3pt\sf rs~cd:ab;\allowbreak t~b:acs;\allowbreak uv~at:bcr}\\
413)&\BL{40{\hskip 3pt}24{\hskip 3pt}0{\hskip 3pt}47{\hskip 3pt}22{\hskip 3pt}20{\hskip 3pt}31{\hskip 3pt}0{\hskip 3pt}0\hskip 3pt\sf rs~cd:ab;\allowbreak t~b:acs;\allowbreak uv~at:bcr}\\
414)&\BL{40{\hskip 3pt}25{\hskip 3pt}0{\hskip 3pt}44{\hskip 3pt}24{\hskip 3pt}14{\hskip 3pt}31{\hskip 3pt}0{\hskip 3pt}0\hskip 3pt\sf rs~cd:ab;\allowbreak t~b:acs;\allowbreak uv~at:bcr}\\
415)&\BL{40{\hskip 3pt}25{\hskip 3pt}0{\hskip 3pt}49{\hskip 3pt}22{\hskip 3pt}18{\hskip 3pt}25{\hskip 3pt}0{\hskip 3pt}2\hskip 3pt\sf rs~cd:ab;\allowbreak t~b:acs;\allowbreak uv~at:bcr}\\
416)&\BL{40{\hskip 3pt}31{\hskip 3pt}0{\hskip 3pt}40{\hskip 3pt}32{\hskip 3pt}3{\hskip 3pt}15{\hskip 3pt}0{\hskip 3pt}0\hskip 3pt\sf rs~cd:ab;\allowbreak t~b:adr;\allowbreak u~b:adrt;\allowbreak v~b:adrtu}\\
417)&\BL{41{\hskip 3pt}22{\hskip 3pt}0{\hskip 3pt}45{\hskip 3pt}26{\hskip 3pt}18{\hskip 3pt}29{\hskip 3pt}0{\hskip 3pt}0\hskip 3pt\sf rs~cd:ab;\allowbreak t~b:acs;\allowbreak uv~at:bcr}\\
418)&\BL{42{\hskip 3pt}21{\hskip 3pt}8{\hskip 3pt}32{\hskip 3pt}29{\hskip 3pt}14{\hskip 3pt}18{\hskip 3pt}0{\hskip 3pt}0\hskip 3pt\sf rs~cd:ab;\allowbreak t~a:bcs;\allowbreak u~a:bcst;\allowbreak v~a:bcstu}\\
419)&\BL{42{\hskip 3pt}45{\hskip 3pt}21{\hskip 3pt}42{\hskip 3pt}0{\hskip 3pt}63{\hskip 3pt}25{\hskip 3pt}0{\hskip 3pt}0\hskip 3pt\sf r~c:ab;\allowbreak st~cd:abr;\allowbreak uv~cr:at}\\
420)&\BL{42{\hskip 3pt}45{\hskip 3pt}24{\hskip 3pt}42{\hskip 3pt}0{\hskip 3pt}60{\hskip 3pt}22{\hskip 3pt}0{\hskip 3pt}0\hskip 3pt\sf r~c:ab;\allowbreak st~cd:abr;\allowbreak uv~cr:at}\\
421)&\BL{45{\hskip 3pt}31{\hskip 3pt}0{\hskip 3pt}43{\hskip 3pt}37{\hskip 3pt}8{\hskip 3pt}12{\hskip 3pt}0{\hskip 3pt}0\hskip 3pt\sf rs~cd:ab;\allowbreak t~b:adr;\allowbreak u~b:adrt;\allowbreak v~b:adrtu}\\
422)&\BL{46{\hskip 3pt}28{\hskip 3pt}0{\hskip 3pt}62{\hskip 3pt}25{\hskip 3pt}22{\hskip 3pt}33{\hskip 3pt}0{\hskip 3pt}0\hskip 3pt\sf rs~cd:ab;\allowbreak t~b:acs;\allowbreak uv~at:bcr}\\
423)&\BL{48{\hskip 3pt}27{\hskip 3pt}0{\hskip 3pt}57{\hskip 3pt}31{\hskip 3pt}16{\hskip 3pt}29{\hskip 3pt}0{\hskip 3pt}5\hskip 3pt\sf rs~cd:ab;\allowbreak t~b:acs;\allowbreak uv~at:bcr}\\
424)&\BL{54{\hskip 3pt}48{\hskip 3pt}27{\hskip 3pt}48{\hskip 3pt}0{\hskip 3pt}81{\hskip 3pt}35{\hskip 3pt}0{\hskip 3pt}0\hskip 3pt\sf r~c:ab;\allowbreak st~cd:abr;\allowbreak uv~cr:at}\\
425)&\BL{55{\hskip 3pt}34{\hskip 3pt}0{\hskip 3pt}65{\hskip 3pt}30{\hskip 3pt}27{\hskip 3pt}36{\hskip 3pt}0{\hskip 3pt}0\hskip 3pt\sf rs~cd:ab;\allowbreak t~b:acs;\allowbreak uv~at:bcr}\\
426)&\BL{58{\hskip 3pt}44{\hskip 3pt}0{\hskip 3pt}78{\hskip 3pt}27{\hskip 3pt}26{\hskip 3pt}27{\hskip 3pt}0{\hskip 3pt}0\hskip 3pt\sf rs~cd:ab;\allowbreak t~b:acs;\allowbreak uv~at:bcr}\\
427)&\BL{61{\hskip 3pt}36{\hskip 3pt}0{\hskip 3pt}71{\hskip 3pt}34{\hskip 3pt}31{\hskip 3pt}45{\hskip 3pt}0{\hskip 3pt}0\hskip 3pt\sf rs~cd:ab;\allowbreak t~b:acs;\allowbreak uv~at:bcr}\\
428)&\BL{62{\hskip 3pt}36{\hskip 3pt}0{\hskip 3pt}77{\hskip 3pt}35{\hskip 3pt}30{\hskip 3pt}47{\hskip 3pt}0{\hskip 3pt}0\hskip 3pt\sf rs~cd:ab;\allowbreak t~b:acs;\allowbreak uv~at:bcr}\\
429)&\BL{66{\hskip 3pt}44{\hskip 3pt}0{\hskip 3pt}86{\hskip 3pt}33{\hskip 3pt}34{\hskip 3pt}41{\hskip 3pt}0{\hskip 3pt}0\hskip 3pt\sf rs~cd:ab;\allowbreak t~b:acs;\allowbreak uv~at:bcr}\\
430)&\BL{67{\hskip 3pt}42{\hskip 3pt}0{\hskip 3pt}79{\hskip 3pt}36{\hskip 3pt}33{\hskip 3pt}46{\hskip 3pt}0{\hskip 3pt}0\hskip 3pt\sf rs~cd:ab;\allowbreak t~b:acs;\allowbreak uv~at:bcr}\\
431)&\BL{70{\hskip 3pt}35{\hskip 3pt}12{\hskip 3pt}54{\hskip 3pt}49{\hskip 3pt}22{\hskip 3pt}32{\hskip 3pt}0{\hskip 3pt}0\hskip 3pt\sf rs~cd:ab;\allowbreak t~a:bcs;\allowbreak u~a:bcst;\allowbreak v~a:bcstu}\\
432)&\BL{72{\hskip 3pt}39{\hskip 3pt}0{\hskip 3pt}85{\hskip 3pt}47{\hskip 3pt}26{\hskip 3pt}43{\hskip 3pt}0{\hskip 3pt}5\hskip 3pt\sf rs~cd:ab;\allowbreak t~b:acs;\allowbreak uv~at:bcr}\\
433)&\BL{74{\hskip 3pt}44{\hskip 3pt}0{\hskip 3pt}91{\hskip 3pt}41{\hskip 3pt}36{\hskip 3pt}57{\hskip 3pt}0{\hskip 3pt}0\hskip 3pt\sf rs~cd:ab;\allowbreak t~b:acs;\allowbreak uv~at:bcr}\\
434)&\BL{80{\hskip 3pt}55{\hskip 3pt}5{\hskip 3pt}70{\hskip 3pt}62{\hskip 3pt}14{\hskip 3pt}29{\hskip 3pt}0{\hskip 3pt}0\hskip 3pt\sf rs~cd:ab;\allowbreak t~b:adr;\allowbreak u~b:adrt;\allowbreak v~b:adrtu}\\
435)&\BL{80{\hskip 3pt}57{\hskip 3pt}0{\hskip 3pt}92{\hskip 3pt}48{\hskip 3pt}18{\hskip 3pt}51{\hskip 3pt}0{\hskip 3pt}0\hskip 3pt\sf rs~cd:ab;\allowbreak t~b:acs;\allowbreak uv~at:bcr}\\
}%